\documentclass[a4paper,UKenglish]{lipics-v2019}
 
\usepackage{microtype}

\title{SD-Regular Transducer Expressions for Aperiodic Transformations}
\author{Luc Dartois}{Univ Paris Est Creteil, LACL, F-94010 Creteil, France}{luc.dartois@lacl.fr}{https://orcid.org/0000-0001-9974-1922}{}
\author{Paul Gastin}{Université Paris-Saclay, ENS Paris-Saclay, CNRS, LMF, 91190, Gif-sur-Yvette, France}{paul.gastin@lsv.fr}{https://orcid.org/0000-0002-1313-7722}{}
\author{Shankara Narayanan Krishna}{IIT Bombay, India} {krishnas@cse.iitb.ac.in}{https://orcid.org/0000-0003-0925-398X}{}
\authorrunning{L. Dartois, P. Gastin, S. Krishna}
\Copyright{Luc Dartois, Paul Gastin and S. Krishna}

\ccsdesc[500]{Theory of computation~Transducers}

\keywords{transducers, aperiodic functions, regular expressions, transition monoids.}

\category{} 

\relatedversion{} 

\supplement{}

\funding{Supported by IRL ReLaX}
\hideLIPIcs
\nolinenumbers
\usepackage{hyperref}
\usepackage{amsmath}
\usepackage{amsthm}
\usepackage[utf8]{inputenc}
\usepackage{multirow}
\usepackage{amssymb}
\usepackage{stmaryrd}
\usepackage{booktabs}
\usepackage{xspace}
\usepackage{placeins}
\usepackage{enumerate}
\usepackage{verbatim}
\usepackage{mathtools}
\usepackage{url}
\usepackage{thmtools}
\usepackage{thm-restate}
\usepackage{color}
\usepackage{todonotes}
\usetikzlibrary{snakes,automata}

\usepackage[pdflatex,recompilepics=false]{gastex}

\newcommand{\paul}[1]{\textcolor{blue}{#1}}

\newcommand{\lang}[1]{\mathcal{L}(#1)}
\newcommand{\reg}[1]{\mathsf{Reg}(#1)}

\newcommand{\TrMon}{\mathsf{TrM}}
\DeclareRobustCommand{\leftleft}{\mathrel{\resizebox{1.5\width}{0.8\height}{\rotatebox[origin=c]{-90}{$\curvearrowright$}}}}
\DeclareRobustCommand{\rightright}{\mathrel{\resizebox{1.5\width}{0.8\height}{\rotatebox[origin=c]{90}{$\curvearrowleft$}}}}
\newcommand{\leftright}{\rightarrow}
\newcommand{\rightleft}{\leftarrow}
\newcommand{\LL}[2]{(\leftleft,#1,#2)}
\newcommand{\LR}[2]{(\leftright,#1,#2)}
\newcommand{\RL}[2]{(\rightleft,#1,#2)}
\newcommand{\RR}[2]{(\rightright,#1,#2)}
\newcommand{\A}{\mathcal{A}}

\newcommand{\C}[3]{\mathsf{C}_{#1,#2}(#3)}
\newcommand{\ZF}[3]{\mathsf{ZC}_{#1}^{#2}(#3)}

\newcommand{\RTE}{\textsf{RTE}\xspace}
\newcommand{\RTEs}{\textsf{RTE}s\xspace}
\newcommand{\SDRTE}{\textsf{SDRTE}\xspace}
\newcommand{\SDRTEs}{\textsf{SDRTE}s\xspace}

\newcommand\Ifthenelse[3]{{#1}\,?\,{#2}:{#3}}
\newcommand\SimpleFun[2]{{#1} \triangleright {#2}}
\newcommand{\rcdot}{\cdot_r}
\newcommand{\kstar}[3]{[#2,#3]^{#1\star}}
\newcommand{\krstar}[3]{[#2,#3]_r^{#1\star}}
\newcommand{\rstar}[1]{#1_r^{\star}}
\newcommand\sem[1]{[\![ #1 ]\!]}
\newcommand\dom[1]{\mathsf{dom}(#1)}
\newcommand{\leftend}{{\vdash}}
\newcommand{\rightend}{{\dashv}}
\newcommand{\D}{\Gamma^{*}}

\newcommand{\LQ}[2]{{#1}^{-1}{#2}}
\newcommand{\RQ}[2]{{#2}{#1}^{-1}}

\let\epsilon=\varepsilon

\begin{document}

\maketitle

\begin{abstract}
  FO transductions, aperiodic deterministic two-way transducers, as well as aperiodic
  streaming string transducers are all equivalent models for first order definable
  functions.  In this paper, we solve the long standing open problem of expressions
  capturing first order definable functions, thereby generalizing the seminal SF=AP (star
  free expressions = aperiodic languages) result of Schützenberger. 
  Our result also generalizes a lesser known characterization by Schützenberger 
  of aperiodic languages by SD-regular expressions (SD=AP).  We show that every first
  order definable function over finite words captured by an aperiodic deterministic
  two-way transducer can be described with an SD-regular transducer expression (\SDRTE).
  An \SDRTE is a regular expression where Kleene stars are used in a restricted way: they
  can appear only on aperiodic languages which are prefix codes of bounded synchronization
  delay.  \SDRTEs are constructed from simple functions using the combinators unambiguous
  sum (deterministic choice), Hadamard product, and unambiguous versions of the Cauchy
  product and the $k$-chained Kleene-star, where the star is restricted as mentioned.  In
  order to construct an \SDRTE associated with an aperiodic deterministic two-way
  transducer, (i) we concretize Schützenberger's SD=AP result, by proving that aperiodic
  languages are captured by SD-regular expressions which are unambiguous and stabilising;
  (ii) by structural induction on the unambiguous, stabilising SD-regular expressions
  describing the domain of the transducer, we construct \SDRTEs.  Finally, we also look at
  various formalisms equivalent to \SDRTEs which use the function composition, allowing to
  trade the $k$-chained star for a 1-star.
\end{abstract}

\section{Introduction}

The seminal result of Kleene, which proves the equivalence of regular expressions and
regular languages, is among the cornerstones of formal language theory.  The B\"{u}chi,
Elgot, Trakhtenbrot theorem which proved the equivalence of regular languages with MSO
definable languages, and the equivalence of regular languages with the class of languages
having a finite syntactic monoid, established the synergy between machines, logic and
algebra.  The fundamental correspondence between machines and logic at the language level
has been generalized to transformations by Engelfreit and Hoogeboom \cite{EH01}, where
regular transformations are defined by two way transducers (2DFTs) as well as by the MSO
transductions of Courcelle~\cite{Cou94}.  A generalization of Kleene's theorem to
transformations can be found in \cite{AlurFreilichRaghothaman14}, \cite{BR-DLT18} and
\cite{DGK-lics18}.

In \cite{AlurFreilichRaghothaman14}, regular transformations were described using additive
cost register automata (ACRA) over finite words.  ACRAs are a generalization of streaming
string transducers (SSTs) \cite{AC10} which make a single left to right pass over the
input and use a finite set of variables over strings from the output alphabet.  ACRAs
compute partial functions from finite words over a finite input alphabet to a monoid
$(\mathbb{D}, +, 0)$.  The main contribution of \cite{AlurFreilichRaghothaman14} was to
provide a set of combinators, analogous to the operators used in regular expressions,
which help in forming combinator expressions computing the output of the ACRA over finite
words.  The result of \cite{AlurFreilichRaghothaman14} was generalized to infinite words
in \cite{DGK-lics18}.  The proof technique in \cite{DGK-lics18} is completely different
from \cite{AlurFreilichRaghothaman14}, and, being algebraic, allows a uniform treatment of
the result for transductions over finite and infinite words.  Subsequently, an alternative
proof for the result of \cite{AlurFreilichRaghothaman14} appeared in \cite{BR-DLT18}.

The class of star-free languages form a strict subset of regular languages.  In 1965,
Schützenberger \cite{Schutzenberger_1965} proved his famous result that star-free
languages (SF) and languages having an aperiodic syntactic monoid (or aperiodic languages
AP) coincide over finite words (SF=AP).  This equivalence gives an effective
characterization of star-free languages, since one can decide if a syntactic monoid is
aperiodic.  This was followed by a result \cite{McNaughtonPapert} of McNaughton and Papert
proving the equivalence of star-free languages with counter-free automata as well as first
order logic, thereby completing the machine-logic-algebra connection once again.  The
generalization of this result to transformations has been investigated in \cite{CD15},
\cite{FKT14}, proving the equivalence of aperiodic two way transducers and FO
transductions a la Courcelle for finite and infinite words.  The counter part of regular
expressions for aperiodic languages are star-free expressions, which are obtained by using
the complementation operation instead of Kleene-star.  The generalization of this result
to transformations has been an open problem, as mentioned in
\cite{AlurFreilichRaghothaman14}, \cite{BR-DLT18} and \cite{DGK-lics18}.  One of the main
challenges in generalizing this result to transformations is the difficulty in finding an
analogue for the complementation operation on sets in the setting of transformations.

\medskip\noindent{\bf{Our Contributions}}.  The following central problem remained open
till now: Given an aperiodic 2DFT $\A$, does there exist a class of expressions over basic
functions and regular combinators such that, one can effectively compute from $\A$, an
expression $E$ in this class, and conversely, such that $\sem{\A}(w)=\sem{E}(w)$ for each
$w \in \dom{\A}$?  We solve this open problem, by providing a characterization by means of
expressions for aperiodic two way transducers.  In the following, we describe the main
steps leading to the solution of the problem.

\medskip\noindent {\bf{Concretizing Schützenberger's characterization}}.  In 1973,
Schützenberger \cite{Schutzenberger1975d} presented a characterization of aperiodic
languages in terms of rational expressions where the star operation is restricted to
prefix codes with bounded synchronization delay and no complementation is used.  This
class of languages is denoted by SD, and this result is known as SD=AP. To circumvent the
difficulty of using complementation in star-free expressions, we use this SD=AP
characterization of aperiodic languages by SD-expressions.  An SD-expression is a regular
expression where the Kleene stars are restricted to appear only on prefix codes of bounded
synchronization delay.  Our \emph{first} contribution is to concretize Schützenberger's
result to more specific SD-expressions.  We show that aperiodic languages can be captured
by \emph{unambiguous}, \emph{stabilising}, SD-expressions.  The \emph{unambiguity} of an
expression refers to the unique way in which it can be parsed, while \emph{stabilising
expressions} is a new notion introduced in this paper.  Our concretization,
(Theorem~\ref{thm:U-S-SD-expressions}) shows that, given a morphism $\varphi$ from the
free monoid $\Sigma^*$ to a finite aperiodic monoid $M$, for each $s \in M$,
$\varphi^{-1}(s)$ can be expressed by an unambiguous, stabilising SD-expression.  The two
notions of \emph{unambiguity} and \emph{stabilising} help us to capture the runs of an
aperiodic two way transducer.  These two notions will be described in detail in 
Section~\ref{sec:SDexpr}.
    
\medskip\noindent{\bf{The Combinators}}.  Our \emph{second} contribution is the definition
of SD-regular transducer expressions ($\SDRTE$).  These are built from basic constant
functions using combinators such as unambiguous sum, unambiguous Cauchy product, Hadamard
product.  In addition, we use $k$-chained Kleene star $\kstar{k}{L}{C}$ (and its reverse)
when the parsing language $L$ is restricted to be aperiodic and a prefix code with bounded
synchronisation delay.  It should be noticed that, contrary to the case of regular
transducer expressions (\RTE) which define all regular functions, the 2-chained Kleene
star $\kstar{2}{L}{C}$ does not seem sufficient to define all aperiodic functions (see
Section~\ref{sec:inexp} as well as Figure~\ref{fig:eg-2dft}), and $k$-chained Kleene stars
for arbitrary large $k$ seem necessary to capture all aperiodic functions.

The semantics of an \SDRTE $C$ is a partial function $\sem{C}\colon\Sigma^*\to\Gamma^*$
with domain denoted $\dom{C}$.  An \SDRTE of the form $\SimpleFun{L}{v}$ where
$L\subseteq\Sigma^{*}$ is an aperiodic language and $v\in\Gamma^{*}$ is such that
$\sem{\SimpleFun{L}{v}}$ is a constant function with value $v$ and domain $L$.  The
Hadamard product $C_1\odot C_2$ when applied to $w\in\dom{C_1}\cap\dom{C_2}$ produces
$\sem{C_1}(w)\cdot\sem{C_2}(w)$.  The unambiguous Cauchy product $C_1\cdot C_2$ when
applied on $w\in\Sigma^*$ produces $\sem{C_1}(u)\cdot\sem{C_2}(v)$ if $w$ can be
unambiguously decomposed as $u\cdot v$, with $u \in \dom{C_1}$ and $v \in \dom{C_2}$.  The
Kleene star $C^{*}$ is defined when $L=\dom{C}$ is an aperiodic language which is a prefix
code with bounded synchronisation delay.  Then $\dom{C^{*}}=L^{*}$, and, for
$w=u_1u_2\cdots u_n$ with $u_i\in L$, we have
$\sem{C^{*}}(w)=\sem{C}(u_1)\sem{C}(u_2)\cdots\sem{C}(u_n)$.
    	 
\begin{figure*} [t]                                                                                                                                                           
  \begin{center}                                                                                                                                                         
  \scalebox{0.8}{                                                                                                                                                       
  \begin{tikzpicture}[->,thick]                                                                                                                                          
    \node[initial, state, initial text ={}] at (0,-2) (A0) {$q_0$} ;
    \node[state,accepting] at (2.5,-2) (A) {$q_1$} ;
    \node[state] at (5,-2) (B) {$q_2$}; 
    \node[state] at (7.5,-2) (C) {$q_3$};
    \node[state] at (10,-2) (D) {$q_4$};
    \node[state] at (12.5,-2) (E) {$q_5$};
    \node[state] at (15,-2) (F) {$q_6$};
    
    \path (A0) edge node [above]{$\$/\epsilon, +1$}node{} (A);
    \path (A) edge[loop above] node [above]{$a/\epsilon, +1$}node {}(A);
    \path (A0) edge[loop above] node [above]{$\leftend/\epsilon, +1$}node {}(A0);
    \path (A) edge node [above]{$\#/\epsilon, +1$}node {}(B);
    \path (B) edge[loop above] node [above]{$a/b, +1$}node {}(B);
    \path (B) edge node [above]{$\$/\epsilon, -1$}node{} (C);
    \path (C) edge[loop above] node [above]{$a/\epsilon, -1$}node {}(C);
    \path (C) edge node [above]{$\#/\epsilon, -1$}node{} (D);
    \path (D) edge[loop above] node [above]{$a/\varepsilon, -1$}node {}(D);
    \path (D) edge node [above]{$\$/\epsilon, +1$}node{} (E);
    \path (E) edge[loop above] node [above]{$a/a, +1$}node {}(E);
    \path (E) edge node [above]{$\#/\epsilon, +1$}node {}(F);
    \path (F) edge[loop above] node [above]{$a/\epsilon, +1$}node {}(F);
    \path (F) edge[bend left] node [above]{$\$/ \epsilon,+1$}node {}(A);
  \end{tikzpicture}
  } 
  \end{center}                                                                                                                                                         
  \caption{An aperiodic 2DFT $\cal{A}$ computing the partial function
  $\sem{{\mathcal{A}}}(\$a^{m_1}\#a^{m_2} \$ a^{m_3} \# a^{m_4} \$ \cdots
  a^{m_{2k}}\$ )=b^{m_2}a^{m_1} b^{m_4}a^{m_3} \cdots b^{m_{2k}}a^{m_{2k-1}}$,
  for $k\geq 0$.  The input alphabet is $\Sigma=\{a, \#, \$\}$ while the output
  alphabet is $\Gamma=\{a,b\}$.}
  \label{fig:2dft}                                                                                                                                                        
\end{figure*}
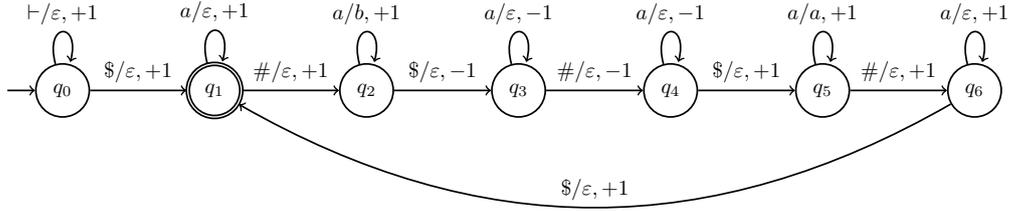 

As an example, consider the $\SDRTE$s $C=C_1\cdot C_2$, $C'=C'_1\cdot C'_2$ and $D=C\odot
C'$ with $C_1= \SimpleFun{(a^*\#)}{\epsilon}$,
$C_2=(\SimpleFun{a}{b})^{*}\cdot(\SimpleFun{\$}{\varepsilon})$, and $C'_1=
(\SimpleFun{a}{a})^{*}\cdot(\SimpleFun{\#}{\varepsilon})$,
$C'_2=\SimpleFun{(a^*\$)}{\epsilon}$.  Then $\dom{C_1}=a^*\#=\dom{C'_1}$,
$\dom{C_2}=a^*\$=\dom{C'_2}$, and $\dom{C}=a^*\#a^*\$=\dom{C'}=\dom{D}$.  Further,
$\sem{C_1}(a^m\#)=\epsilon$, $\sem{C_2}(a^n\$)=b^n$, $\sem{C'_1}(a^m\#)=a^m$, and
$\sem{C'_2}(a^*\$)=\epsilon$.  Also, $\sem{D}(a^m\#a^n\$)=b^na^m$.  Notice that $\dom{D}$
is a prefix code with synchronisation delay 1.  Hence, we can define the \SDRTE $D^{*}$
which has domain the aperiodic language $\dom{D^{*}}=(a^*\#a^*\$)^*$, and
$\sem{D^{*}}(a^2\#a^3\$a^4\#a^5\$)=b^3a^2b^5a^4$.  The $\SDRTE$
$D'=(\SimpleFun{\$}{\epsilon})\cdot D^{*}$ corresponds to the aperiodic 2DFT $\A$ in
Figure~\ref{fig:2dft}: $\sem{\A}=\sem{D'}$.
 
\medskip\noindent{\bf{$\SDRTE \leftrightarrow$ Aperiodic 2DFT}}.  Our \emph{third} and
main contribution solves the open problem by proving the effective equivalence between
aperiodic two way transducers and $\SDRTEs$ over finite words:

\begin{theorem}  \label{thm:intro}	
  (1) Given an $\SDRTE$, we can effectively construct an equivalent aperiodic 2DFT.
  (2) Given an aperiodic 2DFT, we can effectively construct an equivalent $\SDRTE$.
\end{theorem}

The proof of (1) is by structural induction on the $\SDRTE$.  
All cases except the $k$-chained Kleene star are reasonably 
simple, and it is easy to see how to construct the equivalent 2DFT.
The case of the $k$-chained Kleene star is more involved. 
We write $\kstar{k}{L}{C}$ as the composition of 3 aperiodic functions $f_1, f_2, f_3$, where, 
(i) $f_1$ takes as input $u_1u_2 \cdots u_n \in L^*$ with $u_i \in L$ and produces as
output $\# u_1 \# u_2 \# \cdots \#u_n\#$, 
(ii) $f_2$ takes $\# v_1 \# v_2 \# \cdots \#v_m\#$ with $v_i\in\Sigma^{*}$ as input, and
produces $\#v_1 \cdots v_k \# v_2 \cdots v_{k+1}\# \cdots \#v_{m-k+1} \cdots v_m\#$ as
output,
(iii) finally, $f_3$ takes $\# w_1 \# w_2 \# \cdots \#w_\ell\#$ as input with
$w_i\in\Sigma^{*}$ and produces as output $f(w_1) f(w_2) \cdots f(w_\ell)$.  We produce
aperiodic 2DFTs for $f_1, f_2, f_3$, and compose them, obtaining the required aperiodic
2DFT.

The construction of $\SDRTE$ from an aperiodic 2DFT $\A$ is much more involved, and is
based on the transition monoid $\TrMon$ of the 2DFT $\A$.  The translation of $\A$ to
$\SDRTE$ is guided by an unambiguous, stabilising, SD-regular expression induced by
$\TrMon$.  These expressions are obtained thanks to Theorem \ref{thm:U-S-SD-expressions}
applied to the canonical morphism $\varphi\colon \Sigma \rightarrow \TrMon$ where the
transition monoid $\TrMon$ of $\A$ is aperiodic.  This construction is illustrated in
detail via Examples \ref{eg2dftstable}, \ref{eg:product}, \ref{eg:kstar} and
\ref{eg:last}.

\medskip\noindent{\bf{Related Work}}.  A natural operation on functions is that of
composition.  The composition operation can be used in place of the chained-sum operator
of \paul{\cite{AlurFreilichRaghothaman14}}, and also in place of the unambiguous 2-chained
iteration of \cite{DGK-lics18}, preserving expressiveness.  In yet another recent paper,
\cite{paul-inv} proposes simple functions like copy, duplicate and reverse along with
function composition to capture regular word transductions.

A closely related paper to our work is \cite{BDK-lics18}, where first-order and regular
list functions were introduced.  Using the basic functions reverse, append, co-append,
map, block on lists, and combining them with the function combinators of disjoint union,
map, pairing and composition, these were shown to be equivalent (after a suitable
encoding) to FO transductions a la Courcelle (extendible to MSO transductions by adding to
the basic functions, the prefix multiplication operation on groups).  
\cite{BDK-lics18} provides an equivalent characterization (modulo an encoding) for FO
transductions with basic list functions and combinators. 

Contrary to \cite{BDK-lics18} where expressions crucially rely on function composition, we
focus on concatenation and iteration as first class combinators, in the spirit of Kleene's
theorem and of Schützenberger's characterisation AP=SD. We are able to characterise 2DFTs
with such SD-regular expressions without using composition.  Hence, our result is fully
independent and complementary to the work in \cite{BDK-lics18}: both formalisms, \SDRTEs
and list functions are natural choices for describing first order transductions.
  Our basic functions and combinators are inspired from the back and forth traversal of a
  two way automaton, and the restrictions on the usage of the Kleene star comes from the
  unambiguous, stabilising nature of the expressions capturing the aperiodic domain
  of the 2DFT.
  We also study in Section~\ref{sec:composition} how composition may be used to simplify
  our \SDRTEs (Theorem \ref{thm-withComp}).  With composition, $k$-chained Kleene star
  ($k>1$) is no more necessary, resulting in an equivalent formalism, namely, \SDRTE where
  we only use $1$-star.  Yet another equivalent formalism is obtained by restricting
  \SDRTE to simple functions, unambiguous sum, Cauchy product and 1-star, but adding the
  functions duplicate and reverse along with composition.

\medskip\noindent{\bf{Structure of the paper}.}
In Section~\ref{sec:Prelim}, we introduce preliminary notions used throughout
the paper.  In Section~\ref{sec:SDexpr} we give a procedure to construct
complement-free expressions for aperiodic languages that suits our approach.
This is a generic result on languages, independent of two-way transducers.
Section~\ref{sec:comb-expr} presents the  combinators and the chain-star
operators for our characterization. 
The main theorem and technical proofs, which is constructing SD-regular transducer expressions from a two-way aperiodic
transducer, are in Section~\ref{sec:Equiv}.

\begin{gpicture}[name=auto:even,ignore]
  \gasset{Nw=6,Nh=6}
  \node[Nmarks=ir](1)(0,0){1}
  \node(2)(20,0){2}
  \drawedge[curvedepth=2](1,2){$a,b$}
  \drawedge[curvedepth=2](2,1){$a,b$}
\end{gpicture}
\begin{gpicture}[name=auto:U2,ignore]
  \gasset{Nw=6,Nh=6,loopdiam=5}
  \node(1)(0,0){1}
  \node(2)(20,0){2}
  \drawloop[loopangle=180](1){$a,c$}
  \drawloop[loopangle=0](2){$b,c$}
  \drawedge[curvedepth=2](1,2){$a$}
  \drawedge[curvedepth=2](2,1){$b$}
\end{gpicture}
\section{Preliminaries}\label{sec:Prelim}

\subsection{Rational languages and monoids}\label{sec:RatPrelim}
We call a finite set $\Sigma$ an \emph{alphabet} and its elements \emph{letters}.
A finite sequence of letters of $\Sigma$ is called a \emph{word}, and a set of words is a \emph{language}.
The empty word is denoted $\epsilon$, and we denote by $\Sigma^*$ the set of all words over the alphabet $\Sigma$. More generally, given any language $L\subseteq \Sigma^*$, we write $L^*$ for the \emph{Kleene star} of $L$, i.e., the set of words which can be written as a (possibly empty) sequence of words of $L$.
Given a word $u$, we write $|u|$ for the \emph{length} of $u$, i.e., its number of letters, and we denote by $u_i$ its $i^{th}$ letter.

A \emph{monoid} $M$ is a set equipped with a binary associative law, usually denoted $\cdot$ or omitted when clear from context, and a neutral element $1_M$ for this law, meaning that for any $s\in M$, $1_M\cdot s=s\cdot 1_M= s$.
The set of words $\Sigma^*$ can be seen as  the free monoid generated by $\Sigma$ using the concatenation of words as binary law.
Given a \emph{morphism} $\varphi: \Sigma^* \to M$, i.e., a function between 
monoids that satisfies $\varphi(\varepsilon)=1_M$ and
$\varphi(xy)=\varphi(x)\varphi(y)$ for any $x,y$, we say that $\varphi$
recognizes a language $L\subseteq \Sigma^*$ if $M$ is finite and 
$L=\varphi^{-1}(P)$ for some $P\subseteq M$.
A monoid is called \emph{aperiodic} if there exists an integer $n$ such that for any element $s$ of $M$, $s^n=s^{n+1}$. 

\begin{example}\label{ex:Un}
  We define the monoids $\widetilde{U}_n$, for $n\geq 0$, as the set of
  elements $\{1,s_1,\ldots,s_n\}$, with $1$ being the neutral element, and for
  any $1\leq i,j\leq n$, $s_i\cdot s_j= s_i$.  Clearly, $\widetilde{U}_n$ is
  aperiodic, actually idempotent, as $s_i\cdot s_i =s_i$ for any $1\leq i\leq
  n$.  For instance, the monoid $\widetilde{U}_2$ is the transition monoid 
  (defined below)
  of  the automaton below with $\varphi(a)=s_1$, $\varphi(b)=s_2$ and 
  $\varphi(c)=1$.

  \medskip\noindent\hfil\gusepicture{auto:U2}
\end{example}

\emph{Rational} languages are languages that can be described by rational 
expressions, i.e., sets of words constructed from finite sets
using the operations of concatenation, union and Kleene star. 
It is well-known that rational languages are equivalent to 
\emph{regular} languages, i.e., languages accepted by finite automata, and to 
languages recognized by finite monoids (and Monadic Second-order logic~\cite{Buchi60}).
\emph{Star-free} rational expressions are built from finite sets using
the operations of concatenation, union and complement (instead of Kleene star).
They have the same expressive power as finite aperiodic monoids~\cite{Schutzenberger_1965}
(as well as counter-free automata and first-order logic~\cite{McNaughtonPapert}).

\subsection{Two-way transducers}\label{sec:Twoway}

\begin{definition}[Two-way transducer]
  A \emph{(deterministic) two-way transducer} (2DFT) is a tuple
  $\A=(Q,\Sigma,\Gamma,\delta,\gamma, q_0,F)$ defined as follows:
  \begin{itemize} \itemsep0cm
    \item $Q$ is a finite set of \emph{states}.
    \item $\Sigma$ and $\Gamma$ are the finite \emph{input} and \emph{output alphabets}.
    \item $\delta: Q\times (\Sigma\uplus\{\leftend,\rightend\}) \to Q\times\{-1,+1\}$
    is the partial \emph{transition function}.  Contrary to one-way machines, the
    transition function also outputs an integer, indicating the move of the
    reading head.
    The alphabet is enriched with two new symbols $\leftend$ and $\rightend$, which are
    endmarkers that are added respectively at the beginning and at the end of the input
    word, such that for all $q\in Q$, we have ${\delta(q,\leftend)\in Q\times\{+1\}}$ (if
    defined), $\delta(q,\rightend)\in Q\times\{-1\}$ (if defined) and
    $\delta(q,\rightend)$ is undefined for $q\in F$.
    
    \item $\gamma: Q\times (\Sigma\uplus\{\leftend,\rightend\}) \to \Gamma^*$ is
    the partial \emph{production function} with same domain as $\delta$.
    
    \item $q_0\in Q$ is the \emph{initial state}.
    \item $F \subseteq Q$ is the set of final states.
  \end{itemize}
\end{definition}

A \emph{configuration} $c$ of $\A$ over an input word $w=w_1 \cdots w_{|w|}$ is
simply a pair $(p,i)$ where $p\in Q$ is the current state and $0\leq i\leq
|w|+1$ is the position of the head on the input tape containing $\leftend w
\rightend$.  Two configurations $c=(p,i)$ and $c'=(q,j)$ are \emph{successive}
if we have $\delta(p,w_i)=(q,d)$ and $i+d=j$, with $w_0=\leftend$ and
$w_{|w|+1}=\rightend$.  In this case, they produce an output $v=\gamma(p,w_i)$.
Abusing notations we will sometime write $\gamma(c)$ when the input word $w$ is
clear.  A run $\rho$ is a sequence of successive configurations $c_0\cdots c_n$.
The run $\rho$ is \emph{initial} if $c_0=(q_0,0)$ and is \emph{final} if
$c_n=(q,|w|+1)$ for some $q\in F$.  It is \emph{accepting} if it is both initial
and final.
  
The \emph{output} of a run $\rho=c_0\cdots c_n$ is the concatenation of the
output of the configurations, and will be denoted
$\sem{\rho}=\gamma(c_0)\cdots \gamma(c_{n-1})$.  Given a deterministic
two-way transducer $\A$ and an input word $w$, there is at most one accepting
run
of $\A$ over $\leftend w\rightend$, which we will denote
$\rho(w)$.
The output of $\A$ over $w$ is then $\sem{\A}(w)=\sem{\rho(w)}$.
The \emph{domain} of $\A$ is the set $\dom{\A}$ of words $w$ such that there exists an accepting run of $\A$ over $w$.
Finally, the semantics of $\A$ is the partial function 
$\sem{\A}\colon\Sigma^{*}\to\Gamma^{*}$ defined on $\dom{\A}$ by $w\mapsto\sem{\A}(w)$.

Let $\rho=(p_0,i_0)\cdots(p_n,i_n)$ be a run over a nonempty word
$w\in\Sigma^{+}$ such that $1\leq i_j\leq|w|$ for all $0\leq j< n$.  It is a
\emph{left-right} run if $i_0=1$ and $i_n=|w|+1$.  If this is the case, we say
that $\rho$ is a $\LR{p_0}{p_n}$-run.  Similarly, it is a \emph{left-left}
$\LL{p_0}{p_n}$-run if $i_0=1$ and $i_n=0$.  It is a \emph{right-left}
$\RL{p_0}{p_n}$-run if $i_0=|w|$ and $i_n=0$ and it is a \emph{right-right}
$\RR{p_0}{p_n}$-run if $i_0=|w|$ and $i_n=|w|+1$.  Notice that if $|w|=1$, then
left-right runs and right-right runs coincide, also right-left runs and
left-left runs coincide.
\begin{remark}
  Given our semantics of two-way transducers, a run associates states to each position,
  whereas the classical semantics of one-way automata keeps the states between two positions.
  Then, if we consider a word $w=uv$ and a left-left run $\LL{p}{q}$ on $v$, 
  we begin on the first position of $v$ in state $p$, and   the state $q$
  is reached at the end of the run on the last position of $u$.  This allows for
  easy sequential composition of partial runs when concatenating non empty words, as the
  end of a partial run is the start of the next one.

  However, in order to keep our figures as readable as possible, we will represent these
  states between words.  A state $q$ between two words $u$ and $v$ is to be placed on the
  first position of $v$ if it is the start of a run going to the right, and on the last
  position of $u$ otherwise. For instance, in Figure~\ref{fig:LR-run}, state $q_1$ is on 
  the first position of $u_{i+1}$ and state $q_3$ is on the last position of $u_i$.
\end{remark}
\paragraph*{Transition monoid of a two-way automaton}\label{sec:TrMon}
Let $\A$ be a deterministic two-way automaton (2DFA) with set of states $Q$.
When computing the transition monoid of a two-way automaton, we are
interested in the behaviour of the partial runs, i.e., how these partial runs
can be concatenated.  Thus we abstract a given $(d,p,q)$-run $\rho$ over a word
$w$ to a \emph{step}
$(d,p,q)\in\{\leftright,\leftleft,\rightright,\rightleft\}\times Q^{2}$ and we say
that $w$ realises the step $(d,p,q)$.  The transition monoid $\TrMon$ of $\A$ is
a subset of the powerset of steps:
$\TrMon\subseteq\mathcal{P}(\{\leftright,\leftleft,\rightright,\rightleft\}\times Q^{2})$.
The canonical surjective morphism
$\varphi\colon(\Sigma\uplus\{\leftend,\rightend\})^{*}\to\TrMon=\varphi((\Sigma\uplus\{\leftend,\rightend\})^{*})$ is
defined for a word $w\in(\Sigma\uplus\{\leftend,\rightend\})^{*}$ as the set of steps realised by $w$, i.e.,
$\varphi(w)=\{(d,p,q)\mid \text{there is a } (d,p,q)\text{-run on }w\}
\subseteq\{\leftright,\leftleft,\rightright,\rightleft\}\times Q^{2}$.  
As an example, in Figure~\ref{fig:2dft}, we have 
$$\varphi(a\#)=\{(\leftright, q_1, q_2), (\rightright, q_1, q_2), 
(\leftleft, q_3, q_3), (\rightleft,q_3,q_4),
(\leftleft, q_4, q_4), 
(\leftright, q_5, q_6), (\rightright, q_5, q_6)\}\,.$$
The unit of $\TrMon$ is $\mathbf{1}=\{\LR{p}{p},\RL{p}{p}\mid p\in Q\}$ 
and $\varphi(\varepsilon)=\mathbf{1}$.

A 2DFA is \emph{aperiodic} if its transition monoid $\TrMon$ is aperiodic.
Also, a 2DFT is aperiodic if its underlying input 2DFA is aperiodic.

When talking about a given step $(d,p,q)$ belonging to an element of $\TrMon$,
we will sometimes forget $p$ and $q$ and talk about a $d$-step, for
$d\in\{\leftleft,\rightright,\leftright,\rightleft\}$ if the states $p, q$ are
clear from the context, or are immaterial for the discussion. In this case we
also refer to a step $(d, p, q)$ as a $d$-step having $p$ as the starting state
and $q$ as the final state.

\section{Complement-free expressions for aperiodic languages}\label{sec:SDexpr}

As the aim of the paper is to obtain rational expressions corresponding to
transformations computed by aperiodic two-way transducers, we cannot rely on
extending the classical (SF=AP) star-free characterization of aperiodic languages, since
the complement of a function is not a function.  We solve this problem by
considering the SD=AP characterization of aperiodic languages, namely prefix codes
with bounded synchronisation delay, introduced by
Schützenberger~\cite{Schutzenberger1975d}.

A language $L$ is called a \emph{code} if for any word $u\in L^*$, there is a
unique decomposition $u=v_1\cdots v_n$ such that $v_i\in L$ for $1\leq i\leq n$.
For example, the language $W=\{a, ab, ba, bba\}$ is not a code: the words $abba,
aba \in W^*$ have decompositions $a\cdot bba=ab\cdot ba$ and $a\cdot ba=ab\cdot
a$ respectively.  A \emph{prefix code} is a language $L$ such that for any pair of words
$u,v$, if $u,uv\in L$, then $v=\epsilon$.  $W$ is not a prefix code, while
$W_1=W \setminus \{ab\}$ and $W_2=W \setminus \{a\}$ are prefix codes.
Prefix codes play a particular role in the sense that the
unique decomposition can be obtained on the fly while reading the word from left
to right.

\begin{definition}
  Let $d$ be a positive integer.
  A \emph{prefix} code $C$ over an alphabet $\Sigma$ has a synchronisation delay
  $d$ (denoted $d$-SD) if for all $u,v,w\in \Sigma^*$, $uvw\in C^*$ and
  $v\in C^d$ implies $uv\in C^*$ (hence also $w\in C^{*}$). An SD prefix code 
  is a prefix code with a bounded synchronisation delay. 
\end{definition}

As an example, consider the prefix code $C=\{aa, ba\}$ and the word
$ba(aa)^{d}\in C^{*}$.  We have $ba(aa)^{d}=uvw$ with $u=b$, $v=(aa)^{d}\in
C^{d}$ and $w=a$.  Since $uv\notin C^{*}$, the prefix code $C$ is not of bounded
synchronisation delay.  Likewise, $C=\{aa\}$ is also not of bounded
synchronisation delay.
On the other hand, the prefix code $C=\{ba\}$ is 1-SD.

The syntax of regular expressions over the alphabet $\Sigma$ is given by the grammar
$$
E ::= \emptyset \mid \varepsilon \mid a \mid E\cup E \mid E\cdot E \mid E^{*}
$$
where $a\in\Sigma$. We say that an expression is $\varepsilon$-free (resp.\ 
$\emptyset$-free) if it does not use $\varepsilon$ (resp.\ $\emptyset$) as 
subexpressions.  The semantics of a regular expression $E$ is a regular language
over $\Sigma^{*}$ denoted $\lang{E}$.

An SD-regular expression is a regular expression where Kleene-stars are
restricted to SD prefix codes: If $E^{*}$ is a sub-expression then $\lang{E}$ is
a prefix code with bounded synchronization delay.  Thus, the regular expression
$(ba)^*$ is a SD-regular expression while $(aa)^*$ is not.

The relevance of SD-regular expressions comes from the fact that they are a
complement-free characterization of aperiodic languages.

\begin{theorem}\cite{Schutzenberger1975d}
  A language $L$ is recognized by an aperiodic monoid if, and only if, there
  exists an SD-regular expression $E$ such that $L=\lang{E}$.
\end{theorem}

Theorem~\ref{thm:U-S-SD-expressions} concretizes this result, and extends it to
get more specific expressions which are (i) \emph{unambiguous}, a property required for the
regular combinators expressing functions over words, and (ii) \emph{stabilising}, which
is a new notion introduced below that suits our need for characterizing runs of
aperiodic two-way transducers. Our proof technique follows the local divisor technique, which was notably used
by Diekert and Kufleitner to lift the result of Schützenberger to infinite
words~\cite{DiekertK-CSR12,DiekertK-ToCS2015}.

A regular expression $E$ is unambiguous, if it satisfies the following:
\begin{enumerate}
  \item  for each subexpression $E_1\cup E_2$ we have 
  $\lang{E_1}\cap\lang{E_2}=\emptyset$,

  \item  for each subexpression $E_1\cdot E_2$, each word
  $w\in\lang{E_1\cdot E_2}$ has a \emph{unique} factorisation $w=uv$ with 
  $u\in\lang{E_1}$ and $v\in\lang{E_2}$,

  \item for each subexpression $E_1^{*}$, the language $\lang{E_1}$ is a
  \emph{code}, i.e., each word $w\in\lang{E_1^{*}}$ has a \emph{unique}
  factorisation $w=v_1\cdots v_n$ with $v_i\in\lang{E_1}$ for $1\leq i\leq n$.
\end{enumerate}

\begin{definition}
  Given an aperiodic monoid $M$ and $X\subseteq M$, we say that $X$
  is $n$-\emph{stabilising} if $xy=x$ for all $x\in X^{n}$ and $y\in X$. 
  We say that $X$ is stabilising if it is $n$-stabilising for some $n\geq1$.
\end{definition}

\noindent{\bf{Remark}}.  Stabilisation generalizes aperiodicity in some sense.  For
aperiodicity, we require $x^n=x^{n+1}$ for each element $x\in M$ and some $n \in
\mathbb{N}$, i.e., all \emph{singleton} subsets of $M$ should be stabilising.

\begin{example}\label{ex:deux}
  Continuing Example~\ref{ex:Un}, any subset
  $X\subseteq\{s_1,\ldots,s_n\}\subseteq\widetilde{U}_n$ is 1-stabilising.
  
  As another example, consider the aperiodic 2DFT $\A$ in Figure \ref{fig:2dft}, and consider
  its transition monoid $\TrMon$.  Clearly, $\TrMon$ is an aperiodic monoid.  Let
  $\varphi$ be the morphism from $(\Sigma \uplus \{\leftend, \rightend\}) ^*$ to $\TrMon$.
  Consider the subset $Z=\{Y, Y^{2}\}$ of $\TrMon$ where $Y=\varphi(a\#a\$)$:
  \begin{align*}
    Y &= \{ (\leftleft,q_1,q_4), (\leftleft,q_3,q_3), 
    (\leftleft,q_4,q_4), (\leftright,q_5,q_1), (\rightright,q_0,q_1), 
    (\rightleft,q_2,q_4), (\rightright,q_4,q_5), (\rightright,q_6,q_1)\}
    \\
    Y^{2} &= \{ (\leftleft,q_1,q_4), (\leftleft,q_3,q_3), 
    (\leftleft,q_4,q_4), (\leftright,q_5,q_1), (\rightright,q_0,q_1), 
    (\rightright,q_2,q_1), (\rightright,q_4,q_5), (\rightright,q_6,q_1)\} \,.
  \end{align*}
  It can be seen that $Y^{3}=Y^{2}$, hence $Z$ is 2-stabilising.

\end{example}

Let $\varphi\colon\Sigma^{*}\to M$ be a morphism.  We say that a regular
expression $E$ is $\varphi$-\emph{stabilising} (or simply stabilising when
$\varphi$ is clear from the context) if for each subexpression $F^{*}$ of $E$,
the set $\varphi(\lang{F})$ is stabilising. 

Continuing Example~\ref{ex:deux}, we can easily see that $\varphi(a)$ is idempotent and we
get $\varphi(a^{+}\#a^{+}\$)=\{Y\}$.  Since $Y^{3}=Y^{2}$, we deduce that
$(aa^*\#aa^*\$)^{*}$ is a stabilising expression.  Notice also that, by definition,
expressions without a Kleene-star are stabilising vacuously.

\begin{example}\label{eg:stabilising}
  As a more non trivial example to illustrate stabilising expressions, consider the 2DFT
  $\A$ in Figure~\ref{fig:eg-2dft}, whose domain is the language $b(a^{*}b)^{\geq 2}a^*$.
  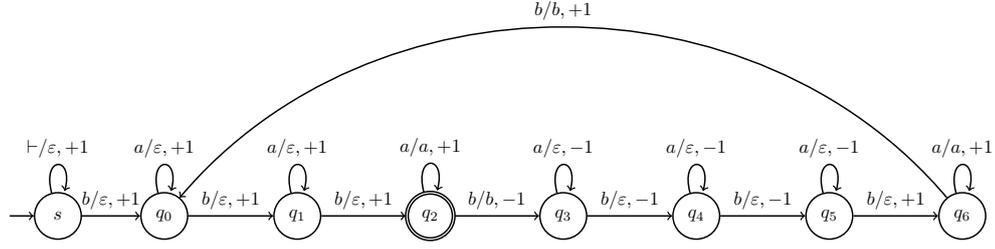
\begin{figure} [t]
  \begin{center}                                                                                                                                                         
    \scalebox{0.7}{                                                                                                                                                       
    \begin{tikzpicture}[->,thick]                                     
      \node[initial, state, initial text ={}] at (-2,-2) (A00) {$s$} ;
      \node[state] at (0,-2) (A0) {$q_0$} ;
      \node[state] at (2.5,-2) (A) {$q_1$} ;
      \node[state,accepting] at (5,-2) (B) {$q_2$}; 
      \node[state] at (7.5,-2) (C) {$q_3$};
      \node[state] at (10,-2) (D) {$q_4$};
      \node[state] at (12.5,-2) (E) {$q_5$};
      \node[state] at (15,-2) (F) {$q_6$};
      \path (A00) edge[loop above] node [above]{$\leftend/\epsilon, +1$}node {}(A00);
      \path (A00) edge node [above]{$b/\epsilon, +1$}node{} (A0);
      \path (A0) edge[loop above] node [above]{$a/\epsilon, +1$}node {}(A0);
      \path (A0) edge node [above]{$b/\epsilon, +1$}node{} (A);
      \path (A) edge[loop above] node [above]{$a/\epsilon, +1$}node {}(A);
      \path (A) edge node [above]{$b/\epsilon, +1$}node {}(B);
      \path (B) edge[loop above] node [above]{$a/a, +1$}node {}(B);
      \path (B) edge node [above]{$b/b, -1$}node{} (C);
      \path (C) edge[loop above] node [above]{$a/\epsilon, -1$}node {}(C);
      \path (C) edge node [above]{$b/\epsilon, -1$}node{} (D);
      \path (D) edge[loop above] node [above]{$a/\varepsilon, -1$}node {}(D);
      \path (D) edge node [above]{$b/\epsilon, -1$}node{} (E);
      \path (E) edge[loop above] node [above]{$a/\epsilon, -1$}node {}(E);
      \path (E) edge node [above]{$b/\epsilon, +1$}node {}(F);
      \path (F) edge[loop above] node [above]{$a/a, +1$}node {}(F);
      \path (F) edge[bend right=50] node [above]{$b/b,+1$}node {}(A0);
    \end{tikzpicture}
    } 
  \end{center}                                                                                                                                                         
  \caption{For $u_i\in a^*b$, an aperiodic 2DFT $\cal{A}$ computing the partial
  function $\sem{{\mathcal{A}}}(bu_1u_2\cdots u_na^k)=u_3u_1u_4u_2 \cdots u_nu_{n-2}a^k$
  if $n \geq 3$, and $a^k$  if $n=2$.  The domain is 
  $b(a^{*}b)^{\geq2}a^{*}$. }
  \label{fig:eg-2dft}                                                                                                                                                        
\end{figure} 
 Consider $b(a^*b)^{\geq 3} \subseteq \dom{\A}$.      
 Note that $a^*b$ is a prefix code with synchronisation delay 1.  Let $X=\varphi(a^*b)$, where 
 $\varphi$ is the morphism from $(\Sigma \uplus \{\leftend, \rightend\})
 ^*$ to $\TrMon$.   We will see that $X$ stabilises.  
 \begin{itemize}
 \item First, we have $X=\{Y_1,Y_2\}$ where
 
 $Y_1=\varphi(b)=\{ \\ (\leftright,s,q_0), 
 (\leftright,q_0,q_1), (\leftright,q_1,q_2), (\leftleft,q_2,q_3), 
 (\leftleft,q_3,q_4), (\leftleft,q_4,q_5), (\leftright,q_5,q_6), (\leftright,q_6,q_0),
 \\
 (\rightright,s,q_0), (\rightright,q_0,q_1), (\rightright,q_1,q_2), (\rightleft,q_2,q_3), 
 (\rightleft,q_3,q_4), (\rightleft,q_4,q_5), (\rightright,q_5,q_6), (\rightright,q_6,q_0) 
 \}$ 

 $Y_2=\varphi(a^+b)=\varphi(ab)=\{ \\
 (\leftright,q_0,q_1), (\leftright,q_1,q_2), (\leftleft,q_2,q_3),
 (\leftleft,q_3,q_3), (\leftleft,q_4,q_4), (\leftleft,q_5,q_5), (\leftright,q_6,q_0), 
 \\
 (\rightright,s,q_0), (\rightright,q_0,q_1), (\rightright,q_1,q_2), (\rightleft,q_2,q_3), 
 (\rightleft,q_3,q_4), (\rightleft,q_4,q_5), (\rightright,q_5,q_6), (\rightright,q_6,q_0) 
 \}$
 
 \item  
 Next, we can check that $X^{2}=\{Y_3,Y_4\}$ where
 
 $Y_3=Y_1Y_1=Y_1Y_2=\{ \\ (\leftright,s,q_1), 
 (\leftright,q_0,q_2), (\leftleft,q_1,q_4), (\leftleft,q_2,q_3), 
 (\leftleft,q_3,q_4), (\leftleft, q_4, q_5), (\leftright,q_5,q_0), (\leftright,q_6,q_1),
 \\
 (\rightright,s,q_0), (\rightright,q_0,q_1), (\rightright,q_1,q_2), (\rightleft,q_2,q_4), 
 (\rightleft,q_3,q_5), (\rightright,q_4,q_0), (\rightright,q_5,q_6), (\rightright,q_6,q_0) 
 \}$ 
 
 $Y_4=Y_2Y_1=Y_2Y_2=\{ \\
 (\leftright,q_0,q_2), (\leftleft,q_1,q_4), (\leftleft, q_2, q_3), 
 (\leftleft,q_3,q_3), (\leftleft, q_4, q_4), (\leftleft,q_5,q_5), (\leftright,q_6,q_1),
 \\
 (\rightright,s,q_0), (\rightright,q_0,q_1), (\rightright,q_1,q_2), (\rightleft, q_2, q_4), 
 (\rightleft,q_3,q_5), (\rightright,q_4,q_0), (\rightright,q_5,q_6), (\rightright,q_6,q_0) 
 \}$ 
 
 \item Then, we have $X^{4}=\{Z_1,Z_2\}$ where
 
 $Z_1=Y_3Y_3=Y_3Y_4=\{ \\ (\leftright,s,q_2), 
 (\leftleft,q_0,q_5), (\leftleft,q_1,q_4), (\leftleft, q_2, q_3), 
 (\leftleft,q_3,q_4), (\leftleft, q_4, q_5), (\leftright,q_5,q_2), (\leftright,q_6,q_2),
 \\
 (\rightright,s,q_0), (\rightright,q_0,q_1), (\rightright,q_1,q_2), (\rightright,q_2,q_2), 
 (\rightright,q_3,q_1), (\rightright,q_4,q_0), (\rightright,q_5,q_6), (\rightright,q_6,q_0) 
 \}$ 
 
 $Z_2=Y_4Y_3=Y_4Y_4=\{ \\
 (\leftleft,q_0,q_5), (\leftleft,q_1,q_4), (\leftleft,q_2,q_3), 
 (\leftleft,q_3,q_3), (\leftleft,q_4,q_4), (\leftleft,q_5,q_5), (\leftright,q_6,q_2),
 \\
 (\rightright,s,q_0), (\rightright,q_0,q_1), (\rightright,q_1,q_2), (\rightright,q_2,q_2), 
 (\rightright,q_3,q_1), (\rightright,q_4,q_0), (\rightright,q_5,q_6), (\rightright,q_6,q_0) 
 \}$ 
 
 \item Finally, we can easily check that $Z_1Y_1=Z_1Y_2=Z_1$ and $Z_2Y_1=Z_2Y_2=Z_2$. 
 Therefore, $X$ is 4-stabilising. Moreover, 
 ${b(a^{*}b)^{\geq3}}\subseteq\varphi^{-1}(Z_1)$ and 
 ${a^{+}b(a^{*}b)^{\geq3}}\subseteq\varphi^{-1}(Z_2)$.
\end{itemize}
\end{example}

Given a morphism $\varphi$ from $\Sigma^{*}$ to some aperiodic monoid $M$, our goal is to
build, for each language $\varphi^{-1}(s)$ with $s\in M$, an SD-regular expression which
is both \emph{unambiguous} and \emph{stabilising}.  The proof is by induction on the
monoid $M$ via the local divisor technique, similar to Diekert and Kufleitner
\cite{DiekertK-CSR12,DiekertK-ToCS2015,Diekert_2016}, and to Perrin and Pin \cite[Chapter VIII, Section
6.1]{Perrin-Pin-Infinite-words}, with the objective to get stronger forms of SD-regular
expressions.

\begin{theorem}\label{thm:U-S-SD-expressions}
  Given a morphism $\varphi$ from the free monoid $\Sigma^{*}$ to a finite
  aperiodic monoid $M$, for each $s\in M$ there exists an
  unambiguous, stabilising, SD-regular expression $E_s$ such
  that $\lang{E_s}=\varphi^{-1}(s)$.
\end{theorem}

The proof of this theorem makes crucial use of \emph{marked substitutions}
(see~\cite{Perrin-Pin-Infinite-words}) that we define and study in the next
section.

\subsection{Marked substitutions}
Let $A, B$ be finite alphabets. A map $\alpha\colon A\to \mathcal{P}(B^{*})$ is called a \emph{marked
substitution} if it satisfies the following two properties:
\begin{itemize}
  \item There exists a partition $B=B_1\uplus B_2$ such that for all $a$ in $A$, $\alpha(a)\subseteq B_1^*B_2$,

  \item For all $a_1$ and $a_2$ in $A$, $a_1\neq a_2$ implies $\alpha(a_1)\cap \alpha(a_2)=\emptyset$.
\end{itemize}
A marked substitution $\alpha\colon A\to \mathcal{P}(B^{*})$ can be naturally
extended to words in $A^*$ using concatenation of languages, i.e., to a
morphism from the free monoid $A^{*}$ to 
$(\mathcal{P}(B^{*}),\cdot,\{\varepsilon\})$. It is then further lifted to 
languages $L\subseteq A^{*}$ by union: $\alpha(L)=\bigcup_{w\in L}\alpha(w)$.

\begin{lemma}[\cite{Perrin-Pin-Infinite-words} Chapter VIII, Proposition 6.2]\label{lem:MarkSubst}
  Let $\alpha\colon A\to \mathcal{P}(B^{*})$ be a marked substitution, and
  $X\subseteq A^{+}$ be a prefix code with synchronisation delay $d$.  
  Then $Y=\alpha(X)\subseteq(B_1^{*}B_2)^{+}$ is a prefix code with
  synchronisation delay $d+1$.
\end{lemma}

\begin{proof}
  First, since $B_1$ and $B_2$ are disjoint, $B_1^{*}B_2\subseteq B^{*}$ is a
  prefix code.  Hence, given a word $w\in\alpha(A^{*})\subseteq(B_1^*B_2)^*$,
  there exists a unique decomposition $w=w_1\cdots w_n$ such that $w_i\in
  B_1^*B_2$ for $1\leq i\leq n$.  Now since images of different letters from $A$
  are disjoint, there exists at most one $a_i$ such that $w_i\in\alpha(a_i)$.
  We deduce that there is exactly one word $w'\in A^{*}$ such that
  $\alpha(w')=w$.  This word is denoted $\alpha^{-1}(w)$.
  
  Now, we prove that $Y$ is a prefix code.  Let $v,w\in\alpha(A^{*})
  \subseteq(B_1^{*}B_2)^{*}$ and assume that $v$ is a prefix of $w$.  Write
  $w=w_1\cdots w_n$ with $w_i\in B_1^{*}B_2$.  Since $v$ ends with a letter from
  $B_2$ we deduce that $v=w_1\cdots w_i$ for some $1\leq i\leq n$.  Let
  $w'=\alpha^{-1}(w)=a_1\cdots a_n$.  We have $v'=\alpha^{-1}(v)=a_1\cdots a_i$.
  Now, if $v,w\in\alpha(X)$ then we get $v',w'\in X$.  Since $X$ is a prefix
  code we get $i=n$.  Hence $v=w$, proving that $Y$ is also a prefix code.
  
  Finally, we prove that $Y$ has synchronization delay $d+1$.
  Let $u,v,w$ in $B^*$ such that $uvw\in Y^{*}$ and $v\in Y^{d+1}$.
  We need to prove that $uv\in Y^{*}$.
  Since $v\in Y^{d+1}$, it can be written $v=v_0v_1\cdots v_d$ with $v_i\in Y$
  for $0\leq i\leq d$.
  Then, let us remark that $\alpha(A)\subseteq B_1^{*}B_2$ is a prefix code with
  synchronisation delay $1$.
  Since $uv_0\cdots v_d w\in Y^{*}\subseteq \alpha(A)^*$ and $v_0\in 
  Y\subseteq\alpha(A)^{+}$, we deduce that
  $uv_0$ belongs to $\alpha(A)^*=\alpha(A^{*})$, as well as $v_1\cdots v_d$ and
  $w$.  Let $r=\alpha^{-1}(uv_0)$, $s=\alpha^{-1}(v_1\cdots v_d)$ and
  $t=\alpha^{-1}(w)$. We have $rst=\alpha^{-1}(uvw)$ and since $uvw\in 
  Y^{*}=\alpha(X^{*})$, we deduce that $rst\in X^{*}$. Similarly, from
  $v_1\cdots v_d\in Y^{d}=\alpha(X)^d$, we get $s\in X^d$. Now, $X$ has 
  synchronisation delay $d$. Therefore, $rs\in X^*$, meaning that
  $uv=uv_0\cdots v_d\in\alpha(rs)\subseteq \alpha(X^*)=Y^{*}$. 
\end{proof}

Marked substitutions also preserve unambiguity of union, concatenation and 
Kleene star.

\begin{lemma}\label{lem:MS-unambiguity}
  Let $\alpha\colon A\to \mathcal{P}(B^{*})$ be a marked substitution and let 
  $L_1,L_2\subseteq A^{*}$.
  \begin{enumerate}
    \item  If the union $L_1\cup L_2$ is unambiguous then so is 
    $\alpha(L_1)\cup\alpha(L_2)$.
  
    \item  If the concatenation $L_1\cdot L_2$ is unambiguous then so is 
    $\alpha(L_1)\cdot\alpha(L_2)$.
  
    \item  If the Kleene star $L_1^{*}$ is unambiguous then so is 
    $\alpha(L_1)^{*}$.
  \end{enumerate}
\end{lemma}

\begin{proof}
As stated in the previous proof, a marked substitution is one-to-one.
We denote by $\alpha^{-1}(w)$ the unique inverse of $w$, for $w$ in $\alpha(A^*)$.
  
  If $w\in\alpha(L_1)\cup\alpha(L_2)$ then $\alpha^{-1}(w)\in L_1\cup L_2$. 
  This shows that unambiguity of union is preserved by $\alpha$.

  Assume now that the concatenation $L_1\cdot L_2$ is unambiguous.  Let
  $w\in\alpha(L_1)\cdot\alpha(L_2)$ and consider its unique factorisation
  $w=v_1b_1\cdots v_nb_n$ as above and $\alpha^{-1}(w)=a_1\cdots a_n$.  Since
  $\alpha(L_1)\subseteq(B_1^{*}B_2)^{*}$, a factorisation of $w$ according to
  $\alpha(L_1)\cdot\alpha(L_2)$ must be of the form $w=(v_1b_1\cdots
  v_ib_i)\cdot(v_{i+1}b_{i+1}\cdots v_nb_n)$ with $a_1\cdots a_i\in L_1$ and
  $a_{i+1}\cdots a_n\in L_2$.  From unambiguity of the product $L_1\cdot L_2$ we
  deduce that such a factorisation of $w$ is unique.  Hence, the product
  $\alpha(L_1)\cdot\alpha(L_2)$ is unambiguous.
  
  We can prove similarly that $\alpha$ preserves unambiguity of Kleene stars.
\end{proof}

We will be interested in marked substitutions that are defined by regular
expressions.  A \emph{regular marked substitution} (RMS) is a map $\alpha\colon
A\to\reg{B^{*}}$ which assigns a regular expression $\alpha(a)$ over $B$ to each
letter $a\in A$ such that $\tilde{\alpha}\colon A\to\mathcal{P}(B^{*})$ defined
by $\tilde{\alpha}(a)=\lang{\alpha(a)}$ is a marked substitution.  

Let $\alpha\colon A\to\reg{B^{*}}$ be an RMS and let $E$ be a regular expression
over $A$.  We define $\alpha(E)=E[\alpha(a)/a, \forall a\in A]$ to be the 
regular expression over $B$ obtained from $E$ by substituting each
occurrence of a letter $a\in A$ with the expression $\alpha(a)$.  Notice
that $\alpha$ is compositional:
$\alpha(E_1\cup E_2)=\alpha(E_1)\cup\alpha(E_2)$,
$\alpha(E_1\cdot E_2)=\alpha(E_1)\cdot\alpha(E_2)$ and
$\alpha(E_1^{*})=\alpha(E_1)^{*}$. 
In particular, we have $\tilde{\alpha}(\lang{E})=\lang{\alpha(E)}$.

Further, we say that a RMS $\alpha$ is \emph{unambiguous} (URMS) if $\alpha(a)$
is unambiguous for each $a\in A$.  Similarly, an RMS $\alpha$ is SD-regular
(SDRMS) if $\alpha(a)$ is an SD-regular expression for each $a\in A$.
We obtain:

\begin{corollary}\label{cor:RMS}
  Let $\alpha\colon A\to\reg{B^{*}}$ be an RMS and $E$ be a regular
  expression over $A$.
  \begin{enumerate}
    \item  If $\alpha$ and $E$ are SD-regular, then $\alpha(E)$ is SD-regular.
  
    \item  If $\alpha$ and $E$ are unambiguous, then $\alpha(E)$ is unambiguous.
  \end{enumerate}
\end{corollary}

\begin{proof}
  1. Let $F^{*}$ be a subexpression of $\alpha(E)$.  If $F^{*}$ is a subexpression
  of some $\alpha(a)$ then, $\alpha(a)$ being SD-regular we obtain than
  $\lang{F}$ is an SD prefix code. Otherwise, $F=\alpha(G)$ where $G^{*}$ is a 
  subexpression of $E$. Since $E$ is SD-regular,  $\lang{G}$ is a SD 
  prefix code. By lemma~\ref{lem:MarkSubst} we deduce that 
  $\lang{F}=\tilde{\alpha}(\lang{G})$ is a SD prefix code.
  
  2. First, we know that each $\alpha(a)$ is unambiguous. Next, a subexpression 
  of $\alpha(E)$ which is not a subexpression of some $\alpha(a)$ must be of 
  the form $\alpha(F)$ where $F$ is a subexpression of $E$. We conclude easily 
  using unambiguity of $E$ and Lemma~\ref{lem:MS-unambiguity}.
\end{proof}

\subsection{Proof of Theorem~\ref{thm:U-S-SD-expressions}}

\begin{proof}
  We first consider the set of neutral letters, i.e., letters whose image is the
  neutral element $1$ of $M$.  To ease the proof, we first explain how to handle
  them, and in the rest of the proof, focus on the case where we do not have neutral letters.

       Let
  $\varphi\colon\Sigma^{*}\to M$ be a morphism and
  $\Sigma_0=\{a\in\Sigma\mid\varphi(a)=1\}$ be the set of neutral letters.
  Further, let $\Sigma_1=\Sigma\setminus\Sigma_0$ and let
  $\varphi_1\colon\Sigma_1^{*}\to M$ be the restriction of $\varphi$ to
  $\Sigma_1^{*}$.  Let $\alpha\colon\Sigma_1\to\reg{\Sigma^{*}}$ be the regular
  marked substitution defined by $\alpha(a)=\Sigma_0^{*}a$.  Clearly, $\alpha$
  is unambiguous and since $\Sigma_0$ is a 1-SD prefix code we get that $\alpha$
  is SD-regular.  By Corollary~\ref{cor:RMS} we deduce that $\alpha$ preserves
  unambiguity and also SD-expressions.  It also preserves stabilising
  expressions, i.e., if $E\in\reg{\Sigma_1^{*}}$ is $\varphi_1$-stabilising then
  $\alpha(E)\in\reg{\Sigma^{*}}$ is $\varphi$-stabilising.  Indeed,
  $\varphi(\Sigma_0)$ is 1-stabilising.  Further, if $G^{*}$ is a subexpression
  of $\alpha(E)$ different from $\Sigma_0^{*}$ then there is a subexpression
  $F^{*}$ of $E$ such that $G=\alpha(F)$.  Hence,
  $\lang{G}=\tilde{\alpha}(\lang{F})$ and
  $X=\varphi(\lang{G})=\varphi_1(\lang{F})$ is stabilising.

  Now, suppose we have unambiguous, stabilising, SD-expressions
  $E_s$ for $\varphi_1$ and each $s\in M$: $\lang{E_s}=\varphi_1^{-1}(s)$.  
  We deduce that $E'_s=\alpha(E_s) \cdot \Sigma_0^*$ is an unambiguous, 
  stabilising, SD-expression. Moreover, we have $\lang{E'_s}=\varphi^{-1}(s)$. 
  
  \medskip
  In the rest of the proof, we assume that the morphism $\varphi$ has no neutral
  letters.  The proof is by induction on the size of $M$, using a result from
  Perrin and Pin \cite[Chapter XI, Proposition 4.14]{Perrin-Pin-Infinite-words} 
  stating that if $\varphi$ is a surjective morphism from $\Sigma^*$ to a finite
  aperiodic monoid $M$, then one of the following cases hold:
  \begin{enumerate}
    \item\label{PP1} $M$ is a cyclic monoid, meaning that $M$ is generated by a single element.
    \item\label{PP2} $M$ is isomorphic to $\widetilde{U}_n$ for some $n\geq1$ 

    \item\label{PP3} There is a partition $\Sigma=A\uplus B$ such that
    $\varphi(A^*)$ and $\varphi((A^*B)^*)$ are proper submonoids of $M$.
  \end{enumerate}
  We now treat the three cases above.
  \begin{enumerate}
    \item $M$ is a cyclic monoid.
    Then $M$ is of the form $\{1,s,s^2,\ldots,s^n\}$ with $s^is^j=s^{i+j}$ if
    $i+j\leq n$ and $s^n$ otherwise.  Notice that since we have no neutral
    letters, $\varphi^{-1}(1)=\{\varepsilon\}$.
    For $1\leq i \leq n$, we denote by $\Sigma_i$ the set of letters whose image is $s^i$.
    Now, we define inductively stabilising, unambiguous, SD-regular expressions 
    $E_j$ such that $\lang{E_j}=\varphi^{-1}(s^{j})$ for $1\leq j\leq n$.
    Let $E_1=\Sigma_1$. Then, for $1<j<n$ we let 
    $$
    E_j=\Sigma_j\cup \bigcup_{1\leq i<j} E_i\Sigma_{j-i} \,.
    $$
    Notice that expressions $E_j$ for $1\leq j<n$ are unambiguous and do not 
    use the Kleene star, hence are stabilising and SD. Finally, let 
    $$
    E_n=\Big( \Sigma_n\cup \bigcup_{1\leq i,j<n\mid n\leq i+j} E_i\Sigma_{j}
    \Big) \Sigma^{*} \,.    
    $$
    Notice that the separation between the first part of $E_n$ and $\Sigma^*$ is done at the letter at which the image of the prefix reaches $s^n$. Then as above, $E_n$ is unambiguous, and $\Sigma$ is an $n$-stabilising, 1-SD
    prefix code. Moreover, $\lang{E_j}=\varphi^{-1}(s^{j})$ for $1\leq j\leq 
    n$, which concludes the proof in the case of cyclic monoids.
    
    \item $M$ is isomorphic to $\widetilde{U}_n$ for some $n$.  Then $M$ is of the
    form $\{1,s_1,\ldots,s_n\}$ where $s_is_j=s_i$ for all $1\leq i,j\leq n$.
    As above, since we have no neutral letters, we deduce that
    $\varphi^{-1}(1)=\{\varepsilon\}$.  We similarly define
    $\Sigma_i=\varphi^{-1}(s_i)\cap\Sigma$.  Clearly,
    $\varphi^{-1}(s_i)=\lang{\Sigma_i\Sigma^{*}}$, and $\Sigma_i\Sigma^{*}$ is
    an unambiguous, 1-stabilising, 1-SD regular expression.

    \item There is a partition $\Sigma=A\uplus B$ such that $M_A=\varphi(A^*)$
    and $M_B=\varphi((A^*B)^*)$ are proper submonoids of $M$.
    We set $C=\varphi(A^*B)\subseteq M_B$ and view $C$ as a new alphabet.
    Note that $C$ generates $M_B$.  Finally, let $f\colon A^*\to M_A$ be the
    restriction of $\varphi$ to $A$ and $g\colon C^*\to M_B$ be the evaluation
    morphism defined by $g(c)=c$ for $c\in C$.

    Then $f$ and $g$ are surjective morphisms to monoids $M_A$ and $M_B$ whose
    size are smaller than $M$.  We can thus use the induction hypothesis and get
    unambiguous, stabilising, SD-expressions for elements of $M_A$ and $M_B$
    with respect to $f$ and $g$ respectively.  Given an element $s$ in $M_A$, we
    get an unambiguous, $f$-stabilising, SD-expression $E_s$ over $A$ for
    $f^{-1}(s)$.  Similarly for an element $t$ in $M_B$, we get an unambiguous,
    $g$-stabilising, SD-expression $F_t$ over the alphabet $C$ for $g^{-1}(t)$.  
    Notice that we have in particular expressions $E_1$ for $f^{-1}(1)$ and 
    $F_1$ for $g^{-1}(1)$. 
    
    To go back from expressions over $A, C$ to expressions over $\Sigma^{*}$, we
    first define expressions $G_c$ for $\varphi^{-1}(c)\cap A^{*}B$, for each
    $c\in C$:
    $$
    G_c=\bigcup\limits_{s\in M_A,b\in B \mid s\varphi(b)=c} E_s\cdot b \,.
    $$  
    Notice that $\alpha\colon C\to\reg{\Sigma^{*}}$ defined by $\alpha(c)=G_c$
    is a regular marked substitution with respect to the partition
    $\Sigma=A\uplus B$.  Indeed, $\alpha(c)\subseteq A^{*}B$ and
    $\varphi(\alpha(c))=c$ which implies that the images are pairwise disjoint.
    Moreover, each $G_c$ is unambiguous, $\varphi$-stabilising and SD-regular.
    
    Let $t\in M_B$.  By Corollary~\ref{cor:RMS}, $\alpha(F_t)$ is unambiguous
    and SD-regular.  It is also $\varphi$-stabilising.  Indeed, let $G^{*}$ be a
    subexpression of $\alpha(F_t)$.  If $G^{*}$ is a subexpression of some
    $\alpha(c)$ we get the result since $G_c$ is $\varphi$-stabilising.  Otherwise, there
    is a subexpression $F^{*}$ of $F_t$ such that $G=\alpha(F)$.  Hence,
    $\lang{G}=\tilde{\alpha}(\lang{F})$ and
    $X=\varphi(\lang{G})=g(\lang{F})$ is stabilising.

    For each $t\in M_B$, we have
    $\lang{\alpha(F_t)}=\varphi^{-1}(t)\cap(A^{*}B)^{*}$.  
    Finally, we need to combine these elements to get an expression for any word
    over $\Sigma$.  Noticing that $\Sigma^*=(A^*B)^*A^*$, we define for $s\in M$
    $$
    E'_s=\bigcup_{t\in M_B,r\in M_A,tr=s} \alpha(F_t)\cdot E_r \,.
    $$
    We can easily show that $E'_s$ is unambiguous and satisfies
    $\lang{E_s}=\varphi^{-1}(s)$.  It is also clearly a $\varphi$-stabilising 
    SD-expressions. Hence, we get the result.  \qedhere
  \end{enumerate}
\end{proof}

\section{Combinator expressions}\label{sec:comb-expr}
In this section, we present our combinators to compute first order definable functions
from finite words to finite words.  The simpler combinators of unambiguous concatenation
and sum are similar to those in \cite{AlurFreilichRaghothaman14,DGK-lics18}, but we
differ in the equivalent of the Kleene-Star to match the aperiodicity that we tackle.

\subsection{Simple Functions}

For each $v\in\Gamma^{*}$ we have a constant function $f_v$ defined by
$f_v(u)=v$ for all $u\in\Sigma^{*}$.  Abusing notations, we simply denote the
constant function $f_v$ by $v$.  We denote by
$\bot\colon\Sigma^{*}\to\Gamma^{*}$ the function with empty domain.  These
atomic functions are the most simple ones.

\subsection{Unambiguous sums}\label{sec:ifthenelse}

We will use two equivalent ways of defining a function by cases. First, the 
if-then-else construct is given by $h=\Ifthenelse{L}{f}g$ where 
$f,g\colon\Sigma^{*}\to\D$ are functions and $L\subseteq\Sigma^{*}$ is a 
language. We have $\dom{h}=(\dom{f}\cap L)\cup(\dom{g}\setminus L)$.
Then, for $w\in\dom{h}$ we have 
$$
h(w)=
\begin{cases}
  f(w) & \text{if } w\in L \\
  g(w) & \text{otherwise.}
\end{cases}
$$
We will often use this case definition with $L=\dom{f}$. To simplify notations 
we define $f+g = \Ifthenelse{\dom{f}}{f}g$. Note that 
$\dom{f+g}=\dom{f}\cup\dom{g}$ but the sum is not 
commutative and $g+f = \Ifthenelse{\dom{g}}{g}f$. For $w\in\dom{f}\cap\dom{g}$ we 
have $(f+g)(w)=f(w)$ and $(g+f)(w)=g(w)$. When the domains of $f$ and $g$ 
are disjoint then $f+g$ and $g+f$ are equivalent functions with domain 
$\dom{f}\uplus\dom{g}$. In all cases the sum is associative and the
sum notation is particularly useful when applied to a sequence 
$f_1,\ldots,f_n$ of functions:
$$
\sum_{1\leq i\leq n}f_i = f_1+\cdots+f_n = \Ifthenelse{\dom{f_1}}{f_1}{
\Ifthenelse{\dom{f_2}}{f_2}{\cdots\Ifthenelse{\dom{f_{n-1}}}{f_{n-1}}{f_n}}}
$$
If the domains of the functions are pairwise disjoint then this sum is 
associative and commutative.

Further, we let $\SimpleFun{L}{f}=\Ifthenelse{L}{f}\bot$ the function $f$ 
restricted to $L\cap\dom{f}$.
When $L=\{w\}$ is a singleton, we simply write $\SimpleFun{w}{f}$.

\subsection{The Hadamard product}

The Hadamard product of two functions $f,g\colon\Sigma^{*}\to\Gamma^{*}$ first
applies $f$ and then applies $g$.  It is denoted by $f\odot g$.  Its domain is
$\dom{f}\cap\dom{g}$ and $(f\odot g)(u)=f(u)g(u)$ for each input word $u$ in its
domain.

\subsection{The unambiguous Cauchy product}

Consider two functions $f,g\colon\Sigma^{*}\to\Gamma^{*}$. The unambiguous 
Cauchy product of $f$ and $g$ is the function $f\cdot g$ whose domain is the set 
of words $w\in\Sigma^{*}$ which admit a unique factorization $w=uv$ with 
$u\in\dom{f}$ and $v\in\dom{g}$, and in this case, the computed output is 
$f(u)g(v)$.

Contrary to the Hadamard product which reads its full input word $w$ twice, 
first applying $f$ and then applying $g$, the Cauchy product \emph{splits 
unamgibuously} its input word $w$ as $uv$, applies $f$ on $u$ and then $g$ on 
$v$.

Sometimes we may want to reverse the output and produce $g(v)f(u)$. This 
\emph{reversed} Cauchy product can be defined using the Hadamard product as
$$
f \rcdot g = ((\SimpleFun{\dom{f}}{\varepsilon})\cdot g)
\odot(f\cdot(\SimpleFun{\dom{g}}{\varepsilon}))
$$

\subsection{The $k$-chained Kleene-star and its reverse}

Let $L\subseteq\Sigma^*$ be a code, let $k\geq1$ be a natural number and let
$f\colon\Sigma^*\to\D$ be a partial function.  We define the $k$-chained
Kleene-star $\kstar{k}{L}{f}\colon\Sigma^{*}\to\D$ and its reverse
$\krstar{k}{L}{f}\colon\Sigma^{*}\to\D$ as follows.  

The domain of both these functions is contained in $L^{*}$, the set of words having a
(unique) factorization over the code $L$.
Let $w\in L^{*}$ and consider its unique factorization $w=u_1u_2\cdots u_n$ with
$n\geq0$ and $u_i\in L$ for all $1\leq i\leq n$.  Then, 
$w\in\dom{\kstar{k}{L}{f}}=\dom{\krstar{k}{L}{f}}$ if $u_{i+1}\cdots u_{i+k}\in\dom{f}$ 
for all $0\leq i\leq n-k$ and in this case we set
\begin{align*}
  \kstar{k}{L}{f}(w) & =
  f(u_1\cdots u_k)\cdot f(u_2\cdots u_{k+1}) \cdots f(u_{n-k+1}\cdots u_n) \\
  \krstar{k}{L}{f}(w) & = f(u_{n-k+1}\cdots u_n) \cdots 
  f(u_2\cdots u_{k+1}) \cdot f(u_1\cdots u_k) \,.
\end{align*}
Notice that when $n<k$, the right-hand side is an empty product and we get
$\kstar{k}{L}{f}(w)=\varepsilon$ and $\krstar{k}{L}{f}(w)=\varepsilon$.
When $k=1$ and $L=\dom{f}$ is a code then we simply write 
$f^{\star}=\kstar{1}{\dom{f}}{f}$ and $\rstar{f}=\krstar{1}{\dom{f}}{f}$.
We have $\dom{f^{\star}}=\dom{\rstar{f}}=L^{*}$.

The k-chained Kleene star was also defined in
\cite{AlurFreilichRaghothaman14,DGK-lics18}; however as we will see below, we use it in a
restricted way for aperiodic functions.

\subsection{SD-regular transducer expressions (SDRTE)}

SD-regular transducer expressions (\SDRTEs) are obtained from classical regular transducer
expressions (\RTEs) \cite{AlurFreilichRaghothaman14,DGK-lics18} by restricting the
$k$-chained Kleene-star $\kstar{k}{L}{f}$ and its reverse $\krstar{k}{L}{f}$ to
aperiodic languages $L$ that are prefix codes of bounded synchronisation delay.
The if-then-else choice $\Ifthenelse{L}{f}g$ is also restricted to aperiodic languages
$L$.  Hence, the syntax of \SDRTEs is given by the grammar:
$$
C ::= \bot \mid v \mid \Ifthenelse{L}{C}{C} \mid C\odot C \mid C\cdot C 
\mid \kstar{k}{L}{C} \mid \krstar{k}{L}{C}
$$
where $v\in\Gamma^{*}$, and $L\subseteq\Sigma^*$ ranges over aperiodic languages
(or equivalently SD-regular expressions), which are also prefix codes with
bounded synchronisation delay for $\kstar{k}{L}{C}$ and $\krstar{k}{L}{C}$.

The semantics of \SDRTEs is defined inductively. $\sem{\bot}$ is the function 
which is nowhere defined, $\sem{v}$ is the constant function such as 
$\sem{v}(u)=v$ for all $u\in\Sigma^{*}$, and the semantics of the other 
combinators has been defined in the above sections.

As discussed in Section~\ref{sec:ifthenelse}, we will use binary sums
$C+C' = \Ifthenelse{\dom{C}}{C}{C'}$ and generalised sums $\sum_{i}C_i$. Also, 
we use the abbreviation $\SimpleFun{L}{C}=\Ifthenelse{L}{C}{\bot}$ and the 
reversed Cauchy product 
$C \rcdot C' = ((\SimpleFun{\dom{C}}{\varepsilon})\cdot C')
\odot(C\cdot(\SimpleFun{\dom{C'}}{\varepsilon}))$.

\begin{lemma}\label{lem-SDRTEdomAp}
If $C$ is an \SDRTE, then $\dom C$ is an aperiodic language.
\end{lemma}

\begin{proof}
  We prove the statement by induction on the syntax of \SDRTEs.  We recall that
  aperiodic languages are closed under concatenation, union, intersection and complement.
  \begin{itemize}
    \item $\dom{\bot}=\emptyset$ and $\dom{v}=\Sigma^{*}$ are aperiodic languages.

    \item $C=\Ifthenelse{L}{C_1}{C_2}$. 
    By induction, $\dom{C_1}$ and $\dom{C_2}$ are aperiodic.
    We have $\dom{C}=(L\cap\dom{C_1})\cup(\dom{C_2}\setminus L)$, which is
    aperiodic thanks to the closure properties of aperiodic languages.
    
    \item $C=C_1\odot C_2$. By induction, $\dom{C_1}$ and $\dom{C_2}$ are aperiodic.
    We deduce that $\dom{C}=\dom{C_1}\cap\dom{C_2}$ is aperiodic.
    
    \item $C=C_1\cdot C_2$.  By induction, $L_1=\dom{C_1}$ and $L_2=\dom{C_2}$
    are aperiodic.  We have $\dom{C}\subseteq\dom{C_1}\cdot\dom{C_2}$.
    However, $C$ is undefined on words having more than one decomposition. 
    A word which admits at least two decompositions can be written $uvw$ with 
    $v\neq\varepsilon$,
    $u,uv\in L_1$ and $vw,w\in L_2$. Let $\varphi\colon\Sigma^{*}\to M$ be a 
    morphism to a finite aperiodic monoid recognising both $L_1$ and $L_2$. We 
    have $L_1=\varphi^{-1}(P_1)$ and $L_2=\varphi^{-1}(P_2)$ for some 
    $P_1,P_2\subseteq M$. The set $L_3$ of words having at least two 
    decompositions is precisely
    $$
    L_3=\bigcup_{\substack{r,s,t\mid r,rs\in P_1 \wedge st,t\in P_2}}
    \varphi^{-1}(r) ( \varphi^{-1}(s) \setminus \{\varepsilon\}) \varphi^{-1}(t)
    $$
    which is aperiodic. We deduce that $\dom{C}=(L_1\cdot L_2)\setminus L_3$ is 
    aperiodic.
    
    \item $C=\kstar{k}{L}{C'}$. By induction, $\dom{C'}$ is aperiodic and by 
    definition $L$ is an aperiodic SD prefix code. Hence $L^{*}$ is aperiodic.
    Notice that $\dom{C}\subseteq L^{*}$ but $C$ is undefined on words 
    $w=u_1\cdots u_n$ with $u_i\in L$ if there is a factor $u_{i+1}\cdots 
    u_{i+k}$ which is not in $\dom{C'}$. We deduce that
    $\dom{C}=L^{*}\setminus(L^*(L^{k}\setminus\dom{C'})L^*)$ which is 
    aperiodic thanks to the closure properties given above.
    
    \item Notice that $\dom{\krstar{k}{L}{C'}}=\dom{\kstar{k}{L}{C'}}$, which 
    is aperiodic as proved above.
    \qedhere
  \end{itemize}
\end{proof}

\begin{proposition}\label{prop:LQ-SDRTE}
  Given an \SDRTE $C$ and a letter $a\in\Sigma$, 
  \begin{enumerate}
    \item we can construct an \SDRTE $\LQ{a}{C}$ such that
    $\dom{\LQ{a}{C}}=a^{-1}\dom{C}$ and $\sem{\LQ{a}{C}}(w)=\sem{C}(aw)$ for all
    $w\in a^{-1}\dom{C}$,
  
    \item we can construct an \SDRTE $\RQ{a}{C}$ such that
    $\dom{\RQ{a}{C}}=\dom{C}a^{-1}$ and $\sem{\RQ{a}{C}}(w)=\sem{C}(wa)$ for all
    $w\in\dom{C}a^{-1}$.
  \end{enumerate}
\end{proposition}

\begin{proof}
We recall that aperiodic languages are closed under left and right quotients.
  The proof is by structural induction on the given \SDRTE $C$ over alphabet $\Sigma$.  We
  only construct below the \SDRTEs for the left quotient.  Formulas for the right quotient
  can be obtained similarly.  A point to note is that, unlike the left quotient, the right
  quotient of a language might break its prefix code property, which could be a problem if
  applied to a parsing language $L$ used for $k$-star or its reverse.  
  However, the quotient by a letter only modifies the first or last copy of $L$, which can
  be decoupled so that the remaining iterations are still performed with the same parsing
  language $L$.
  
  \begin{description}
    \item[Basic cases.]  We define $\LQ{a}{\bot}=\bot$ and $\LQ{a}{v}=v$ for
    $v\in\Gamma^{*}$.
  
    \item[If-then-else.]  Let $C=\Ifthenelse{L}{C_1}{C_2}$.  We define
    $\LQ{a}{C}=\Ifthenelse{a^{-1}L}{\LQ{a}{C_1}}{\LQ{a}{C_2}}$.
    
    Recall that $\dom{C}=(\dom{C_1}\cap L)\cup(\dom{C_2}\setminus L)$. We 
    deduce that
    \begin{align*}
      a^{-1}\dom{C} & = 
      ((a^{-1}\dom{C_1})\cap(a^{-1}L))\cup((a^{-1}\dom{C_2})\setminus(a^{-1}L))\\
       & = \dom{\Ifthenelse{a^{-1}L}{\LQ{a}{C_1}}{\LQ{a}{C_2}}}
    \end{align*}
    Moreover, for $w\in a^{-1}\dom{C}$, we have
    \begin{align*}
      \sem{C}(aw) & = 
      \begin{cases}
        \sem{C_1}(aw) & \text{if } aw\in L  \\
        \sem{C_2}(aw) & \text{otherwise.}
      \end{cases}
      =
      \begin{cases}
        \sem{\LQ{a}{C_1}}(w) & \text{if } w\in a^{-1}L  \\
        \sem{\LQ{a}{C_2}}(w) & \text{otherwise.}
      \end{cases}
      \\
      &= \sem{\Ifthenelse{a^{-1}L}{\LQ{a}{C_1}}{\LQ{a}{C_2}}}(w)
    \end{align*}
    
    \item[Hadamard product.]  Let $C=C_1\odot C_2$.  We define
    $\LQ{a}{C}=\LQ{a}{C_1}\odot\LQ{a}{C_2}$.
    
    Recall that $\dom{C}=\dom{C_1}\cap\dom{C_2}$.  We deduce that
    \begin{align*}
      a^{-1}\dom{C} & = (a^{-1}\dom{C_1})\cap(a^{-1}\dom{C_2})
      = \dom{\LQ{a}{C_1}\odot\LQ{a}{C_2}}
    \end{align*}
    Moreover, for $w\in a^{-1}\dom{C}$, we have
    \begin{align*}
      \sem{C}(aw) & = \sem{C_1}(aw)\sem{C_2}(aw)
      = \sem{\LQ{a}{C_1}}(w) \sem{\LQ{a}{C_2}}(w)
      = \sem{\LQ{a}{C_1}\odot\LQ{a}{C_2}}(w)
    \end{align*}
    
    \item[Cauchy product.] Let $C=C_1\cdot C_2$. The \SDRTE $\LQ{a}{C}$ is the 
    unambiguous sum of two expressions depending on whether the letter $a$ is 
    removed from $C_1$ or from $C_2$. Hence, we let
    $C'=(\LQ{a}{C_1})\cdot{C_2}$ and
    $C''=(\SimpleFun{\varepsilon}{\sem{C_1}(\varepsilon)})\cdot(\LQ{a}{C_2})$. 
    Notice that $\dom{C''}=\emptyset$ when $\varepsilon\notin\dom{C_1}$ 
    (i.e., $\sem{C_1}(\varepsilon)=\bot$).
    Now, we define $\LQ{a}{C}=\SimpleFun{(a^{-1}\dom{C})}{(C'+C'')}$.
    
    Let $w\in a^{-1}\dom{C}$. Then $aw$ admits a unique factorization
    $aw=uv$ with $u\in\dom{C_1}$ and $v\in\dom{C_2}$. There are two exclusive cases.
    
    If $u\neq\varepsilon$ then $u=au'$ with $u'\in a^{-1}\dom{C_1}$. The word 
    $w$ admits a unique factorization according to $\dom{\LQ{a}{C_1}}\dom{C_2}$ 
    which is $w=u'v$. Hence, $w\in\dom{C'}$ and
    $$
    \sem{C}(aw) = \sem{C_1}(u)\sem{C_2}(v) 
    = \sem{\LQ{a}{C_1}}(u')\sem{C_2}(v)
    = \sem{C'}(w) \,.
    $$
    
    If $u=\varepsilon$ then $v=av'$ and $v'\in a^{-1}\dom{C_2}$.  The word $w=v'$
    admits a unique factorization according to
    $\{\varepsilon\}\cdot\dom{\LQ{a}{C_2}}$ which is $w=\varepsilon\cdot w$.
    Hence, $w\in\dom{C''}$ and
    $$
    \sem{C}(aw) = \sem{C_1}(\varepsilon)\sem{C_2}(v) 
    = \sem{C_1}(\varepsilon)\sem{\LQ{a}{C_2}}(w)
    = \sem{C''}(w) \,.
    $$
    
    We deduce that $a^{-1}\dom{C}\subseteq\dom{C'}\cup\dom{C''}=\dom{C'+C''}$ and 
    $\dom{a^{-1}C}=a^{-1}\dom{C}$ as desired.
    
    Finally, assume that $w\in\dom{C'}\cap\dom{C''}$. Then, $w$ admits two 
    factorizations $w=u'v=\varepsilon v'$ with $u'\in\dom{\LQ{a}{C_1}}$, 
    $v\in\dom{C_2}$, $\varepsilon\in\dom{C_1}$ and $v'\in\dom{\LQ{a}{C_2}}$. We 
    deduce that $aw$ admits two distinct factorizations $aw=(au')v=\varepsilon(av')$ 
    with $au',\varepsilon\in\dom{C_1}$ and $v,av'\in\dom{C_2}$. This is a 
    contradiction with $aw\in\dom{C}$.
    
    We deduce that in both cases above, we have
    $$
    \sem{C}(aw) = \sem{\SimpleFun{(a^{-1}\dom{C})}{(C'+C'')}}(w) \,.
    $$
    
    \item[$k$-star.] Let $L\subseteq\Sigma^{*}$ be an aperiodic prefix code with bounded 
    synchronisation delay and let $C$ an \SDRTE. 
    Notice that, since $L$ is a code, $\varepsilon\notin L$. Also, 
    $a^{-1}\dom{\kstar{k}{L}{C}}\subseteq a^{-1}L^{*}=(a^{-1}L)L^{*}$. Let 
    $w\in a^{-1}L^{*}$. It admits a unique factorization $w=u'_1u_2\cdots 
    u_n$ with $u_1=au'_1\in L$ and $u_2,\ldots,u_n\in L$. The unique 
    factorization of $aw$ according to the code $L$ is $aw=u_1u_2\cdots u_n$.
    
    Now, by defintion of $k$-star, when $n<k$ we have 
    $\sem{\kstar{k}{L}{C}}(aw)=\varepsilon$. Hence, we let 
    $C'=\SimpleFun{\big( (a^{-1}L)\cdot L^{<k-1} \big)}{\varepsilon}$ so that 
    in the case $n<k$ we get
    $$
    \sem{\kstar{k}{L}{C}}(aw)=\varepsilon=\sem{C'}(w) \,.
    $$
    Next we assume that $n\geq k$. We define two \SDRTEs:
    \begin{align*}
      C'' & = \big( \SimpleFun{\big( (a^{-1}L)\cdot L^{k-1} \big)}{\LQ{a}{C}} \big) 
      \cdot (\SimpleFun{L^{*}}{\varepsilon}) \\
      C''' & = (\SimpleFun{a^{-1}L}{\varepsilon}) \cdot \kstar{k}{L}{C}
    \end{align*}
    Notice that $L^{*}$ is aperiodic since $L$ is an aperiodic prefix code with bounded 
    synchronisation delay. Hence, $C''$ is indeed an \SDRTE.
    We get
    \begin{align*}
      \sem{C''}(w) & = \sem{\LQ{a}{C}}(u'_1u_2\cdots u_{k})
      = \sem{C}(u_1u_2\cdots u_{k}) \\
      \sem{C'''}(w) & = \sem{C}(u_2\cdots u_{k+1}) \sem{C}(u_3\cdots u_{k+2})
      \cdots \sem{C}(u_{n-k+1}\cdots u_n) \\
      \sem{C'' \odot C'''}(w) & = \sem{\kstar{k}{L}{C}}(aw) \,.
    \end{align*}
    Therefore, we define the \SDRTE
    $$
    \LQ{a}{(\kstar{k}{L}{C})} = \SimpleFun{(a^{-1}L^{*})}{\big(C'+(C''\odot C''')\big)} 
    \,.
    $$
    Notice that $\dom{C'}=a^{-1}L^{<k}$ and $\dom{C''\odot C'''}\subseteq a^{-1}L^{\geq
    k}$ are disjoint.
    
    \item[Reverse $k$-star.]  This case is similar.  Let $L\subseteq\Sigma^{*}$
    be an aperiodic prefix code with bounded synchronisation delay and let $C$ an \SDRTE.
    We define
    $$
    \LQ{a}{(\krstar{k}{L}{C})} = \SimpleFun{(a^{-1}L^{*})}{\big(C'+(C''''\odot C'')\big)}
    $$
    where $C',C''$ are as above and $C''''=(\SimpleFun{a^{-1}L}{\varepsilon}) 
    \cdot \krstar{k}{L}{C}$.
    \qedhere
  \end{description}
\end{proof}

\begin{lemma}\label{lem-projection}
	Given an \SDRTE $C$ over an alphabet $\Sigma$ and a sub-alphabet $\Sigma'\subseteq\Sigma$, 
  we can construct an \SDRTE $C'$ \emph{over alphabet $\Sigma'$} such that
  $\dom{C'}\subseteq \Sigma'^{*}$ and for any word
  $w$ in $\Sigma'^*$, $\sem{C}(w)=\sem{C'}(w)$.
\end{lemma}
\begin{proof}
	The proof itself is rather straightforward, and simply amounts to get rid of letters
	that do not appear in $\Sigma'$.  We first construct $C'$ by structural induction from
	$C$, and then prove that it is indeed an \SDRTE. Thus $C'$ is defined as follows:
  \begin{itemize}
     \item if $C=\bot$ then $C'=\bot$, 
     \item if $C=v$ then $C'=v$, with $\dom{v}=\Sigma'^{*}$ here since $C'$ is over $\Sigma'$, 
    \item if $C=\Ifthenelse{L}{C_1}{C_2}$ then $C'=\Ifthenelse{(L\cap\Sigma'^{*})}{C'_1}{C'_2}$,
    \item if $C=C_1\odot C_2$ then $C'=C'_1\odot C'_2$, 
    \item if $C=C_1\cdot C_2$ then $C'=C'_1\cdot C'_2$,
    \item  if $C=\kstar{k}{L}{C_1}$ then $C'=\kstar{k}{L\cap\Sigma'^{*}}{C'_1}$,
    \item if $C=\krstar{k}{L}{C_1}$ then $C'=\krstar{k}{L\cap\Sigma'^{*}}{C'_1}$.
  \end{itemize}
  To prove that $C'$ is SD-regular,
	we construct, given an SD-expression $E$ for $L$ over $\Sigma$,
	an SD-expression $E'$ over $\Sigma'$ for $L\cap \Sigma'^*$.
	Again, the proof is an easy structural induction:
\begin{itemize}
	\item   if $E=\emptyset$ then $E'=\emptyset$, 
  \item if $E=a\in\Sigma'$ then $E'=a$, 
 \item if $E=a\in\Sigma\setminus\Sigma'$ then $E'=\emptyset$,
  \item if $E=E_1+E_2$ then $E'=E'_1+E'_2$, 
  \item if $E=E_1\cdot E_2$ then $E'=E'_1\cdot E'_2$, 
  \item if $E=E_1^{*}$ then $E'=E_1'^{*}$.
  \end{itemize}
  We conclude by stating that being a prefix code with bounded synchronisation delay is a
  property preserved by subsets, hence $E'$ is an SD-expression.
\end{proof}

%
\begin{gpicture}[name=gpic:InternalRun,ignore]
  \gasset{Nframe=n}
  \node(ui)(15,2){$w$}
  \node(ui+1)(45,2){$w'$}
  \gasset{AHnb=0}
  \drawline(0,0)(0,-31)
  \drawline(30,0)(30,-31)
  \drawline(60,0)(60,-31)
  \gasset{AHnb=1,arcradius=0.8}
  \drawline(30,-5)(50,-5)(50,-7)(35,-7)(35,-9)(45,-9)(45,-11)(30,-11)
  \node(x0)(53,-6){$x_0$}
  \drawline(30,-11)(20,-11)(20,-13)(25,-13)(25,-15)(18,-15)(18,-17)(30,-17)
  \node(x1)(17,-12){$x_1$}
  \drawline(30,-17)(40,-17)(40,-19)(30,-19)
  \node(x2)(43,-18){$x_2$}
  \drawline(30,-19)(15,-19)(15,-21)(22,-21)(22,-23)(10,-23)(10,-25)(30,-25)
  \node(x3)(12,-20){$x_3$}
  \gasset{Nw=1.6,Nh=1.6,Nfill=y,ExtNL=y,NLdist=.7}
  \node[NLangle=45](p1)(30,-5){$p_1$}
  \node[NLangle=135](q1)(30,-11){$q_1$}
  \node[NLangle=45](q2)(30,-17){$q_2$}
  \node[NLangle=-135](q2)(30,-19){$q_3$}
  \node[NLangle=0](q3)(30,-25){$p_2$}
\end{gpicture}
\begin{gpicture}[name=gpic:LR-concat,ignore]
  \gasset{Nframe=n}
  \node(ui)(15,3){$u$}
  \node(ui+1)(45,3){$v$}
  \gasset{AHnb=0}
  \drawline(0,2)(0,-31)
  \drawline(30,2)(30,-31)
  \drawline(60,2)(60,-31)
  \gasset{AHnb=1,arcradius=0.8}
  \drawline(0,-1)(20,-1)(20,-3)(5,-3)(5,-5)(30,-5)
  \node(x0)(23,-2){$\rho_0$}
  \drawline[linecolor=red](30,-5)(50,-5)(50,-7)(35,-7)(35,-9)(45,-9)(45,-11)(30,-11)
  \node(x0)(53,-6){\color{red}$\rho_1$}
  \drawline[linecolor=red](30,-11)(20,-11)(20,-13)(25,-13)(25,-15)(18,-15)(18,-17)(30,-17)
  \node(x1)(17,-12){\color{red}$\rho_2$}
  \drawline[linecolor=red](30,-17)(40,-17)(40,-19)(30,-19)
  \node(x2)(43,-18){\color{red}$\rho_3$}
  \drawline[linecolor=red](30,-19)(15,-19)(15,-21)(22,-21)(22,-23)(10,-23)(10,-25)(30,-25)
  \node(x3)(12,-20){\color{red}$\rho_4$}
  \drawline(30,-25)(50,-25)(50,-27)(40,-27)(40,-29)(60,-29)
  \node(x2)(53,-26){$\rho_5$}
  \gasset{Nw=1.6,Nh=1.6,Nfill=y,ExtNL=y,NLdist=.7}
  \node[NLangle=45](p)(0,-1){$p$}
  \node[fillcolor=red,NLangle=45](p1)(30,-5){\color{red}$p_1$}
  \node[fillcolor=red,NLangle=135](p2)(30,-11){\color{red}$p_2$}
  \node[fillcolor=red,NLangle=45](p3)(30,-17){\color{red}$p_3$}
  \node[fillcolor=red,NLangle=-135](p4)(30,-19){\color{red}$p_4$}
  \node[fillcolor=red,NLangle=45](p5)(30,-25){\color{red}$p_5$}
  \node[NLangle=0](q)(60,-29){$q$}
\end{gpicture}
\begin{gpicture}[name=gpic:LL-concat,ignore]
  \gasset{Nframe=n}
  \node(ui)(15,3){$u$}
  \node(ui+1)(45,3){$v$}
  \gasset{AHnb=0}
  \drawline(0,2)(0,-31)
  \drawline(30,2)(30,-31)
  \drawline(60,2)(60,-31)
  \gasset{AHnb=1,arcradius=0.8}
  \drawline(0,-1)(20,-1)(20,-3)(5,-3)(5,-5)(30,-5)
  \node(x0)(23,-2){$\rho_0$}
  \drawline[linecolor=red](30,-5)(50,-5)(50,-7)(35,-7)(35,-9)(45,-9)(45,-11)(30,-11)
  \node(x0)(53,-6){\color{red}$\rho_1$}
  \drawline[linecolor=red](30,-11)(20,-11)(20,-13)(25,-13)(25,-15)(18,-15)(18,-17)(30,-17)
  \node(x1)(17,-12){\color{red}$\rho_2$}
  \drawline[linecolor=red](30,-17)(40,-17)(40,-19)(30,-19)
  \node(x2)(43,-18){\color{red}$\rho_3$}
  \drawline[linecolor=black](30,-19)(15,-19)(15,-21)(22,-21)(22,-23)(10,-23)(10,-25)(15,-25)
  (15,-27)(0,-27)
  \node(x3)(12,-20){\color{black}$\rho_4$}
  \gasset{Nw=1.6,Nh=1.6,Nfill=y,ExtNL=y,NLdist=.7}
  \node[NLangle=45](p)(0,-1){$p$}
  \node[fillcolor=red,NLangle=45](p1)(30,-5){\color{red}$p_1$}
  \node[fillcolor=red,NLangle=135](p2)(30,-11){\color{red}$p_2$}
  \node[fillcolor=red,NLangle=45](p3)(30,-17){\color{red}$p_3$}
  \node[fillcolor=red,NLangle=-135](p4)(30,-19){\color{red}$p_4$}
  \node[NLangle=180](q)(0,-27){$q$}
\end{gpicture}
\begin{gpicture}[name=gpic:LR-run,ignore]
  \gasset{Nframe=n}
  \node(u1)(0,2){$u_1$}
  \node(u2)(10,2){$u_2$}
  \node(ui)(30,2){$u_i$}
  \node(ui+1)(40,2){$u_{i+1}$}
  \node(ui+k)(70,2){$u_{i+k}$}
  \node(ui+k+1)(82,2){$u_{i+k+1}$}
  \node(un)(110,2){$u_n$}
  \gasset{AHnb=0}
  \drawline(-5,0)(-5,-35)
  \drawline(35,0)(35,-35)
  \drawline(75,0)(75,-35)
  \drawline(115,0)(115,-35)
  \gasset{AHnb=1,arcradius=0.8}
  \drawline(-5,-2)(20,-2)(20,-4)(10,-4)(10,-6)(40,-6)(40,-8)
  (30,-8)(30,-10)(34,-10)
  \node(rho1)(15,-8){$\rho_1$}
  \drawline[linecolor=red](35,-10)(60,-10)(60,-12)(50,-12)(50,-14)(74,-14)
  \node(rho2)(55,-8){\color{red}$\rho_2$}
  \drawline[linecolor=blue](75,-14)(95,-14)(95,-16)(65,-16)
  (65,-18)(85,-18)(85,-20)(40,-20)(40,-22)(50,-22)(50,-24)(36,-24)
  \node(rho3)(60,-22){\color{blue}$\rho_3$}
  \drawline[linecolor=black](35,-24)(25,-24)(25,-26)(65,-26)(65,-28)(55,-28)
  (55,-30)(85,-30)(85,-32)(65,-32)(65,-34)(114,-34)
  \node(rho4)(50,-28){\color{black}$\rho_4$}
  \gasset{Nw=1.6,Nh=1.6,Nfill=y,ExtNL=y,NLdist=.7}
  \node[NLangle=180](p)(-5,-2){$p$}
  \node[NLangle=-45](q1)(35,-10){$q_1$}
  \node[NLangle=45](q2)(75,-14){$q_2$}
  \node[NLangle=135](q3)(35,-24){$q_3$}
  \node[NLangle=0](q)(115,-34){$q$}
\end{gpicture}
\begin{gpicture}[name=gpic:RL-run,ignore]
  \gasset{Nframe=n}
  \node(u1)(0,2){$u_1$}
  \node(u2)(10,2){$u_2$}
  \node(ui)(30,2){$u_i$}
  \node(ui+1)(40,2){$u_{i+1}$}
  \node(ui+k)(70,2){$u_{i+k}$}
  \node(ui+k+1)(82,2){$u_{i+k+1}$}
  \node(un)(110,2){$u_n$}
  \gasset{AHnb=0}
  \drawline(-5,0)(-5,-35)
  \drawline(35,0)(35,-35)
  \drawline(75,0)(75,-35)
  \drawline(115,0)(115,-35)
  \gasset{AHnb=1,arcradius=0.8}
  \drawline(115,-2)(90,-2)(90,-4)(100,-4)(100,-6)
  (50,-6)(50,-8)(60,-8)(60,-10)(20,-10)(20,-12)(40,-12)(40,-14)
  (30,-14)(30,-16)(34,-16)
  \node(rho1)(95,-8){$\rho_1$}
  \drawline[linecolor=red](35,-16)(65,-16)(65,-18)(55,-18)(55,-20)(74,-20)
  \node(rho2)(50,-18){\color{red}$\rho_2$}
  \drawline[linecolor=blue](75,-20)(95,-20)(95,-22)(65,-22)(65,-24)
  (85,-24)(85,-26)(40,-26)(40,-28)(50,-28)(50,-30)(36,-30)
  \node(rho3)(60,-28){\color{blue}$\rho_3$}
  \drawline[linecolor=black](35,-30)(15,-30)(15,-32)(45,-32)(45,-34)(-4,-34)
  \node(rho4)(10,-32){\color{black}$\rho_4$}
  \gasset{Nw=1.6,Nh=1.6,Nfill=y,ExtNL=y,NLdist=.7}
  \node[NLangle=0](p)(115,-2){$p$}
  \node[NLangle=-45](q1)(35,-16){$q_1$}
  \node[NLangle=45](q2)(75,-20){$q_2$}
  \node[NLangle=135](q3)(35,-30){$q_3$}
  \node[NLangle=180](q)(-5,-34){$q$}
\end{gpicture}
\begin{gpicture}[name=gpic:RR-run,ignore]
  \gasset{Nframe=n}
  \node(u1)(0,2){$u_1$}
  \node(u2)(10,2){$u_2$}
  \node(ui)(30,2){$u_i$}
  \node(ui+1)(40,2){$u_{i+1}$}
  \node(ui+k)(70,2){$u_{i+k}$}
  \node(ui+k+1)(82,2){$u_{i+k+1}$}
  \node(un)(110,2){$u_n$}
  \gasset{AHnb=0}
  \drawline(-5,0)(-5,-35)
  \drawline(25,0)(25,-35)
  \drawline(35,0)(35,-35)
  \drawline(75,0)(75,-35)
  \drawline(115,0)(115,-35)
  \gasset{AHnb=1,arcradius=0.8}
  \drawline(115,-2)(90,-2)(90,-4)(100,-4)(100,-6)
  (50,-6)(50,-8)(60,-8)(60,-10)(36,-10)
  \node(rho1)(90,-8){$\rho_1$}
  \drawline[linecolor=red](35,-10)(30,-10)(30,-12)(45,-12)(45,-14)
  (32,-14)(32,-16)(65,-16)(65,-18)(27,-18)(27,-20)(34,-20)
  \node(rho2)(50,-13){\color{red}$\rho_2$}
  \drawline[linecolor=blue](35,-20)(50,-20)(50,-22)(40,-22)(40,-24)
  (80,-24)(80,-26)(70,-26)(70,-28)(105,-28)(105,-30)(95,-30)(95,-32)(114,-32)
  \node(rho3)(90,-26){\color{blue}$\rho_3$}
  \gasset{Nw=1.6,Nh=1.6,Nfill=y,ExtNL=y,NLdist=.7}
  \node[NLangle=0](p)(115,-2){$p$}
  \node[NLangle=135](q1)(35,-10){$q_1$}
  \node[NLangle=-45](q2)(35,-20){$q_2$}
  \node[NLangle=0](q)(115,-32){$q$}
\end{gpicture}
\begin{gpicture}[name=gpic:LR-run2,ignore]
  \gasset{Nframe=n}
  \node(u1)(0,2){$u_1$}
  \node(u2)(10,2){$u_2$}
  \node(u3)(20,2){$u_3$}
  \node(u4)(30,2){$u_4$}
  \node(u5)(40,2){$u_5$}
  \node(u6)(50,2){$u_6$}
  \node(u7)(60,2){$u_7$}
  \node(u8)(70,2){$u_8$}
  \node(u9)(80,2){$u_9$}
  \node(u10)(90,2){$u_{10}$}
  \node(u11)(100,2){$u_{11}$}
  \node(u12)(110,2){$u_{12}$}
  \gasset{AHnb=0}
  \drawline(-5,0)(-5,-35)
  \drawline( 5,-1)( 5,-35)
  \drawline(15,0)(15,-35)
  \drawline(25,0)(25,-35)
  \drawline(35,0)(35,-35)
  \drawline(45,0)(45,-35)
  \drawline(55,0)(55,-35)
  \drawline(65,0)(65,-35)
  \drawline(75,0)(75,-35)
  \drawline(85,0)(85,-35)
  \drawline(95,0)(95,-35)
  \drawline(105,0)(105,-35)
  \drawline(115,0)(115,-35)
  \gasset{AHnb=1,arcradius=0.8}
  \drawline(-5,-2)(14,-2)
  \node(rho0)(3,-0){\color{black}$\rho_0$}
  \drawline[linecolor=red](15,-2)(20,-2)(20,-4)(0,-4)(0,-6)(24,-6)
  \node(rho1)(22,-1){\color{red}$\rho_1$}
  \drawline[linecolor=blue](25,-6)(34,-6)
  \node(rho2)(31,-4){\color{blue}$\rho_2$}
  \drawline(35,-6)(40,-6)(40,-8)(20,-8)(20,-10)(40,-10)(40,-12)(30,-12)(30,-14)(44,-14)
  \node(rho3)(42,-5){\color{black}$\rho_3$}
  \drawline[linecolor=red](45,-14)(50,-14)(50,-16)(40,-16)(40,-18)(54,-18)
  \node(rho4)(52,-13){\color{red}$\rho_4$}
  \drawline[linecolor=blue](55,-18)(64,-18)
  \node(rho5)(61,-16){\color{blue}$\rho_5$}
  \drawline[linecolor=black](65,-18)(74,-18)
  \node(rho6)(71,-16){\color{black}$\rho_6$}
  \drawline[linecolor=red](75,-18)(80,-18)(80,-20)(60,-20)(60,-22)(70,-22)(70,-24)
  (60,-24)(60,-26)(84,-26)
  \node(rho7)(82,-17){\color{red}$\rho_7$}
  \drawline[linecolor=blue](85,-26)(90,-26)(90,-28)(70,-28)(70,-30)(94,-30)
  \node(rho8)(92,-25){\color{blue}$\rho_8$}
  \drawline[linecolor=black](95,-30)(104,-30)
  \node(rho9)(101,-28){\color{black}$\rho_9$}
  \drawline[linecolor=red](105,-30)(110,-30)(110,-32)(100,-32)(100,-34)(114,-34)
  \node(rho10)(112,-29){\color{red}$\rho_{10}$}
  \gasset{Nw=1.6,Nh=1.6,Nfill=y,ExtNL=y,NLdist=.5}
  \node[NLangle=180](p)(-5,-2){$p$}
  \node[NLangle=45](q)(15,-2){$q$}
  \node[NLangle=-45](q1)(15,-6){\color{red}$q_1$}
  \node[NLangle=45](q)(25,-6){$q$}
  \node[NLangle=45](q)(35,-6){$q$}
  \node[NLangle=-45](q3)(35,-14){$q_3$}
  \node[NLangle=45](q)(45,-14){$q$}
  \node[NLangle=-45](q4)(45,-18){\color{red}$q_4$}
  \node[NLangle=45](q)(55,-18){$q$}
  \node[NLangle=45](q)(65,-18){$q$}
  \node[NLangle=45](q)(75,-18){$q$}
  \node[NLangle=45](q7)(75,-26){\color{red}$q_7$}
  \node[NLangle=45](q)(85,-26){$q$}
  \node[NLangle=-45](q8)(85,-30){\color{blue}$q_8$}
  \node[NLangle=45](q)(95,-30){$q$}
  \node[NLangle=45](q)(105,-30){$q$}
  \node[NLangle=-45](q10)(105,-34){\color{red}$q_{10}$}
  \node[NLangle=0](q)(115,-34){$q$}
  \node[NLangle=-90,Nfill=n,NLdist=1](q2)(35,-34){$\textcolor{red}{q_2}=\textcolor{blue}{q_5}=q_6=q_9=q$}
\end{gpicture}
\begin{gpicture}[name=gpic:RL-run2,ignore]
  \gasset{Nframe=n}
  \node(u1)(0,4){$u_1$}
  \node(u2)(10,4){$u_2$}
  \node(u3)(20,4){$u_3$}
  \node(u4)(30,4){$u_4$}
  \node(u5)(40,4){$u_5$}
  \node(u6)(50,4){$u_6$}
  \node(u7)(60,4){$u_7$}
  \node(u8)(70,4){$u_8$}
  \node(u9)(80,4){$u_9$}
  \node(u10)(90,4){$u_{10}$}
  \node(u11)(100,4){$u_{11}$}
  \node(u12)(110,4){$u_{12}$}
  \gasset{AHnb=0}
  \drawline(-5,2)(-5,-35)
  \drawline( 5,2)( 5,-35)
  \drawline(15,2)(15,-35)
  \drawline(25,2)(25,-35)
  \drawline(35,2)(35,-35)
  \drawline(45,2)(45,-35)
  \drawline(55,2)(55,-35)
  \drawline(65,2)(65,-35)
  \drawline(75,2)(75,-35)
  \drawline(85,2)(85,-35)
  \drawline(95,2)(95,-35)
  \drawline(105,2)(105,-35)
  \drawline(115,2)(115,-35)
  \gasset{AHnb=1,arcradius=0.8}
  \drawline(115,-2)(96,-2)
  \node(rho0)(109,-0){\color{black}$\rho_{11}$}
  \drawline[linecolor=red](95,-2)(90,-2)(90,-4)(110,-4)(110,-6)(86,-6)
  \node(rho1)(88,-1){\color{red}$\rho_{10}$}
  \drawline[linecolor=blue](85,-6)(76,-6)
  \node(rho2)(79,-4){\color{blue}$\rho_9$}
  \drawline(75,-6)(70,-6)(70,-8)(90,-8)(90,-10)(70,-10)(70,-12)(80,-12)(80,-14)(66,-14)
  \node(rho3)(68,-5){\color{black}$\rho_8$}
  \drawline[linecolor=red](65,-14)(60,-14)(60,-16)(70,-16)(70,-18)(56,-18)
  \node(rho4)(58,-13){\color{red}$\rho_7$}
  \drawline[linecolor=blue](55,-18)(46,-18)
  \node(rho5)(49,-16){\color{blue}$\rho_6$}
  \drawline[linecolor=black](45,-18)(36,-18)
  \node(rho6)(39,-16){\color{black}$\rho_5$}
  \drawline[linecolor=red](35,-18)(30,-18)(30,-20)(50,-20)(50,-22)(40,-22)(40,-24)
  (50,-24)(50,-26)(26,-26)
  \node(rho7)(28,-17){\color{red}$\rho_4$}
  \drawline[linecolor=blue](25,-26)(20,-26)(20,-28)(40,-28)(40,-30)(16,-30)
  \node(rho8)(18,-25){\color{blue}$\rho_3$}
  \drawline[linecolor=black](15,-30)(6,-30)
  \node(rho9)(9,-28){\color{black}$\rho_2$}
  \drawline[linecolor=red](5,-30)(0,-30)(0,-32)(10,-32)(10,-34)(-4,-34)
  \node(rho10)(-2,-29){\color{red}$\rho_{1}$}
  \gasset{Nw=1.6,Nh=1.6,Nfill=y,ExtNL=y,NLdist=.3}
  \node[NLangle=0](p)(115,-2){$p$}
  \node[NLangle=135](q)(95,-2){\color{red}$q_{10}$}
  \node[NLangle=-45](qp)(95,-6){\color{red}$q'_{10}$}
  \node[NLangle=135](q)(85,-6){\color{blue}$q_{9}$}
  \node[NLangle=135](q)(75,-6){$q_8$}
  \node[NLangle=-45](qp)(75,-14){$q'_{8}$}
  \node[NLangle=135](q)(65,-14){\color{red}$q_7$}
  \node[NLangle=-45](qp)(65,-18){\color{red}$q'_{7}$}
  \node[NLangle=135](q)(55,-18){\color{blue}$q_6$}
  \node[NLangle=-45](qp)(55,-18){\color{blue}$q'_{6}$}
  \node[NLangle=135](q)(45,-18){$q_5$}
  \node[NLangle=135](q)(35,-18){\color{red}$q_4$}
  \node[NLangle=135](qp)(35,-26){\color{red}$q'_{4}$}
  \node[NLangle=135](q)(25,-26){\color{blue}$q_3$}
  \node[NLangle=-45](qp)(25,-30){\color{blue}$q'_{3}$}
  \node[NLangle=135](q)(15,-30){$q_2$}
  \node[NLangle=-45](qp)(15,-30){$q'_{2}$}
  \node[NLangle=135](q)(5,-30){\color{red}$q_1$}
  \node[NLangle=-45](qp)(5,-34){\color{red}$q'_{1}$}
  \node[NLangle=180](q)(-5,-34){$q$}
  \node[NLangle=-90,Nfill=n,NLdist=1](q5)(45,-34){$\textcolor{black}{q'_5=q_5}$}
  \node[NLangle=-90,Nfill=n,NLdist=1](q9)(85,-34){$\textcolor{blue}{q'_9=q_9}$}
\end{gpicture}
\begin{gpicture}[name=KleSt-f2,ignore]
  \gasset{Nframe=n,Nw=5,Nh=5}
 \node(ihash0)(0,30){$\#$}
 \node(iu1)(10,30){$u_1$}
 \node(ihash1)(20,30){$\#$}
 \node(iu2)(30,30){$u_2$}
 \node(ihash2)(40,30){$\#$}
 \node(iu3)(50,30){$u_3$}
 \node(ihash3)(60,30){$\#$} 
 \node(iu4)(70,30){$u_4$}
  \node(ihash4)(80,30){$\#$}
 \node(iu5)(90,30){$u_5$}
  \node(ihash5)(100,30){$\#$}
  
 \node(ohash0)(0,20){$\#$}
 \node(ou1)(10,20){$u_1$}
 \node(ohash1)(20,20){$\#$}
 \node(ou2)(30,20){$u_2$}
 \node(ohash2)(40,20){$\#$}
 \node(ou3)(50,20){$u_3$}
 \node(ohash3)(60,20){$\#$} 
 \node(o2u2)(30,10){$u_2$}
 \node(o2u3)(50,10){$u_3$}
 \node(o2u4)(70,10){$u_4$}

 \node(o3u3)(50,0){$u_3$}
 \node(o3u4)(70,0){$u_4$}
  \node(o3u5)(90,0){$u_5$}
  
 \node(input)(-30,30){input :}
 \node(c1)(-30,20){output copy 1:}
 \node(c2)(-30,10){output copy 2:}
 \node(c3)(-30,0){output copy 3:}
 \drawedge(ohash0,ou1){}
 \drawedge(ou1,o2u2){}
 \drawedge(o2u2,o3u3){}
 \drawedge(o3u3,ohash1){}
 \drawedge(ohash1,ou2){}
 \drawedge(ou2,o2u3){}
 \drawedge(o2u3,o3u4){}
 \drawedge(o3u4,ohash2){}
 \drawedge(ohash2,ou3){}
 \drawedge(ou3,o2u4){}
 \drawedge(o2u4,o3u5){}
 \drawedge(o3u5,ohash3){}
\end{gpicture}
\begin{gpicture}[name=SSTap,ignore]
  \gasset{Nframe=n,Nw=5,Nh=5}
 \node(iu1)(10,30){$u_1$}
 \node(iu2)(30,30){$u_2$}
 \node(iu3)(50,30){$u_3$}
 \node(iu4)(70,30){$u_4$}
 \node(iu5)(90,30){$u_5$}
  
 {\gasset{Nw=1,ExtNL=y,NLangle=180,NLdist=1}
 \node(input)(-10,30){input:}
 \node(x1)(-10,20){$X_1$:}
 \node(x2)(-10,10){$X_2$:}
 \node(x3)(-10,0){$X_3$:}
 \node(varO)(-10,-10){$O$:}
 }
 
 \node(pu11)(2,20){$.$}
 \node(su11)(18,20){$.$}
 \node(pu21)(22,20){$.$}
 \node(su21)(38,20){$.$}
 \node(pu31)(42,20){$.$}
 \node(su31)(58,20){$.$}
 \node(pu41)(62,20){$.$}
 \node(su41)(78,20){$.$}
 \node(pu51)(82,20){$.$}
 \node(su51)(98,20){$.$}
 
 \node(pu12)(2,10){$.$}
 \node(su12)(18,10){$.$}
 \node(pu22)(22,10){$.$}
 \node(su22)(38,10){$.$}
 \node(pu32)(42,10){$.$}
 \node(su32)(58,10){$.$}
 \node(pu42)(62,10){$.$}
 \node(su42)(78,10){$.$}
 \node(pu52)(82,10){$.$}
 \node(su52)(98,10){$.$}

 \node(pu13)(2,0){$.$}
 \node(su13)(18,0){$\times$}
 \node(pu23)(22,0){$.$}
 \node(su23)(38,0){$\times$}
 \node(pu33)(42,0){$.$}
 \node(su33)(58,0){$.$}
 \node(pu43)(62,0){$.$}
 \node(su43)(78,0){$.$}
 \node(pu53)(82,0){$.$}
 \node(su53)(98,0){$.$}
 
 \node(in)(-8,-10){$.$}
 \node(o0)(2,-10){$.$}
 \node(so3)(62,-10){$.$}
 \node(so4)(82,-10){$.$}
 \node(so5)(102,-10){$.$}

\drawedge(pu11,su11){$u_1$}
\drawedge(pu12,su12){$u_1$} 
\drawedge(pu13,su13){$u_1$} 
\drawedge(pu21,su21){$u_2$}
\drawedge(pu22,su22){$u_2$} 
\drawedge(pu23,su23){$u_2$} 
\drawedge(pu31,su31){$u_3$}
\drawedge(pu32,su32){$u_3$} 
\drawedge(pu33,su33){$u_3$} 
\drawedge(pu41,su41){$u_4$}
\drawedge(pu42,su42){$u_4$} 
\drawedge(pu43,su43){$u_4$} 
\drawedge(pu51,su51){$u_5$}
\drawedge(pu52,su52){$u_5$} 
\drawedge(pu53,su53){$u_5$} 

\drawedge(su11,pu22){}
\drawedge(su12,pu23){}
\drawedge(su21,pu32){}
\drawedge(su22,pu33){}
\drawedge(su31,pu42){}
\drawedge(su32,pu43){}
\drawedge(su41,pu52){}
\drawedge(su42,pu53){}
\drawedge(su33,so3){$\#$}
\drawedge(su43,so4){$\#$}
\drawedge(su53,so5){$\#$}
\drawedge(in,o0){$\#$}
\drawedge(o0,so3){}
\drawedge(so3,so4){}
\drawedge(so4,so5){}
\end{gpicture}

\subsection{Can the 2-chained Kleene star suffice for all aperiodic functions?} 
\label{sec:inexp}
It is known \cite{DGK-lics18} that the 2-chained Kleene star can simulate the 
$k$-chained Kleene-star for regular functions. However,  
we believe that, contrary to the case of regular functions, the $k$-chained
Kleene-star operator cannot be simulated by the $2$-chained Kleene-star
while preserving the aperiodicity of the expression. 
 The key idea is that, in order to simulate a $k$-chained Kleene-star on a SD 
prefix code $L$ using a $2$-chained Kleene-star, one needs to use $L^{\lceil k/2
\rceil}$ as a parser. 
However, for any given prefix code $L$, the language $L^n$ for $n>1$, while
still a prefix code, is not of bounded synchronisation delay (for the same
reason that $\{aa\}$ is not, i.e., for $v=(aa)^d$ that we consider, $a v
a$ belongs to $(aa)^*$ but $av$ does not). Intuitively, 
parsing $L^n$ reduces to  \emph{count}ing factors of $L$ \emph{modulo $n$}, which
is a classical example of non-aperiodicity.

As an example, consider the prefix code $L=(a+b)^{*}c$ which has  
synchronisation delay $1$.  Define a function $f$ with domain $L^{3}$ by
$f(u_1u_2u_3)=u_3u_1$ when $u_1,u_2,u_3\in L$, which can be written using combinators as 
$\big( (\SimpleFun{L^2}{\epsilon}) \cdot (\SimpleFun{L}{id}) \big) \odot 
\big((\SimpleFun{L}{id}) \cdot (\SimpleFun{L^2}{\epsilon})\big)$.
The identity function $id$ can itself be written as
$(\SimpleFun{a}{a}+\SimpleFun{b}{b}+\SimpleFun{c}{c})^{\star}$ (see also Figure \ref{fig:eg-2dft}, which is a simplification of the same function, but neverthless has the same 
inexpressiveness with 2 chained star). 
Then we believe that the function $\kstar{3}{L}{f}$, which associates to a word
$u_1\cdots u_n\in L^{*}$ the word $u_3u_1 u_4u_2\cdots u_nu_{n-2}$ is not
definable using only $2$-chained Kleene-stars.
While not a proof, the intuition behind this is that, in order to construct
$u_{i+1}u_{i-1}$, we need to highlight words from $L^3$.
In order to do this with a $2$-chained Kleene-star, it seems necessary to apply a
chained star with parser $L^2$, which is a prefix code but not of bounded
synchronisation delay.
A similar argument would hold for any $\kstar{k}{L}{f}$, $k\geq 3$ with a function
$f(u_1u_2\cdots u_k)=u_ku_1$.

\section{The Equivalence of SDRTE and Aperiodic 2DFT}\label{sec:Equiv}

In this section, we prove the main result of the paper, namely the equivalence
between \SDRTE and aperiodic 2DFT stated in Theorem~\ref{thm:intro}. The first direction, given an \SDRTE $C$,
constructing an equivalent aperiodic 2DFT $\A$ is given by
Theorem~\ref{thm:onedir}, while Theorem~\ref{thm:main} handles the converse.

\subsection{Aperiodic 2DFTs for SD-regular transducer expressions }\label{sec:A2DFTtoSDRET}

\begin{theorem}\label{thm:SDRTEtoAperiodic2DFT}
  Given an \SDRTE $C$, we can construct an equivalent aperiodic 2DFT $\A$ with 
  $\sem{C}=\sem{\A}$.
  \label{thm:onedir}
\end{theorem}

\begin{proof}
We construct $\A$ by induction on the structure of the \SDRTE $C$.
In the suitable cases, we will suppose thanks to induction that we have
aperiodic transducers $\A_i$ for expressions $C_i$, $i\leq 2$.  We also have a 
deterministic and complete
aperiodic automaton $A_L$ for any aperiodic language $L$.
\begin{itemize}
  \item $C=\bot$.  Then $\A$ is a single state transducer with no final state so that its
  domain is empty.

  \item $C=v$.  Then $\A$ is a single state transducer which produces $v$ and
  accepts any input word.  Clearly, $\A$ is aperiodic.

  \item $C= \Ifthenelse{L}{C_1}{C_2}$. 
  The transducer $\A$ first reads its input, simulating $A_L$. Upon reaching the end of the input word, it goes back to the beginning of the word, and either executes $\A_1$ if the word was accepted by $A_L$, or executes $\A_2$ otherwise. Since every machine was aperiodic, so is $\A$.

  \item $C=C_1\odot C_2$.
  The transducer $\A$ does a first pass executing $\A_1$, then resets to the beginning of the word and simulates $\A_2$. Since both transducers are aperiodic, so is $\A$.

  \item $C=C_1\cdot C_2$.
  We express $\A$ as the composition of three functions $f_1,f_2,f_3$, each
  aperiodic.  Since aperiodic functions are closed under composition, we get the
  result.
  The first function $f_1$ associates to each word $w\in\Sigma^{*}$ the word $u_1\# u_2
  \#\cdots \# u_n$, such that $w=u_1u_2 \cdots u_n$ and for any prefix $u$ of $w$, $u$
  belongs to the domain of $C_1$ if, and only if, $u=u_1\cdots u_i$ for some
  $1\leq i<n$.
  Notice that $u_1=\varepsilon$ iff $\varepsilon\in\dom{C_1}$ and 
  $u_n=\varepsilon$ iff $w\in\dom{C_1}$. The other $u_i$'s must be nonempty.
  The second function $f_2$ takes as input a word in $(\Sigma\cup\{\#\})^{*}$,
  reads it from right to left, and suppresses all $\#$ symbols except for the ones
  whose corresponding suffix belongs to the domain of $C_2$.  
  Then, $f_2(f_1(w))$ contains exactly one $\#$ symbol if and only if $w$ has a 
  unique factorisation $w=uv$ with $u\in\dom{C_1}$ and $v\in\dom{C_2}$. In this 
  case, $f_2(f_1(w))=u\#v$.

  Finally, the function $f_3$ has domain $\Sigma^{*}\#\Sigma^{*}$ and first executes
  $\A_1$ on the prefix of its input upto the $\#$ symbol, treating it as the right
  endmarker $\rightend$, and then executes $\A_2$ on the second part, treating $\#$ as the
  left endmarker $\leftend$.
  
  The functions $f_1$ and $f_2$ can be realised by aperiodic transducers as they
  only simulate automata for the aperiodic domains of $C_1$ and the reverse of
  $C_2$ respectively, and the function $f_3$ executes $\A_1$ and $\A_2$ one after
  the other, and hence is also aperiodic.

  \item $C=\kstar{k}{L}{C_1}$ or $C=\krstar{k}{L}{C_1}$.
  Here $L\subseteq\Sigma^*$ is an aperiodic language which is also a prefix code
  with bounded synchronisation delay, and $k\geq1$ is a natural number. Let
  $f=\sem{C_1}\colon\Sigma^*\to\D$ be the aperiodic
  function defined by $C_1$.
  We write $\kstar{k}{L}{f}=
  \Ifthenelse{L^{<k}}{(\SimpleFun{\Sigma^*}{\epsilon})}{(f_3\circ f_2\circ f_1)}$, 
  where $\epsilon$ is the output produced when the input has less than $k$ $L$ factors;
  otherwise the output is produced by the composition of 3 aperiodic functions.  As
  aperiodic functions are closed under composition, this gives the result.  The first one,
  $f_1\colon\Sigma^{*}\to(\Sigma\cup\{\#\})^{*}$ splits an input word $w\in
  L^{*}$ according to the unique factorization $w=u_1 u_2\cdots u_n$ with
  $n\geq0$ and $u_i\in L$ for all $1\leq i\leq n$ and inserts $\#$ symbols:
  $f_1(w)=\#u_1\#u_2\#\cdots\#u_n\#$. The domain of $f_1$ is $L^{*}$.

  The second function $f_2$ constructs the sequence of $k$ factors.  Its domain
  is $\#(\Sigma^{*}\#)^{\geq k}$ and it is defined as follows, with
  $u_i\in\Sigma^{*}$:
  $$
  f_2(\#u_1\#u_2\#\cdots\#u_n\#)=\#u_1u_2\cdots u_k\#u_2\cdots u_{k+1}\#\cdots\#
  u_{n-k+1}\cdots u_n\# \,.
  $$
  
  Finally, the third function simply applies $f$ and erases the $\#$ symbols:
  $$
  f_3(\#v_1\#v_2\#\cdots\#v_m\#)=f(v_1) f(v_2) \cdots f(v_m) \,.
  $$
  In particular, $f_3(\#)=\varepsilon$.
  We have $\dom{f_3}=\#(\dom{f}\#)^{*}$.
  
  For the reverse iteration, we simply change the last function and use instead
  $$
  f_4(\#v_1\#v_2\#\cdots\#v_m\#)=f(v_m)\cdots f(v_2) f(v_1) \,.
  $$
  Lemma~\ref{lem-fisAperiodic2} below proves that the functions $f_i$ for $i\leq 
  4$ are aperiodic.
  \qedhere
\end{itemize}
\end{proof}

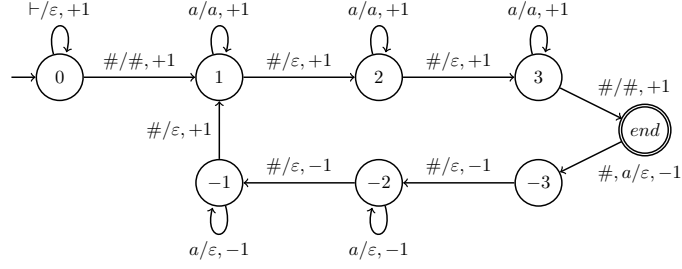
\begin{figure} [t]                                                                                                                                                           
 \begin{center}                                                                                                                                                         
    \scalebox{0.7}{                                                                                                                                                       
    \begin{tikzpicture}[->,thick]   
    
	\node[initial, state, initial text={}] (s0) at (0,0) {$0$};      
 	\node[state] (s1) at (3,0) {$1$};         
 	\node[state] (s2) at (6,0) {$2$}; 
 	\node[state] (s3) at (9,0) {$3$}; 
 	\node[state,] (sm3) at (9,-2) {$-3$};
 	\node[state,accepting] (end) at ((11,-1) {$end$};
 	\node[state] (sm1) at (3,-2) {$-1$}; 
 	\node[state] (sm2) at (6,-2) {$-2$};  	 	

    \path (s0) edge[loop above] node [above]{$\leftend/ \epsilon, +1$}node {}(); 	 	
    \path (s1) edge[loop above] node [above]{$a/ a, +1$}node {}();
    \path (s2) edge[loop above] node [above]{$a/ a, +1$}node {}();
    \path (s3) edge[loop above] node [above]{$a/ a, +1$}node {}();

    \path (sm1) edge[loop below] node [below]{$a/\epsilon, -1$}node {}();
    \path (sm2) edge[loop below] node [below]{$a/\epsilon, -1$}node {}();
    
    \path (s0) edge node [above]{$\#/\#, +1$}node {}(s1);
   \path (s1) edge node [above]{$\#/\epsilon, +1$}node {}(s2);
   \path (s2) edge node [above]{$\#/\epsilon, +1$}node {}(s3);
   \path (s3) edge node [above right]{$\#/\#, +1$}node {}(end);
   \path (end) edge node [below right]{$\#,a/\epsilon, -1$}node {}(sm3);
   \path (sm3) edge node [above]{$\#/\epsilon, -1$}node {}(sm2);
   \path (sm2) edge node [above]{$\#/\epsilon, -1$}node {}(sm1);
   \path (sm1) edge node [left]{$\#/\epsilon, +1$}node {}(s1);
    \end{tikzpicture}
    } 
  \end{center}                                                                                                                                                         
  \caption{The transducer $T_2$ for $k=3$.}
  \label{fig:SDto2DFT:f2}                                                                                                                                                        
\end{figure}

\begin{lemma}\label{lem-fisAperiodic2}
  The functions $f_1,f_2,f_3,f_4$ are realised by aperiodic 2DFTs.
\end{lemma}

\begin{proof}
  \textbf{The function $f_1$.}
  First, since $L$ is an aperiodic language which is a prefix code with
  bounded synchronisation delay, $L^*$ is  aperiodic.  
  Let $\A_1$ be an aperiodic deterministic automaton that recognizes
  $L^*$.  Let $w$ be a word in $L^*$ and $w=u_1\cdots u_n$ with $u_i \in L$.
  Since $L$ is a code, this decomposition is unique.  Notice that
  $\varepsilon\notin L$.  We claim that the run of $\A_1$ over $w$ reaches
  final states exactly at the end of each $u_i$.  Should this hold, then we
  can easily construct a (one-way) aperiodic transducer $T_1$ realising $f_1$ by simply
  simulating $\A_1$ and copying its input, adding $\#$ symbols each time $\A_1$
  reaches a final state.
  
  It remains to prove the claim.  First, since for any $1\leq i\leq n$,
  $u_1\cdots u_i$ belongs to $L^*$, $\A_1$ reaches a final state after reading
  $u_i$.  Conversely, suppose $\A_1$ reaches a final state after reading some
  nonempty prefix $v$ of $w$.  Then $v$ can be written $u_1\cdots u_i u'$ for
  some index $0\leq i<n$ and some nonempty prefix $u'$ of $u_{i+1}$.  But
  since $\A_1$ reaches a final state on $v$, we have $v\in L^*$.  Hence, there
  is a unique decomposition $v=v_1\cdots v_m$ with $v_j \in L$.  Since
  $v=u_1\cdots u_i u' = v_1 \cdots v_m$, either $u_1$ is a prefix of $v_1$ or
  conversely.  Since $L$ is a prefix code, and both $u_1$ and $v_1$ belong to $L$,
  we obtain $u_1=v_1$.  By induction, we get that $u_j=v_j$ for $j\leq i$.
  Now, $u'=v_{i+1}\cdots v_m$ is a nonempty prefix of $u_{i+1}$.  Using again
  that $L$ is a prefix code, we get $m=i+1$ and $u'=v_{i+1}=u_{i+1}$,
  which concludes the proof of the claim.
  
  \textbf{The function $f_2$.}    
  The domain of $f_2$ is the language $K=\#(\Sigma^{*}\#)^{\geq k}$.
  We construct an aperiodic 2DFT $T_2$ for $f_2$  (see
  Figure~\ref{fig:SDto2DFT:f2} for $T_2$ where $k=3$).
  Let $T_2=(\{-k,-k+1,\ldots,0,\ldots,k-1,k\}\cup\{end\},\Sigma\cup\{\#\},\Sigma\cup\{\#\},
  \delta_2,\gamma_2,0,\{end\})$ be the 2DFT realising $f_2$.
  The transition function $\delta_2$ is defined as:
  \begin{itemize}
    \item $\delta_2(0,\leftend)=(0,+1)$,
    \item $\delta_2(i,a)=(i,+1)$ for $0<i\leq k$ and $a\in\Sigma$,
    \item $\delta_2(i,a)=(i,-1)$ for $-k<i<0$ and $a\in\Sigma$,
    \item $\delta_2(end,a)=\delta_2(end,\#)=(-k,-1)$,
    \item $\delta_2(i,\#)=(i+1,+1)$ for $0\leq i < k$, 
    \item $\delta_2(k,\#)=(end,+1)$,
    \item $\delta_2(i,\#)=(i+1,-1)$ for $-k\leq i < -1$, 
    \item $\delta_2(-1,\#)=(1,+1)$.
  \end{itemize}
  
  The production function $\gamma_2$ is then simply $\gamma_2(i,a)=a$ for $i>0$,
  $\gamma_2(i,\#)=\#$ for $i=0$ and $i=k$, and is set to $\epsilon$ for all other 
  transitions.

  The way the transducer $T_2$ works is that it reads forward, in the strictly positive
  states, a factor of the input containing $k$ $\#$ symbols, copying it to the output.
  Upon reaching the $k^{th}$ $\#$ symbol, it goes to state $end$ to check if it was the
  last $\#$ symbol, and otherwise reads back the last $k-1$ $\#$ symbols and starts
  again.  

Let us prove the aperiodicity of $T_2$. 
First, notice that the $\leftleft$ and $\rightright$ are always aperiodic relations, for
any finite 2DFT. This is due to the fact that if a $\leftleft$ step exists in some
$u\neq\varepsilon$, it also appears in $uv$.  So for $(v^n)_{n>0}$, the $\leftleft$
and $\rightright$ relations are monotone, and since we consider finite state machines,
they eventually stabilize.
So we turn to traversal steps. These traversal steps only depend on the number of $\#$ symbols in the word, as well as the starting and ending symbols.
In particular, if a word $v$ has $k+1$ or more $\#$ symbols, the only
traversals it realises are in $\{\LR{0}{k},\LR{0}{end},\LR{1}{k},\LR{1}{end}\}$, starting 
from $0$ is possible only if $v$ starts with $\#$, the target state is $end$ if the last 
letter of $v$ is $\#$, otherwise it is $k$. Notice that both $\LR{0}{end}$ and $\LR{1}{end}$ 
are possible if $v\in\#(\Sigma^{*}\#)^{\geq k}$. Similarly, both $\LR{0}{k}$ and $\LR{1}{k}$ 
are possible if $v\in\#(\Sigma^{*}\#)^{\geq k}\Sigma^{+}$.
Then given any word $v\in(\Sigma\cup\{\#\})^{+}$, both $v^{k+1}$ and $v^{k+2}$ have either
no $\#$ symbol, or at least $k+1$ $\#$ symbols; further they have the same starting and
ending letters.  Thus they realise the same steps.

  \textbf{The function $f_3$.}
	The goal of $f_3$ is to iteratively simulate $f$ on each factor appearing between $\#$
	symbols.  To this end, $T_3$ is defined as the transducer $T$ realising $f$, with the
	exception that it reads the $\#$ symbols as endmarkers, and upon reaching a final state
	while reading a $\#$ symbol, it first checks if the next symbol is $\rightend$, and in this
	case ends the run, or simulates the move of $T$ reading $\leftend$ from the initial state.
	Note that $\#$ being used for both endmarkers could generate some non-determinism,
	however this can be avoided as the left endmarker can only be reached while moving to
	the left, and symmetrically for the right endmarker.  Then we solve non-determinism by
	duplicating the problematic states $q$ to states $q_\ell$ (where $\#$ is seen as the
	left endmarker) and $q_r$ (where $\#$ is seen a right endmarker), which can only be
	reached while moving to the left or the right respectively.

	We now turn to the aperiodicity of $T_3$. 
	If the input word $v$ does not contain any $\#$ symbol, then the 
  $(\leftleft,\rightright,\leftright,\rightleft)$-runs of $v^n$ are the
	same as the ones in $T$, and since $T$ is aperiodic then we get
	$\varphi(v^n)=\varphi(v^{n+1})$ for some $n$, where $\varphi$ is the syntactic morphism of $T$.
	
	Otherwise, let us remark that by design, once the reading head has gone right of a given
	$\#$ symbol, it never goes back to its left, and secondly the behavior of $T_3$ when
  going from left to right of a $\#$ symbol is always the same since it simulates the 
  initial transition of $T$. 
	So given a word $v$ with at least one $\#$ symbol, let $u_1$ and $u_2$ be the prefix and
	suffix of $v$ upto the first and from the last $\#$ respectively, i.e., 
  $v=u_1\#w_1\#\cdots w_m\#u_2$ with $m\geq0$ and $u_1,u_2,w_1,\ldots,w_m\in \Sigma^*$.
	Then there exists no $\rightleft$ traversal of $v^{n}$ for $n\geq2$ since the reading
	head cannot move from right to left of a $\#$ symbol.
	The $\leftright$ traversals of $v^n$, for $n\geq 2$, exist if and only if
	$u_2u_1,w_1,\ldots,w_m$ belong to the domain of $T$, and consist of all $\LR{p}{q}$,
	where $\varphi(u_1)$ contains $\LR{p}{f}$ for some final state $f$, and $\varphi(u_2)$
	contains $\LR{\iota}{q}$ where $\iota$ is the initial state.
	These traversals are then the same for $v^2$ and $v^3$, which concludes the proof of aperiodicity of $T_3$.

  \textbf{The function $f_4$.}
  The transducer $T_4$ realising $f_4$ is similar to $T_3$. 
  The main difference is that it starts by reaching the end of the word, then goes back to
  the previous $\#$ symbol to simulate $T$.  On reaching the end of the run in $T$ (in a
  final state of $T$ when reading $\#$) it treats $\#$ as $\rightend$ and then enters into a
  special state which moves the reading head to the left, till the time it has finished
  reading two $\#$ symbols, while outputting $\epsilon$ all along.  When it reads the
  second $\#$, it moves right entering the state $T$ would, on reading $\leftend$ from 
  its initial state, and continues simulating $T$.  This goes on until it reaches the
  start symbol $\leftend$, and then it goes to the final state of $T_4$ that only moves to
  the right outputting $\epsilon$ all along until the end of the input to $\rightend$.

  The arguments for the aperiodicity of $T_4$ are similar to the ones for $T_3$.
\end{proof}

\subsection{SD-regular transducer expressions for aperiodic 2DFTs}\label{sec:SDRETtoA2DFT}
In this section, we show that the runs of an aperiodic 2DFT have a ``stabilising'' 
property.  This property crucially distinguishes aperiodic 2DFTs from
non aperiodic ones, and we use this in our proof to obtain \SDRTEs from aperiodic 2DFTs.
In the remainder of this section, we fix an aperiodic 2DFT
$\A=(Q,\Sigma,\Gamma,\delta,\gamma,q_0,F)$.  Let
$\varphi\colon(\Sigma\uplus\{\leftend,\rightend\})^{*}\to\TrMon$ be the
canonical surjective morphism to the transition monoid of $\A$.
 
\subsubsection{Stabilising runs in two-way automata}\label{sec:run-stable}

Consider a \emph{code} $L\subseteq\Sigma^{*}$ such that $X=\varphi(L)$ is 
$k$-stabilizing for some $k>0$. We will see that a run of $\A$ over a word 
$w\in L^{*}$ has some nice properties. Intuitively, if it moves forward through 
$k$ factors from $L$ then it never moves backward through more than $k$ factors.

More precisely, let $w=u_1u_2\cdots u_n$ be the unique factorisation of $w\in
L^{*}$ with $u_i\in L$ for $1\leq i\leq n$.  We assume that $n\geq k$.
We start with the easiest fact.

\begin{lemma}\label{lem:LL}
  If $\LL{p}{q}\in\varphi(w)$ then the run of $\A$ over $w$ starting on the left
  in state $p$ only visits the first $k$ factors $u_1\cdots u_k$ of $w$.
  \end{lemma}

\begin{proof}
  Since $X$ is $k$-stabilising, we have $\varphi(w)=\varphi(u_1\cdots u_k)$.
  Hence, $\LL{p}{q}\in\varphi(u_1\cdots u_k)$ and the result follows since $\A$ 
  is deterministic.
\end{proof}

Notice that the right-right ($\rightright$) runs of $\A$ over $w$ need not visit
the last $k$ factors only (see Lemma~\ref{lem:RR} below).  This is due to the
fact that \emph{stabilising} is not a symmetric notion.

Next, we consider the left-right runs of $\A$ over $w$.

\begin{lemma}\label{lem:LR}
  Assume that $\LR{p}{q}\in\varphi(w)$.  Then the run $\rho$ of $\A$ over $w$
  starting on the left in state $p$ has the following property, that we call
  $k$-forward-progressing: for each $1\leq i<n-k$, after reaching the suffix
  $u_{i+k+1}\cdots u_n$ of $w$, the run $\rho$ will never visit again the prefix
  $u_1\cdots u_i$.  See Figure~\ref{fig:LR-run} for a non-example and
  Figure~\ref{fig:LR-2progressing} for an example.
\end{lemma}

\begin{proof}
  Towards a contradiction, assume that for some $1\leq i<n-k$, the run $\rho$
  visits $u_1\cdots u_i$ after visiting $u_{i+k+1}\cdots u_n$ (See
  Figure~\ref{fig:LR-run}).  
  Then, there exists a subrun $\rho'$ of $\rho$ making some
  $(\leftleft,q_1,q_3)$-step on $u_{i+1}\cdots u_n$ and visiting $u_{i+k+1}$ (on
  Figure~\ref{fig:LR-run} we have $\rho'=\rho_2\rho_3$).  Hence 
  $(\leftleft,q_1,q_3)\in\varphi(u_{i+1}\cdots u_n)$ and by Lemma~\ref{lem:LL} we deduce 
  that $\rho'$ visits $u_{i+1}\cdots u_{i+k}$ only, a contradiction.
\begin{figure}[tb]
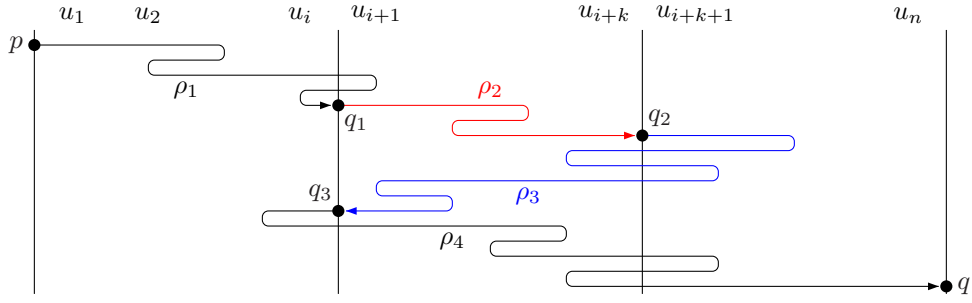

  \centering
  \gusepicture{gpic:LR-run}
  \caption{A left-right run which is not $k$-forward-progressing}
  \label{fig:LR-run}
\end{figure}
\end{proof}

\begin{lemma}\label{lem:RL}
  Assume that $\RL{p}{q}\in\varphi(w)$.  Then the run $\rho$ of $\A$ over $w$
  starting on the right in state $p$ has the following property, that we call
  $k$-backward-progressing: for each $1\leq i<n-k$, after reaching the prefix
  $u_1\cdots u_{i}$ of $w$, the run $\rho$ will never visit again the suffix
  $u_{i+k+1}\cdots u_n$. 
\end{lemma}
  
\begin{proof}
	This Lemma is a consequence of Lemma~\ref{lem:LL}.
	Indeed, consider any part of $\rho$ that visits $u_{i+1}$ again (in some state $q_1$)
	after visiting $u_i$, for some $1\leq i<n-k$.
	As $\rho$ is a $\rightleft$ run, it will later cross from $u_{i+1}$ to $u_i$ (reaching
	some state $q_3$).  Then $\LL{q_1}{q_3}$ is a run on $u_{i+1}\cdots u_n$.
  By Lemma~\ref{lem:LL}, it does not visit $u_{i+k+1}\cdots u_n$, which concludes the
  proof (See Figure~\ref{fig:RL-run} for a non-example).
\begin{figure}[tb]
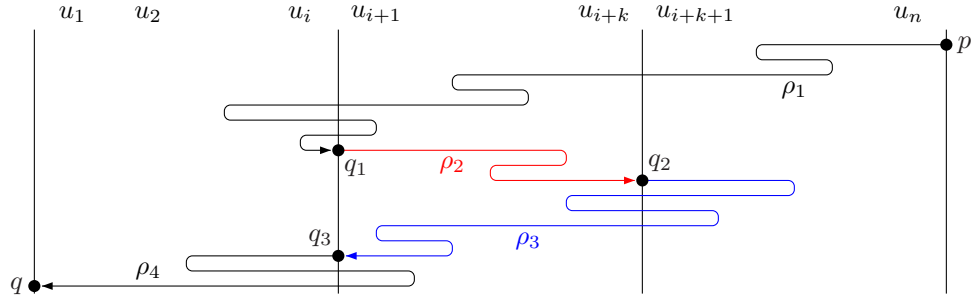

  \centering
  \gusepicture{gpic:RL-run}
  \caption{A right-left run which is not $k$-backward-progressing}
  \label{fig:RL-run}
\end{figure}
\end{proof}

\begin{lemma}\label{lem:RR}
  Assume that $\RR{p}{q}\in\varphi(w)$ and let $\rho$ be the run of $\A$ over
  $w$ starting on the right in state $p$.  Then, either $\rho$ visits only the
  last $k$ factors $u_{n-k+1}\cdots u_n$, or for some $1\leq i\leq n-k$ the run
  $\rho$ is the concatenation $\rho_1\rho_2\rho_3$ of a $k$-backward-progressing
  run $\rho_1$ over $u_{i+1}\cdots u_n$ followed by a run $\rho_2$ staying
  inside some $u_i\cdots u_{i+k}$, followed by some $k$-forward-progressing run
  $\rho_3$ over $u_{i+1}\cdots u_n$.  See Figure~\ref{fig:RR-run}.
\end{lemma}

\begin{proof}
  Assume that $\rho$ visits $u_1\cdots u_{n-k}$ and let $u_i$ ($1\leq i\leq
  n-k$) be the left-most factor visited by $\rho$.  We split $\rho$ in
  $\rho_1\rho_2\rho_3$ (see Figure~\ref{fig:RR-run}) where
\begin{figure}[t]
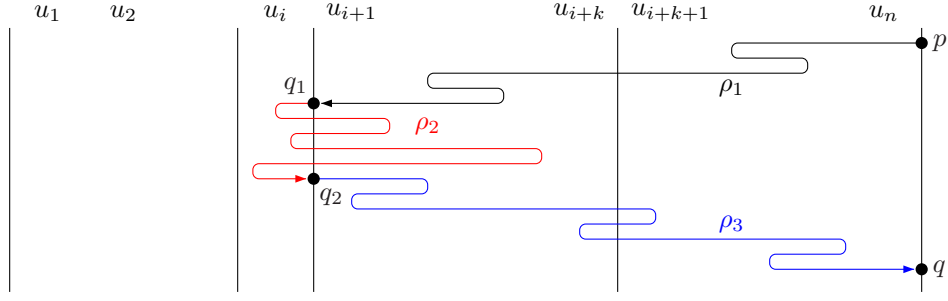

  \centering
  \gusepicture{gpic:RR-run}
  \caption{A right-right run $\rho_1\rho_2\rho_3$ where $\rho_1$ is 
  $k$-backward-progressing, $\rho_2$ is local to $u_i\cdots u_{i+k}$ and 
  $\rho_3$ is $k$-forward-progressing.}
  \label{fig:RR-run}
\end{figure}
  \begin{itemize}
    \item $\rho_1$ is the prefix of $\rho$, starting on the right of $w$ in
    state $p$ and going until the first time $\rho$ crosses from $u_{i+1}$ to
    $u_{i}$. Hence, $\rho_1$ is a run over $u_{i+1}\cdots u_{n}$ starting
    on the right in state $p$ and exiting on the left in some state
    $q_1$.  We have $\RL{p}{q_1}\in\varphi(u_{i+1}\cdots u_{n})$. By 
    Lemma~\ref{lem:RL} we deduce that $\rho_1$ is $k$-backward-progressing.
  
    \item Then, $\rho_2$ goes until the last crossing from $u_{i}$ to
    $u_{i+1}$.   
    
    \item Finally, $\rho_3$ is the remaining suffix of $\rho$.  Hence, $\rho_3$
    is a run over $u_{i+1}\cdots u_{n}$ starting on the left in some state $q_2$
    and exiting on the right in state $q$.  We have
    $\LR{q_2}{q}\in\varphi(u_{i+1}\cdots u_{n})$.  By Lemma~\ref{lem:LR} we
    deduce that $\rho_3$ is $k$-forward-progressing.
  \end{itemize}  
  It remains to show that $\rho_2$ stays inside $u_i\cdots u_{i+k}$.  Since $u_i$ is the
  left-most factor visited by $\rho$, we already know that $\rho_2$ does not visit
  $u_1\cdots u_{i-1}$.
  Similarly to Lemma~\ref{lem:RL}, any maximal subrun $\rho_2'$ of $\rho_2$ that does not
  visit $u_i$ is a $\leftleft$ run on $u_{i+1}\cdots u_n$ since $\rho_2$ starts and ends
  at the frontier between $u_i$ and $u_{i+1}$.  By Lemma~\ref{lem:LL}, the subrun
  $\rho_2'$ does not visit $u_{i+k+1}\cdots u_n$ and thus $\rho_2$ stays inside $u_i\cdots
  u_{i+k}$.
\end{proof}

\begin{example}\label{eg2dftstable}
  We illustrate the stabilising runs of an aperiodic 2DFT using the aperiodic 2DFT $\A$ in
  Figure~\ref{fig:eg-2dft}.  Figure~\ref{fig:run-eg} depicts the run of $\A$ on words in
  $b(a^*b)^{\geq3}$.  We use the set $Z_1$ computed in Example \ref{eg:stabilising}.
  Notice that a run of $\A$ on such words is 4-forward-progressing, as seen below.  For
  each $w=u_1 u_2 \cdots u_n$ with $n>3$, $u_1=b$ and $u_i\in a^*b$ for $2\leq i\leq n$, 
  we have $\varphi(w)=Z_1$ and one can see that
  \begin{itemize}
    \item each $(\leftleft, p, q) \in Z_1$, is such that, whenever the run of $\A$ starts
    at the left of $w$ in state $p$, it stays within $u_1 \cdots u_4$ and never visits
    $u_{5}\cdots u_n$ (as in Lemma~\ref{lem:LL}).
    
    \item each $(\leftright, p, q) \in Z_1$, is such that, whenever the run of $\A$ starts
    at the left of $w$ in state $p$ and reaches $u_{i+5}$, for $i\geq1$, it no longer
    visits any of $u_1\cdots u_{i}$ (4-forward-progressing as in Lemma~\ref{lem:LR}).
    
    \item each $(\rightright, p, q) \in Z_1$, is such that, whenever the run of $\A$
    starts at the right of $w$ in state $p$, it never visits $u_1\cdots u_{n-4}$ (the
    easy case of Lemma~\ref{lem:RR}).
  \end{itemize}
\end{example}

 \begin{figure}[t]
	\includegraphics[scale=.4]{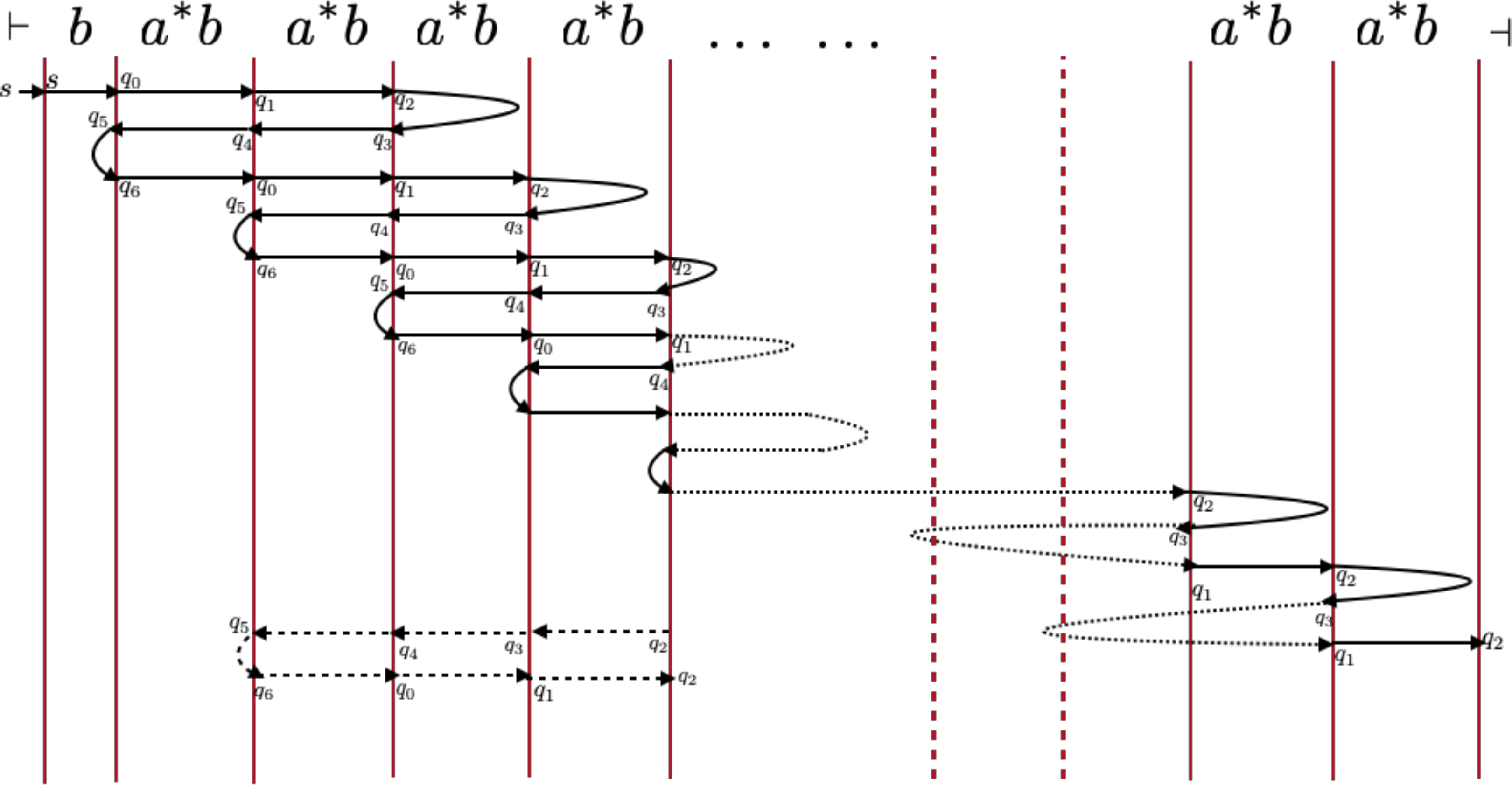}
	\caption{An accepting run on words in $b(a^*b)^{\geq3}$. Bottom left: $(\rightright,q_2,q_2)$.}
	\label{fig:run-eg} 
\end{figure}

\subsubsection{Computing \SDRTE}
In this section, we show how to construct \SDRTEs which are equivalent to
aperiodic 2DFTs.  Recall that 
$\varphi\colon(\Sigma\uplus\{\leftend,\rightend\})^{*}\to\TrMon$ is the
canonical surjective morphism to the transition monoid of the 2DFT $\A=(Q,\Sigma,\Gamma,\delta,\gamma,q_0,F)$. 
Given a regular expression $E$ and a monoid element $s\in\TrMon$, we let
$\lang{E,s}=\lang{E}\cap\varphi^{-1}(s)$.
The main construction of this section is given by Theorem \ref{thm:main}.

Recall that $\TrMon$ represents the transition monoid of a 2DFT, and 
consists of elements $\varphi(w)$ for all $w \in \Sigma^*$, where  each 
$\varphi(w)=\{(d,p,q)\mid \text{there is a } (d,p,q)\text{-run on }w\}
\subseteq\{\leftright,\leftleft,\rightright,\rightleft\}\times Q^{2}$. 
The elements of $\varphi(w)$ are called \emph{steps}, since 
any run of $w$ is obtained by a sequence of such steps. If the states $p, q$ in
a step $(d, p, q)$ are clear from the context, or is immaterial 
for the discussion we also refer to a step as a $d$ step, 
$d \in \{\leftleft, \rightright, \leftright, \rightleft\}$. 
In this case we also refer to a step $(d, p, q)$ as 
a $d$ step having $p$ as the starting state and $q$ as the final state.

\begin{theorem}\label{thm:main}
  Let $E$ be an unambiguous, stabilising, SD-regular expression over
  $\Sigma\uplus\{\leftend,\rightend\}$ and let $s\in\TrMon$.  For each step
  $x\in\{\leftright,\leftleft,\rightright,\rightleft\}\times Q^{2}$, we can
  construct an \SDRTE $\C{E}{s}x$ such that:
  \begin{enumerate}
    \item  $\C{E}{s}x=\bot$ when $x\notin s$, and otherwise
    
    \item $\dom{\sem{\C{E}{s}x}}=\lang{E,s}$ and for all words $w\in\lang{E,s}$,
    $\sem{\C{E}{s}x}(w)$ is the output produced by $\A$ running over $w$
    according to step $x$.
    
    When $w=\varepsilon$ and $s=\mathbf{1}=\varphi(\varepsilon)$ with
    $x\in\mathbf{1}$, this means $\sem{\C{E}{s}x}(\varepsilon)=\varepsilon$.
  \end{enumerate}  
\end{theorem}

\begin{proof}
  The construction is by structural induction on $E$.

\paragraph*{Atomic expressions}
We first define $\C{E}{s}{x}$ 
when $E$ is an atomic expression, i.e., $\emptyset$, $\varepsilon$ or $a$ for $a\in\Sigma$.
\begin{itemize}
\item $E=\emptyset$: we simply set $\C{\emptyset}{s}{x}=\bot$, which is the 
nowhere defined function.

\item $E=\varepsilon$: when $s=\mathbf{1}$ and $x\in s$ then we set
$\C{\varepsilon}{s}{x}=\SimpleFun{\varepsilon}{\varepsilon}$ and otherwise we 
set $\C{\varepsilon}{s}{x}=\bot$.

\item $E=a\in\Sigma\uplus\{\leftend,\rightend\}$: again, we set
$\C{a}{s}{x}=\bot$ if $s\neq\varphi(a)$ or $x\notin s$.  Otherwise, there are
two cases.  Either $x\in\{\LR{p}{q},\RR{p}{q}\}$ for some states $p,q$ such that
$\delta(p,a)=(q,+1)$, or $x\in\{\RL{p}{q},\LL{p}{q}\}$ for some states $p,q$
with $\delta(p,a)=(q,-1)$.  In both cases the output produced is $\gamma(p,a)$
and we set $\C{a}{s}{x}=\SimpleFun{a}{\gamma(p,a)}$.
\end{itemize}

\paragraph*{Disjoint union}
If the expression is $E\cup F$ with $\lang{E}$ and $\lang{F}$ disjoint, then we simply set 
$\C{E\cup F}{s}{x}=\C{E}{s}{x}+\C{F}{s}{x}$.

\paragraph*{Unambiguous concatenation $E\cdot F$}\label{sec:unambig-concat}

Here, we suppose that we have \SDRTEs for $\C{E}{s}{x}$ and $\C{F}{s}{x}$ for
all $s$ in $\TrMon$ and all steps
$x\in\{\leftright,\leftleft,\rightright,\rightleft\}\times Q^{2}$.  We show how
to construct \SDRTEs for $\C{E\cdot F}{s}{x}$, assuming that the concatenation
$\lang{E}\cdot\lang{F}$ is unambiguous.

A word $w\in\lang{E\cdot F}$ has a unique factorization $w=uv$ with
$u\in\lang{E}$ and $v\in\lang{F}$.  Let $s=\varphi(u)$ and $t=\varphi(v)$.  A
run $\rho$ over $w$ is obtained by stitching together runs over $u$ and runs
over $v$ as shown in Figure~\ref{fig:LR-concat}.  In the left figure, the run over
$w$ follows step $x=\LR{p}{q}$ starting on the left in state $p$ and exiting on
the right in state $q$.  The run $\rho$ splits as
$\rho_0\rho_1\rho_2\rho_3\rho_4\rho_5$ as shown in the figure.  The output of
the initial part $\rho_0$ is computed by $\C{E}{s}{\LR{p}{p_1}}$ over $u$ and
the output of the final part $\rho_5$ is computed by $\C{F}{t}{\LR{p_5}{q}}$
over $v$.  We focus now on the internal part $\rho_1\rho_2\rho_3\rho_4$ which
consists of an alternate sequence of left-left runs over $v$ and right-right
runs over $u$.  The corresponding sequence of steps $x_1=\LL{p_1}{p_2}\in t$,
$x_2=\RR{p_2}{p_3}\in s$, $x_3=\LL{p_3}{p_4}\in t$ and $x_4=\RR{p_4}{p_5}\in s$
depends only on $s=\varphi(u)$ and $t=\varphi(v)$.

\begin{figure}[tb]
  \centering
  \gusepicture{gpic:LR-concat}
  \hfill
  \gusepicture{gpic:LL-concat}
  \caption{Decomposition of a $\LR{p}{q}$-run and a $\LL{p}{q}$-run over the product $w=uv$.}
  \label{fig:LR-concat}
\end{figure}

These internal zigzag runs will be frequently used when dealing with 
concatenation or Kleene star. 
They alternate left-left ($\leftleft$) steps on the right word $v$ and 
right-right ($\rightright$) steps on the left word $u$. 
They may start with a $\leftleft$-step or a $\rightright$-step. The sequence of 
steps in a \emph{maximal} zigzag run is entirely determined by the monoid 
elements $s=\varphi(u)$, $t=\varphi(v)$, the starting step 
$d\in\{\leftleft,\rightright\}$ and the starting state $p'$ of step $d$. The 
final step of this \emph{maximal} sequence is some 
$d'\in\{\leftleft,\rightright\}$ and reaches some state $q'$. We write 
$Z_{s,t}(p',d)=(d',q')$. For instance, on the left of Figure~\ref{fig:LR-concat} 
we get $Z_{s,t}(p_1,\leftleft)=(\rightright,p_5)$ whereas on the right of  
Figure~\ref{fig:LR-concat} we get $Z_{s,t}(p_1,\leftleft)=(\leftleft,p_4)$. By 
convention, if the sequence of zigzag steps is empty then we define
$Z_{s,t}(p,\leftleft)=(\rightright,p)$ and $Z_{s,t}(p,\rightright)=(\leftleft,p)$.

\begin{lemma}\label{lem:InternalRun}
  We use the above notation.  We can construct  \SDRTEs
  $\ZF{E,s}{F,t}{p,d}$ for $p\in Q$ and $d\in\{\leftleft,\rightright\}$ such
  that $\dom{\sem{\ZF{E,s}{F,t}{p,d}}}=\lang{E,s}\lang{F,t}$ and
  for all $u\in\lang{E,s}$ and $v\in\lang{F,t}$ the value
  $\sem{\ZF{E,s}{F,t}{p,d}}(uv)$ is the output produced by the \emph{internal}
  zigzag run of $\A$ over $(u,v)$ following the maximal sequence of steps 
  starting in state $p$ with a $d$-step.
\end{lemma}

\begin{proof}
  We first consider the case $d={\leftleft}$
  and $Z_{s,t}(p,\leftleft)=(\rightright,q)$ for some $q\in Q$
  which is illustrated on the left of Figure~\ref{fig:LR-concat}.  Since 
  $\A$ is deterministic, there is a unique maximal sequence of steps (with 
  $n\geq0$, $p_1=p$ and $p_{2n+1}=q$):
  $x_1=\LL{p_1}{p_2}\in t$, $x_2=\RR{p_2}{p_3}\in s$, \ldots,
  $x_{2n-1}=\LL{p_{2n-1}}{p_{2n}}\in t$, $x_{2n}=\RR{p_{2n}}{p_{2n+1}}\in s$.
  The zigzag run $\rho$ following this sequence of steps over $uv$ splits as
  $\rho_1\rho_2\cdots\rho_{2n}$ where $\rho_{2i}$ is the unique
  run on $u$ following step $x_{2i}$ and $\rho_{2i+1}$ is the unique run on $v$
  following step $x_{2i+1}$.  The output of these runs are given by
  $\sem{\C{E}{s}{x_{2i}}}(u)$ and $\sem{\C{F}{t}{x_{2i+1}}}(v)$.  When $n=0$ the
  zigzag run $\rho$ is empty and we simply set
  $\ZF{E,s}{F,t}{p,\leftleft}=\SimpleFun{(\lang{E,s}\lang{F,t})}{\varepsilon}$.
  Assume now that $n>0$. The required  \SDRTE computing the 
  output of 
  $\rho$ can be defined as
  \begin{align*}
    \ZF{E,s}{F,t}{p,\leftleft} ={}  &
    \big((\SimpleFun{\lang{E,s}}{\varepsilon})\cdot \C{F}{t}{x_1}\big)  
    \odot \big(\C{E}{s}{x_2} \cdot \C{F}{t}{x_3}\big) \odot \cdots \odot \\
    &  \big(\C{E}{s}{x_{2n-2}} \cdot \C{F}{t}{x_{2n-1}}\big) \odot
    \big( \C{E}{s}{x_{2n}}\cdot(\SimpleFun{\lang{F,t}}{\varepsilon})\Big) \,.
  \end{align*}
  Notice that each Cauchy product in this expression is unambiguous since the 
  product $\lang{E}\cdot\lang{F}$ is unambiguous.

  The other cases can be handled similarly.  For instance, when
  $Z_{s,t}(p,\leftleft)=(\leftleft,q)$ as on the right of
  Figure~\ref{fig:LR-concat}, the sequence of steps ends with
  $x_{2n-1}=\LL{p_{2n-1}}{p_{2n}}\in t$ with $n>0$ and $p_{2n}=q$ and the zigzag
  run $\rho$ is $\rho_1\rho_2\cdots\rho_{2n-1}$.  The  \SDRTE
  $\ZF{E,s}{F,t}{p,\leftleft}$ is given by
  $$
    \big((\SimpleFun{\lang{E,s}}{\varepsilon})\cdot \C{F}{t}{x_1}\big)  
    \odot \big(\C{E}{s}{x_2} \cdot \C{F}{t}{x_3}\big) \odot \cdots \odot 
    \big(\C{E}{s}{x_{2n-2}} \cdot \C{F}{t}{x_{2n-1}}\big) \,.
  $$
  The situation is symmetric for $\ZF{E,s}{F,t}{p,\rightright}$: the sequence 
  starts with a right-right step $x_2=\RR{p_2}{p_3}\in s$ with $p=p_2$ and we 
  obtain the  \SDRTE simply by removing the
  first factor $\big((\SimpleFun{\lang{E,s}}{\varepsilon})\cdot 
  \C{F}{t}{x_1}\big)$ in the Hadamard products above.
\end{proof}

We come back to the definition of the  \SDRTEs for $\C{E\cdot
F}{r}{x}$ with $r\in\TrMon$ and $x\in r$.  As explained above, the output
produced by a run $\rho$ following step $x$ over a word $w=uv$ with
$u\in\lang{E,s}$, $v\in\lang{F,t}$ and $r=st$ consists of an initial part,
a zigzag internal part, and a final part.  
There are four cases
depending on the step $x$.
\begin{itemize}
  \item $x=\LL{p}{q}$. Either the run $\rho$ stays inside $u$ (zigzag part 
  empty) or there is a zigzag internal part starting with $(p',\leftleft)$ such 
  that $\LR{p}{p'}\in\varphi(u)$ and ending with $(\leftleft,q')$ such that 
  $\RL{q'}{q}\in\varphi(u)$.
  Thus we define the  \SDRTE $\C{E\cdot F}{r}{x}$ as
  \begin{align*}
    \sum_{st=r\mid x\in s}
    &\hspace{-2mm}\C{E}{s}{x}\cdot \big(\SimpleFun{\lang{F,t}}{\varepsilon}\big) +{} \\ 
    & \hspace{-18mm}\sum_{\substack{st=r,\, (p',q')\,\mid \\
    Z_{s,t}(p',\leftleft)=(\leftleft,q')}} \hspace{-8mm}
    \big(\C{E}{s}{\LR{p}{p'}}\cdot (\SimpleFun{\lang{F,t}}{\varepsilon})\big) 
    \odot \ZF{E,s}{F,t}{p',\leftleft} \odot 
    \big(\C{E}{s}{\RL{q'}{q}}\cdot (\SimpleFun{\lang{F,t}}{\varepsilon})\big)
  \end{align*}
  Notice that all Cauchy products are unambiguous since the concatenation
  $\lang{E}\cdot\lang{F}$ is unambiguous.  The sums are also unambiguous.
  Indeed, a word $w\in\lang{E\cdot F,r}$ has a unique factorization $w=uv$ with
  $u\in\lang{E}$ and $v\in\lang{F}$.  Hence $s=\varphi(u)$ and $t=\varphi(v)$
  are uniquely determined and satisfy $st=r$.  Then, either $x\in s$ and $w$ is
  only in the domain of $\C{E}{s}{x}\cdot
  \big(\SimpleFun{\lang{F,t}}{\varepsilon}\big)$.  Or there is a unique $p'$
  with $\LR{p}{p'}\in s$ and a unique $q'$ with
  $Z_{s,t}(p',\leftleft)=(\leftleft,q')$ and $\RL{q'}{q}\in s$.
  Notice that if $\LR{p}{p'}\notin s$ then $\C{E}{s}{\LR{p}{p'}}=\bot$ and
  similarly if $\RL{q'}{q}\notin s$. Hence we could have added the condition 
  $\LR{p}{p'},\RL{q'}{q}\in s$ to the second sum, but do not, to reduce clutter. 
  
  \item $x=\LR{p}{q}$. Here the run must cross from left to right. 
  Thus we define the  \SDRTE $\C{E\cdot F}{r}{x}$ as
  $$
  \hspace{-6mm}\sum_{\substack{st=r,\, (p',q')\,\mid \\ 
  Z_{s,t}(p',\leftleft)=(\rightright,q')}} \hspace{-8mm}
    \big(\C{E}{s}{\LR{p}{p'}}\cdot (\SimpleFun{\lang{F,t}}{\varepsilon})\big) 
    \odot \ZF{E,s}{F,t}{p',\leftleft} \odot 
    \big((\SimpleFun{\lang{E,s}}{\varepsilon})\cdot \C{F}{t}{\LR{q'}{q}}\big)
  $$
  
  \item $x=\RL{p}{q}$. This case is similar.
  The  \SDRTE $\C{E\cdot F}{r}{x}$ is
  $$
  \hspace{-6mm}\sum_{\substack{st=r,\, (p',q')\,\mid \\ 
  Z_{s,t}(p',\rightright)=(\leftleft,q')}} \hspace{-8mm}
    \big((\SimpleFun{\lang{E,s}}{\varepsilon})\cdot \C{F}{t}{\RL{p}{p'}}\big)
    \odot \ZF{E,s}{F,t}{p',\rightright} \odot 
    \big(\C{E}{s}{\RL{q'}{q}}\cdot (\SimpleFun{\lang{F,t}}{\varepsilon})\big) 
  $$

  \item $x=\RR{p}{q}$.  Finally, for right-right runs, the 
  \SDRTE $\C{E\cdot F}{r}{x}$ is
  \begin{align*}
    \sum_{st=r\mid x\in t}
    &\hspace{-2mm} (\SimpleFun{\lang{E,s}}{\varepsilon})\cdot\C{F}{t}{x} +{} \\ 
    & \hspace{-17mm}\sum_{\substack{st=r,\, (p',q')\,\mid \\ 
    Z_{s,t}(p',\rightright)=(\rightright,q')}} \hspace{-8mm}
    \big((\SimpleFun{\lang{E,s}}{\varepsilon})\cdot\C{F}{t}{\RL{p}{p'}}\big) 
    \odot \ZF{E,s}{F,t}{p',\rightright} \odot 
    \big((\SimpleFun{\lang{E,s}}{\varepsilon})\cdot\C{F}{t}{\LR{q'}{q}}\big)
  \end{align*}
\end{itemize}

\begin{example}\label{eg:product}	
  We go back to our running example of the aperiodic 2DFT $\A$ in Figure~\ref{fig:eg-2dft}
  and illustrate the unambiguous concatenation.  Consider $F=a^{+}b$, $G=F^{2}$ and
  $E=F^{4}=G^{2}$.  We know from Example~\ref{eg:stabilising} that $\varphi(F)=Y_2$,
  $\varphi(G)=Y_4$ and $\varphi(E)=Z_2$.  We compute below some steps of $Y_2$, $Y_4$ and
  $Z_2$.
  
  First, we look at some steps in $F$ for which the \SDRTE are obtained directly by 
  looking at the automaton $\A$ in Figure~\ref{fig:eg-2dft} (we cannot give more details 
  here since we have not explained yet how to deal with Kleene-plus, hence we rely on 
  intuition for these steps). 
  \begin{align*}
    \C{F}{Y_2}{\LR{q_0}{q_1}} &= (\SimpleFun{a^{+}b}{\varepsilon})
    &
    \C{F}{Y_2}{\LR{q_1}{q_2}} &= (\SimpleFun{a^{+}b}{\varepsilon})
    \\
    \C{F}{Y_2}{\LL{q_2}{q_3}} &= (\SimpleFun{a}{a})^{+} \cdot (\SimpleFun{b}{b})
    &
    \C{F}{Y_2}{\LR{q_6}{q_0}} &= (\SimpleFun{a}{a})^{+} \cdot (\SimpleFun{b}{b})
    \\
    \C{F}{Y_2}{\RL{q_3}{q_4}} &= (\SimpleFun{a^{+}b}{\varepsilon})
    &
    \C{F}{Y_2}{\RL{q_4}{q_5}} &= (\SimpleFun{a^{+}b}{\varepsilon})
    \\
    \C{F}{Y_2}{\RR{q_5}{q_6}} &= (\SimpleFun{a^{+}b}{\varepsilon}) \,.
  \end{align*}
  Next, we compute some steps using the unambiguous concatenation $G=F\cdot F$. We start 
  with step $\LR{q_0}{q_2}$ for which the zigzag part is empty: 
  $\ZF{F,Y_2}{F,Y_2}{q_1,\leftleft}=\SimpleFun{F^{2}}{\varepsilon}$. Hence, we get using 
  the formula in the proof above 
  \begin{align*}
    \C{G}{Y_4}{\LR{q_0}{q_2}} &= 
    \big( \C{F}{Y_2}{\LR{q_0}{q_1}} \cdot (\SimpleFun{F}{\varepsilon}) \big) 
    \odot (\SimpleFun{F^{2}}{\varepsilon}) \odot
    \big( (\SimpleFun{F}{\varepsilon}) \cdot \C{F}{Y_2}{\LR{q_1}{q_2}} \big)
    \\
    \intertext{and after some simplifications}
    \C{G}{Y_4}{\LR{q_0}{q_2}} &= 
    \C{F}{Y_2}{\LR{q_0}{q_1}} \cdot \C{F}{Y_2}{\LR{q_1}{q_2}} 
    = (\SimpleFun{(a^{+}b)^{2}}{\varepsilon}) \,.
  \end{align*}
  Similarly, we can compute the following steps
  \begin{align*}
    \C{G}{Y_4}{\LR{q_6}{q_1}} &= 
    \C{F}{Y_2}{\LR{q_6}{q_0}} \cdot \C{F}{Y_2}{\LR{q_0}{q_1}} 
    = (\SimpleFun{a}{a})^{+} \cdot (\SimpleFun{b}{b}) \cdot 
    (\SimpleFun{a^{+}b}{\varepsilon}) \,.
  \end{align*}
  For step $\LL{q_2}{q_3}$, the run only visits the first factor 
  since $\LL{q_2}{q_3}\in Y_2$:
  \begin{align*}
    \C{G}{Y_4}{\LL{q_2}{q_3}} &= 
    \C{F}{Y_2}{\LL{q_2}{q_3}} \cdot (\SimpleFun{F}{\varepsilon}) 
    = (\SimpleFun{a}{a})^{+} \cdot (\SimpleFun{b}{b}) \cdot 
    (\SimpleFun{a^{+}b}{\varepsilon}) \,.
  \end{align*}
  Now, for step $\LL{q_1}{q_4}$, the zigzag part is reduced to step $\LL{q_2}{q_3}$ and 
  we get:
  \begin{align*}
    \C{G}{Y_4}{\LL{q_1}{q_4}} &= 
    \big( \C{F}{Y_2}{\LR{q_1}{q_2}} \cdot (\SimpleFun{F}{\varepsilon}) \big) \odot
    \big( (\SimpleFun{F}{\varepsilon}) \cdot \C{F}{Y_2}{\LL{q_2}{q_3}} \big) 
    \\ & \hspace{20mm}{}\odot
    \big( \C{F}{Y_2}{\RL{q_3}{q_4}} \cdot (\SimpleFun{F}{\varepsilon}) \big) 
    \\
    &= (\SimpleFun{a^{+}b}{\varepsilon}) \cdot (\SimpleFun{a}{a})^{+} \cdot 
    (\SimpleFun{b}{b}) \,.
  \end{align*}
  Similarly, we compute
  \begin{align*}
    \C{G}{Y_4}{\RR{q_4}{q_0}} &= 
    \big( (\SimpleFun{F}{\varepsilon}) \cdot \C{F}{Y_2}{\RL{q_4}{q_5}} \big) \odot
    \big( \C{F}{Y_2}{\RR{q_5}{q_6}} \cdot (\SimpleFun{F}{\varepsilon}) \big) 
    \\ & \hspace{20mm}{}\odot
    \big( (\SimpleFun{F}{\varepsilon}) \cdot \C{F}{Y_2}{\LR{q_6}{q_0}} \big) 
    \\
    &= (\SimpleFun{a^{+}b}{\varepsilon}) \cdot (\SimpleFun{a}{a})^{+} \cdot 
    (\SimpleFun{b}{b}) \,.
  \end{align*}
\begin{figure}[t]
\begin{center}
\includegraphics[scale=0.7]{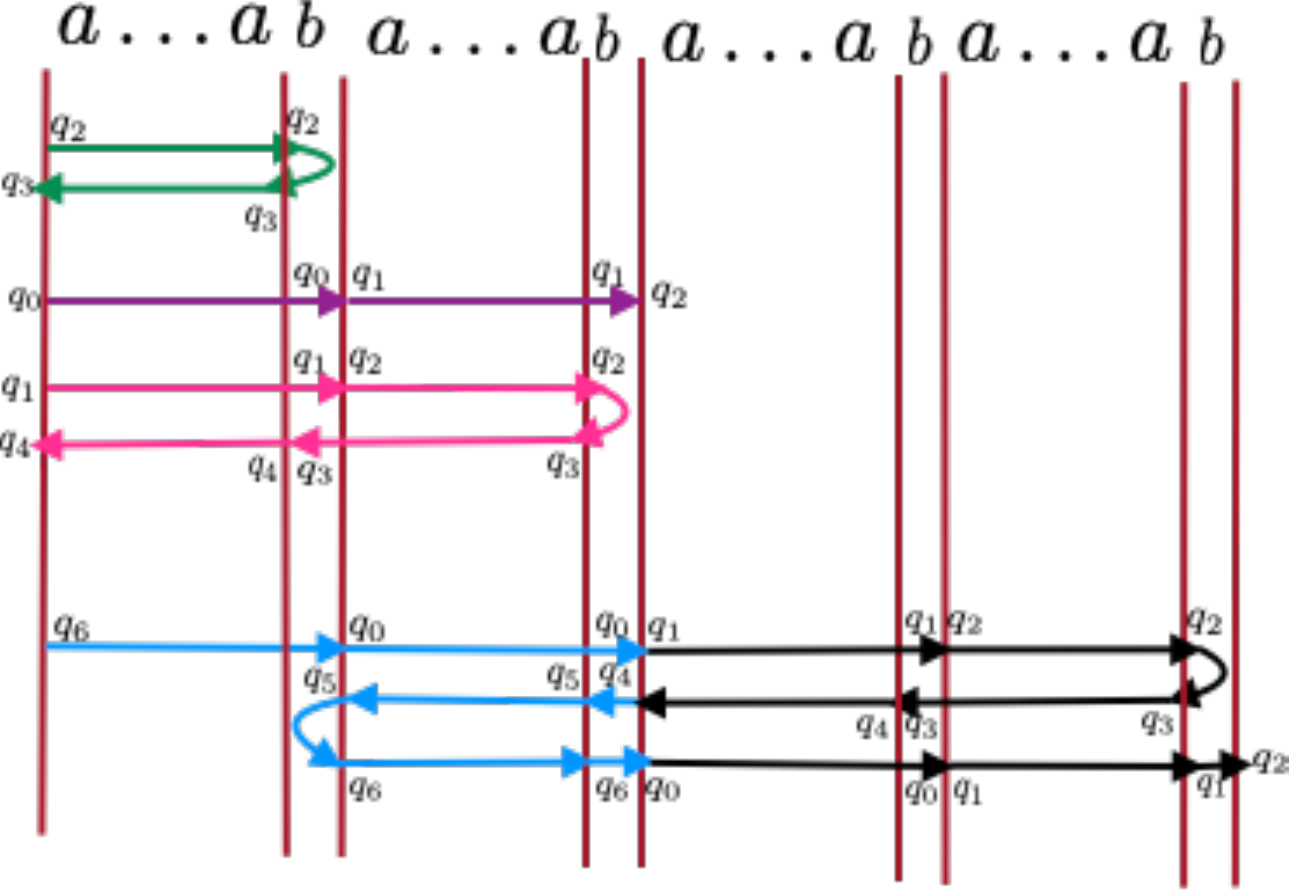}	
\end{center}
\caption{Illustration for Example \ref{eg:product}.}
\end{figure}

Finally, we consider the unambiguous decomposition $E=G\cdot G$ in order to compute
$\C{E}{Z_2}{x}$ where $x=\LR{q_6}{q_2}$.  Notice that $\LR{q_6}{q_1}, \LL{q_1}{q_4},
\RR{q_4}{q_0}, \LR{q_0}{q_2}\in Y_4=\varphi(G)$.  Hence,
$Z_{Y_4,Y_4}(q_1,\leftleft)=(\rightright,q_0)$ and applying the formulas in the proof
above we obtain
\begin{align*}
  \C{E}{Z_2}{x} &= 
  \big( \C{G}{Y_4}{\LR{q_6}{q_1}}\cdot (\SimpleFun{G}{\varepsilon}) \big) 
  \odot \ZF{G,Y_4}{G,Y_4}{q_1,\leftleft} \odot 
  \big( (\SimpleFun{G}{\varepsilon}) \cdot \C{G}{Y_4}{\LR{q_0}{q_2}} \big) 
  \\
  \ZF{G,Y_4}{G,Y_4}{q_1,\leftleft} &= 
  \big( (\SimpleFun{G}{\varepsilon}) \cdot \C{G}{Y_4}{\LL{q_1}{q_4}} \big) \odot
  \big( \C{G}{Y_4}{\RR{q_4}{q_0}}\cdot (\SimpleFun{G}{\varepsilon}) \big) \,.
\end{align*}
Putting everything together and still after some simplifications, we get
\begin{align*}
  \C{E}{Z_2}{x} &= 
  \big( (\SimpleFun{a}{a})^{+} \cdot (\SimpleFun{b}{b}) \cdot (\SimpleFun{(a^{+}b)^{3}}{\varepsilon}) \big) 
  \odot
  \big( (\SimpleFun{(a^{+}b)^{3}}{\varepsilon}) \cdot (\SimpleFun{a}{a})^{+} \cdot (\SimpleFun{b}{b}) \big) 
  \\
  &\hspace{20mm} {}\odot   
  \big( (\SimpleFun{a^{+}b}{\varepsilon}) \cdot (\SimpleFun{a}{a})^{+} \cdot 
  (\SimpleFun{b}{b}) \cdot (\SimpleFun{(a^{+}b)^{2}}{\varepsilon}) \big)
  \,.
\end{align*}
For instance, applying $\C{E}{Z_2}{x}$ to $w=aba^{2}ba^{3}ba^{4}b$ we obtain
$aba^{4}ba^{2}b$. \qedhere

\end{example}

\paragraph*{SD-Kleene Star}\label{sec:SD-star}
The most interesting case is when $E=F^{*}$. Let 
$L=\lang{F}\subseteq\Sigma^{*}$. Since $E$ is a stabilising SD-regular expression, $L$ 
is an aperiodic prefix code of bounded synchronisation delay, and $X=\varphi(L)$ is 
$k$-stabilising for some $k>0$. Hence, we may apply the results of 
Section~\ref{sec:run-stable}.

By induction, we suppose that we have  \SDRTEs 
$\C{F}{s}{x}$ for all $s$ in $\TrMon$ and steps $x$.  Since $L=\lang{F}$ is
a code, for each fixed $\ell>0$, the expression $F^{\ell}=F\cdot F \cdots F$ is
an unambiguous concatenation.  Hence, from the proof above for the unambiguous 
concatenation,
we may also assume that we have \SDRTEs $\C{F^{\ell}}{s}{x}$ for all $s\in\TrMon$ and
steps $x$. Similarly, we have \SDRTEs for $\ZF{F^{k},s}{F,t}-$ and $\ZF{F,s}{F^{k},t}-$.
Notice that $F^{0}$ is equivalent to $\varepsilon$ hence we have
$\C{F^{0}}{s}{x}=\C{\varepsilon}{s}{x}$.

We show how to construct  \SDRTEs $\C{E}{s}{x}$ for $E=F^{*}$.
There are four cases, which are dealt with below, depending on the step $x$.

We fix some notation common to all four cases.  Fix some $w\in\lang{E,s}=
L^{*}\cap\varphi^{-1}(s)$ and let $w=u_1\cdots u_n$ be the unique
factorization of $w$ with $n\geq0$ and $u_i\in L$ for $1\leq i\leq n$.  For a
step $x\in s$, we denote by $\rho$ the unique run of $\A$ over $w$ following
step $x$.

\begin{itemize}
  \item $x=\LL{p}{q}\in s$
  
  The easiest case is for left-left steps. 
  If $n<k$ then the output of $\rho$ is $\sem{\C{F^{n}}{s}x}(w)$. Notice that 
  here, $\C{F^{0}}{s}{x}=\bot$ since $x\notin\mathbf{1}=\varphi(\varepsilon)$. Now, if
  $n\geq k$ then, by Lemma~\ref{lem:LL}, the run $\rho$ stays inside $u_1\cdots
  u_k$.  We deduce that the output of $\rho$ is $\sem{\C{F^{k}}{s}x}(u_1\cdots
  u_k)$.  Therefore, we define
  \begin{align*}
    \C{E}{s}x & = \Big( \sum_{n<k} \C{F^{n}}{s}x \Big) + 
    \Big(\C{F^{k}}{s}x\cdot(\SimpleFun{F^{*}}{\varepsilon})\Big)
  \end{align*}
  Notice that the sums are unambiguous since $L=\lang{F}$ is a code. The 
  concatenation $F^{k}\cdot F^{*}$ is also unambiguous.
  
  \item $x=\LR{p}{q}\in s$
  
  We turn now to the more interesting left-right steps. 
  Again, if $n<k$ then the output of $\rho$ is $\sem{\C{F^{n}}{s}x}(w)$.  Assume
  now that $n\geq k$.  We apply Lemma~\ref{lem:LR} to deduce that the run $\rho$
  is $k$-forward-progressing.  See Figure~\ref{fig:LR-2progressing} for a sample
  run which is $2$-forward-progressing.  We split $\rho$ in
  $\rho_0\rho_1\cdots\rho_{n-k}$ where $\rho_0$ is the prefix of $\rho$ going
  until the first crossing from $u_k$ to $u_{k+1}$.  Then, $\rho_1$ goes until
  the first crossing from $u_{k+1}$ to $u_{k+2}$.  Continuing in the same way,
  for $1\leq i<n-k$, $\rho_i$ goes until the first crossing from $u_{k+i}$ to
  $u_{k+i+1}$.  Finally, $\rho_{n-k}$ is the remaining suffix, going until the
  run exits from $w$ on the right.  Since the run $\rho$ is $k$-forward
  progressing, we deduce that $\rho_i$ does not go back to $u_1\cdots u_{i-1}$,
  hence it stays inside $u_i\cdots u_{i+k}$, starting on the left of $u_{i+k}$
  and exiting on the right of $u_{i+k}$.

  Since $X=\varphi(L)$ is $k$-stabilising, we have $\varphi(u_1\cdots
  u_{k+i})=\varphi(w)$ for all $0\leq i\leq n-k$.  Now, $\rho_0\cdots\rho_i$ is a
  run on $u_1\cdots u_{k+i}$ starting on the left in state $p$ and exiting on the
  right.  Since $\A$ is deterministic and
  $x=\LR{p}{q}\in\varphi(w)=\varphi(u_1\cdots u_{k+i})$ we deduce that $\rho_i$
  exits on the right of $u_{k+i}$ in state $q$.  In particular, $\rho_0$ is a run
  on $u_1\cdots u_k$ starting on the left in state $p$ and exiting on the right in
  state $q$.  Moreover, for each $1\leq i\leq n-k$, $\rho_i$ is the
  concatenation of a zigzag internal run over $(u_i\cdots u_{i+k-1},u_{i+k})$
  starting with $(q,\leftleft)$ ending with
  $(\rightright,q_i)=Z_{s',s''}(q,\leftleft)$ where $s'=\varphi(u_i\cdots
  u_{i+k-1})$, $s''=\varphi(u_{i+k})$ and a $\LR{q_i}{q}$ run over $u_{i+k}$.
  
  \begin{figure}[t]
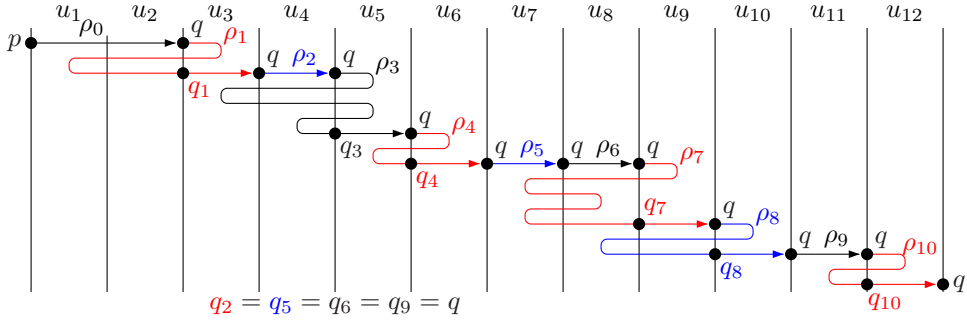

    \centering
    \gusepicture{gpic:LR-run2}
    \caption{A left-right run which is 2-forward-progressing.}
    \label{fig:LR-2progressing}
  \end{figure}
  
  Let $v_i$ be the output produced by $\rho_i$ for $0\leq i\leq n-k$. 
  Then, using Lemma~\ref{lem:InternalRun}, the productions $v_i$ with $0<i\leq n-k$ are
  given by the \SDRTE $f$ defined as
  \begin{align*}
    f=\hspace{-7mm}\sum_{\substack{s',s'',q'\,\mid \\ 
    (\rightright,q')=Z_{s',s''}(q,\leftleft)}}\hspace{-7mm}
    \ZF{F^{k},s'}{F,s''}{q,\leftleft} \odot
    \big((\SimpleFun{\lang{F^{k},s'}}{\varepsilon}) \cdot \C{F}{s''}{\LR{q'}{q}}\big)
  \end{align*}
  Then the product $v_1\cdots v_{n-k}$ is produced by the $(k+1)$-chained
  Kleene-star $\kstar{(k+1)}{L}{f}(w)$.  From the above discussion, we also deduce
  that $v_0=\sem{\C{F^{k}}{s}x}(u_1\cdots u_k)$.  Therefore, we define
  \begin{align*}
    \C{E}{s}x & = \Big( \sum_{n<k} \C{F^{n}}{s}x \Big) + 
    \Big(\big(\C{F^{k}}{s}x\cdot(\SimpleFun{F^{*}}{\varepsilon})\big)
    \odot \kstar{(k+1)}{F}{f}\Big)
  \end{align*}

  \item $x=\RL{p}{q}\in s$
  
  The case of right-left runs is almost symmetric to the
  case of left-right runs.    
  Again, if $n<k$ then the
  output of $\rho$ is $\sem{\C{F^{n}}{s}x}(w)$.  Assume now that $n\geq k$.  We
  apply Lemma~\ref{lem:RL} to deduce that the run $\rho$ is
  $k$-backward-progressing.  As illustrated in Figure~\ref{fig:RL-2progressing},
  we split $\rho$ in $\rho_{n-k+1}\cdots\rho_2\rho_1$ where $\rho_{n-k+1}$ is
  the prefix of $\rho$ going until the first crossing from $u_{n-k+1}$ to
  $u_{n-k}$.  Then, $\rho_{n-k}$ goes until the first crossing from $u_{n-k}$ to
  $u_{n-k-1}$.  Continuing in the same way, for $n-k>i>1$, $\rho_i$ goes until
  the first crossing from $u_{i}$ to $u_{i-1}$.  Finally, $\rho_{1}$ is the
  remaining suffix, going until the run exits from $u_1$ on the left.  Since the
  run $\rho$ is $k$-backward progressing, we deduce that, for $1\leq i\leq n-k$,
  the run $\rho_i$ does not go back to $u_{i+k+1}\cdots u_{n}$. Hence it is the
  concatenation of a zigzag internal run over $(u_{i},u_{i+1}\cdots u_{i+k})$,
  starting with some $(q_i,\rightright)$ and exiting with 
  $(\leftleft,q'_i)=Z_{s',s''}(q_i,\rightright)$ where $s'=\varphi(u_i)$, 
  $s''=\varphi(u_{i+1}\cdots u_{i+k})$, and a $\RL{q'_i}{q_{i-1}}$-run over 
  $u_i$ (see again Figure~\ref{fig:RL-2progressing}).
  Let $v_i$ be the output produced by $\rho_i$ for $1\leq i\leq n-k+1$.  The
  output produced by $\rho$ is $v=v_{n-k+1}\cdots v_2v_1$.  Now, the situation
  is slightly more complicated than for left-right runs where we could prove
  that $q_i=q$ for each $i$.  Instead, let us remark that $q_i$ is the unique
  state (by determinacy of $\A$) such that there is a run over $u_{i+1}\cdots
  u_n$ following step $\RL{p}{q_i}$.  But since $X$ is $k$-stabilising, we know that
  $\varphi(u_{i+1}\cdots u_n)=\varphi(u_{i+1}\cdots u_{i+k})$. 
  Then, given $s'=\varphi(u_i)$ and
  $s''=\varphi(u_{i+1}\cdots u_{i+k})$, we get that $q_i$ is the unique state
  such that $\RL{p}{q_i}\in s''$ and $q_{i-1}$ the one such that
  $\RL{p}{q_{i-1}}\in s's''$.  Thus we define the function $g$ generating the
  $v_i$ with $1\leq i\leq n-k$ by
  $$
  g=\hspace{-7mm}\sum_{\substack{s',s'',p',q',q''\,\mid\, \RL{p}{p'}\in s'',\\ 
    (\leftleft,q')=Z_{s',s''}(p',\rightright),\,\RL{q'}{q''}\in s'}}\hspace{-7mm}
    \ZF{F,s'}{F^{k},s''}{p',\rightright} \odot
    \Big( \C{F}{s'}{\RL{q'}{q''}} \cdot (\SimpleFun{\lang{F^{k},s''}}{\varepsilon})\Big)
  $$

  \begin{figure}[t]
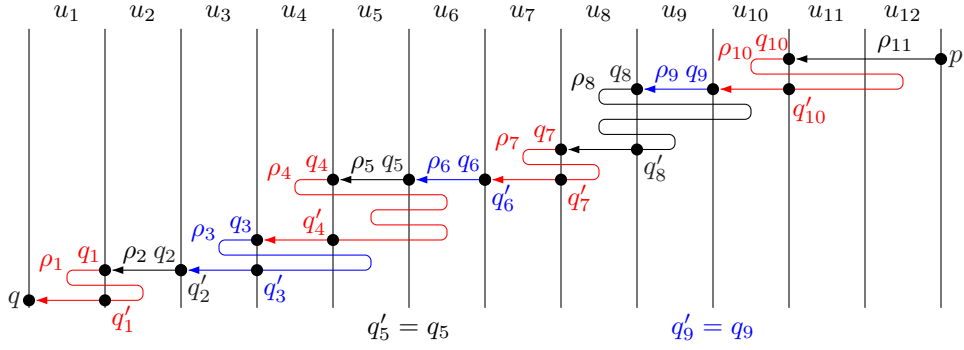

    \centering
    \gusepicture{gpic:RL-run2}
    \caption{A right-left run which is 2-backward-progressing.}
    \label{fig:RL-2progressing}
  \end{figure}

  Finally, a right-left run for $F^*$ is either a right-left run over $F^n$ for
  $n<k$, or the concatenation of a right-left run $\rho_{n-k+1}$ over the $k$
  rightmost iterations of $F$, and a sequence of runs $\rho_i$ with $1\leq i\leq
  n-k$ as previously and whose outputs are computed by $g$.  Therefore, we define
  $$
  \C{E}{s}x = \Big( \sum_{n<k} \C{F^{n}}{s}x \Big) + 
  \Big(\hspace{-2mm}\sum_{\substack{s=s's'',\,p'\,\mid \\ \RL{p}{p'}\in s''}} 
  \hspace{-4mm}
  (\SimpleFun{\lang{F^{*},s'}}{\varepsilon})\cdot\C{F^{k}}{s''}{\RL{p}{p'}}\Big)
  \odot \krstar{(k+1)}{F}{g}
  $$

  \item $x=\RR{p}{q}\in s$

  We finally deal with right-right runs, thus completing the case of $E=F^*$.  
  If $n<k$, then the output of $\rho$ is $\sem{\C{F^{n}}{s}x}(w)$.  Assume now
  that $n\geq k$.  If $\rho$ only visits the last $k$ factors, then $w$ can be
  decomposed as $uv$ where $u\in\lang{F^*,s'}$ and $v\in\lang{F^{k},s''}$ with
  $s=s's''$ and $\RR{p}{q}\in s''$.  The output of $\rho$ is
  $\sem{(\SimpleFun{\lang{F^*,s'}}{\varepsilon})\cdot\C{F^{k}}{s''}x}(w)$.
  
  Otherwise, by Lemma~\ref{lem:RR}, we know that there is a \emph{unique}
  integer $1\leq i\leq n-k$ such that $\rho$ only visits $u_i\cdots u_n$, and is
  the concatenation of 3 runs $\rho_1,\rho_2$ and $\rho_3$, where $\rho_1$ is a
  $k$-backward-progressing run over $u_{i+1}\cdots u_n$, $\rho_2$ is a zigzag
  internal run over $(u_i,u_{i+1}\cdots u_{i+k})$ and $\rho_3$ is a
  $k$-forward-progressing run over $u_{i+1}\cdots u_n$, as depicted in
  Figure~\ref{fig:RR-run}.
  
  Let $r=\varphi(u_1\cdots u_{i-1})$, $s'=\varphi(u_i)$ and
  $s''=\varphi(u_{i+1}\cdots u_{i+k})$. Since $i$ is uniquely determined, the 
  tuple $(r,s',s'')$ is also unique. Notice that $s''=\varphi(u_{i+1}\cdots
  u_n)$ since $X$ is $k$-stabilising and $s''\in X^{k}$, hence we have
  $s=rs's''$.  Moreover, the starting and ending states of $\rho_2$ are the
  \emph{unique} states $p',q'$ such that $\RL{p}{p'}\in s''$,
  $Z_{s',s''}(p',\rightright)=(\rightright,q')$ and $\LR{q'}{q}\in s''$.  Once
  the tuple $(r,s',s'',p',q')$ is fixed, using the previous points, we have
  \SDRTEs computing the outputs of $\rho_1$ and $\rho_3$:
  $$
  \sem{\C{E}{s''}{\RL{p}{p'}}}(u_{i+1}\cdots u_n)
  \hspace{10mm}
  \sem{\C{E}{s''}{\LR{q'}{q}}}(u_{i+1}\cdots u_n)
  $$
  and from Lemma~\ref{lem:InternalRun}, the output of $\rho_2$ is
  $\sem{\ZF{F,s'}{F^{k},s''}{p',\rightright}}(u_iu_{i+1}\cdots u_{i+k})$.
  This leads to the following \SDRTE:
  \begin{align*}
    h = \hspace{-5mm}\sum_{\substack{r,s',s'',p',q'\,\mid \\
    s=rs's'',\,\RL{p}{p'}\in s'' \\
    Z_{s',s''}(p',\rightright)=(\rightright,q') \\
    \LR{q'}{q}\in s''}}\hspace{-5mm}
    (\SimpleFun{\lang{F^{*},r}}{\varepsilon})\cdot
    \Big(&\big(\SimpleFun{\lang{F,s'}}{\varepsilon}\big)\cdot 
    \C{E}{s''}{\RL{p}{p'}} \\[-12mm]
    &\odot \ZF{F,s'}{F^{k},s''}{p',\rightright}\cdot(\SimpleFun{F^*}{\varepsilon}) \\
    &\odot (\SimpleFun{\lang{F,s'}}{\varepsilon})\cdot \C{E}{s''}{\LR{q'}{q}}\Big)
  \end{align*}
  Notice that if $r,s',s'',p',q'$ satisfy the conditions of the sum, then there 
  is a unique index $i$ with $\varphi(u_i)=s'$ and $\varphi(u_{i+1}\cdots 
  u_n)=s''$: there is a $\RL{p}{p'}$-run on $u_{i+1}\cdots u_n$ but not on 
  $u_i\cdots u_n$. Hence all the Cauchy products in $h$ are unambiguous and 
  uniquely decompose $w$ 
  in $u_1\cdots u_{i-1}$ matched by $\lang{F^{*},r}$, $u_i$ matched by 
  $\lang{F,s'}$ and $u_{i+1}\cdots u_n$ matched by $\C{E}{s''}{\RL{p}{p'}}$ and 
  $\C{E}{s''}{\LR{q'}{q}}$. Also, $u_iu_{i+1}\cdots u_{i+k}$ is matched by 
  $\ZF{F,s'}{F^{k},s''}{p',\rightright}$ and $n\geq i+k$.

  Finally, for the right-right step $x=\RR{p}{q}$ we define 
  \begin{align*}
    \C{E}{s}{x} = & \Big(\sum_{n< k} \C{F^n}{s}{x}\Big) + 
    \Big( \sum_{s=s's''} (\SimpleFun{\lang{F^*,s'}}{\varepsilon})\cdot \C{F^k}{s''}{x}\Big) 
    + h \,.
  \end{align*}
  The sums above are unambiguous. Indeed,
  the first case happens exactly when $w\in F^n$ for $n<k$.  The second happens
  exactly when $n\geq k$ and the right-right run does not visit $u_{n-k}$.
  Otherwise, the third case happens.  
  \qedhere
\end{itemize}
\end{proof}

\begin{example}\label{eg:kstar}	
  Let us continue with our example in Figure~\ref{fig:eg-2dft}.  Here we illustrate
  computing an \SDRTE for some $E=F^{*}$ and some left-right step.  We consider $F=a^{+}b$
  so that $\varphi(F)=Y_2$ as computed in Example~\ref{eg:stabilising}, where we also
  computed $Z_2=Y_2^{4}$.  Consider $x=\LR{q_6}{q_2}\in Z_2$.  We explain how to compute
  $\C{E}{Z_2}x$.  Since $\varphi(F)=Y_2$ is 4-stabilising, the runs following step $x$
  over words in $\lang{E,Z_2}=(a^{+}b)^{\geq4}$ are 4-forward-progressing, and the
  construction in the proof above uses a 5-chained star\footnote{Actually, the runs over
  step $x$ are 3-forward-progressing. Hence, we could simplify the example below by using a
  4-chained star only. But we decided to use a 5-star in order to follow the construction 
  described in the proof above.}.
  
  Let $w=u_1u_2\cdots u_n$ with $n\geq4$ and $u_1,u_2,\ldots,u_n\in a^{+}b$.  As can be
  seen on Figure~\ref{fig:star} (with $n=8$), the 4-forward-progressing run $\rho$ over $w$
  following step $x=\LR{q_6}{q_2}$ splits as $\rho=\rho_0\cdots\rho_4$ where $\rho_0$ is
  the prefix from $u_1$ till first crossing from $u_4$ to $u_5$, $\rho_1$ is the
  part from $u_5$ till the first crossing from $u_5$ to $u_6$, 
  $\rho_2$ is the part of the run from $u_6$ till the first crossing from $u_6$ to $u_7$,
  $\rho_3$ is the part of the run from $u_7$ till the first crossing from $u_7$ to $u_8$,
  and $\rho_4$ is the part from $u_8$ till exiting at the right of $u_8$.
  
  We obtain
  \begin{align*}
    \C{(a^+b)^*}{Z_2}{\LR{q_6}{q_2}} &= (\C{(a^+b)^4}{Z_2}{\LR{q_6}{q_2}} \cdot 
    (\SimpleFun{(a^+b)^*}{\epsilon})) \odot \kstar{5}{a^+b}{f}
    \\
    \intertext{where}
    f & = \ZF{(a^+b)^4,Z_2}{a^+b,Y_2}{q_2,\leftleft} \odot
    \Big( (\SimpleFun{a^+b)^4}{\epsilon}) \cdot \C{a^+b}{Y_2}{\LR{q_1}{q_2}} \Big)
    \\
    \ZF{(a^+b)^4,Z_2}{a^+b,Y_2}{q_2,\leftleft} & =
    \Big( (\SimpleFun{(a^+b)^4}{\epsilon}) \cdot \C{a^+b}{Y_2}{\LL{q_2}{q_3}} \Big)
    \odot
    \Big( \C{(a^+b)^{4}}{Z_2}{\RR{q_3}{q_1}} \cdot (\SimpleFun{a^+b}{\epsilon}) \Big)
    \\
    \intertext{and the expressions below where computed in Example~\ref{eg:product}}
    \C{a^+b}{Y_2}{\LR{q_1}{q_2}} &= \SimpleFun{a^+b}{\epsilon}
    \\
    \C{a^+b}{Y_2}{\LL{q_2}{q_3}} &= (\SimpleFun{a}{a})^{+}\cdot(\SimpleFun{b}{b})
    \\
    \C{(a^+b)^{4}}{Z_2}{\RR{q_3}{q_1}} &= (\SimpleFun{(a^+b)^{2}}{\epsilon}) \cdot
    (\SimpleFun{a}{a})^{+}\cdot(\SimpleFun{b}{b}) \cdot (\SimpleFun{a^+b}{\epsilon})
    \\
    \C{(a^+b)^{4}}{Z_2}{\LR{q_6}{q_2}} &= 
    \Big( (\SimpleFun{a}{a})^{+}\cdot(\SimpleFun{b}{b}) \cdot
    (\SimpleFun{(a^+b)^{2}}{\epsilon}) \cdot
    (\SimpleFun{a}{a})^{+}\cdot(\SimpleFun{b}{b}) \Big) \odot{}
    \\ & \qquad
    \Big( (\SimpleFun{a^+b}{\epsilon}) \cdot 
    (\SimpleFun{a}{a})^{+}\cdot(\SimpleFun{b}{b}) \cdot
    (\SimpleFun{(a^+b)^{2}}{\epsilon}) \Big) 
  \end{align*}
      \begin{figure}[t]
  \begin{center}
  	\includegraphics[scale=0.6]{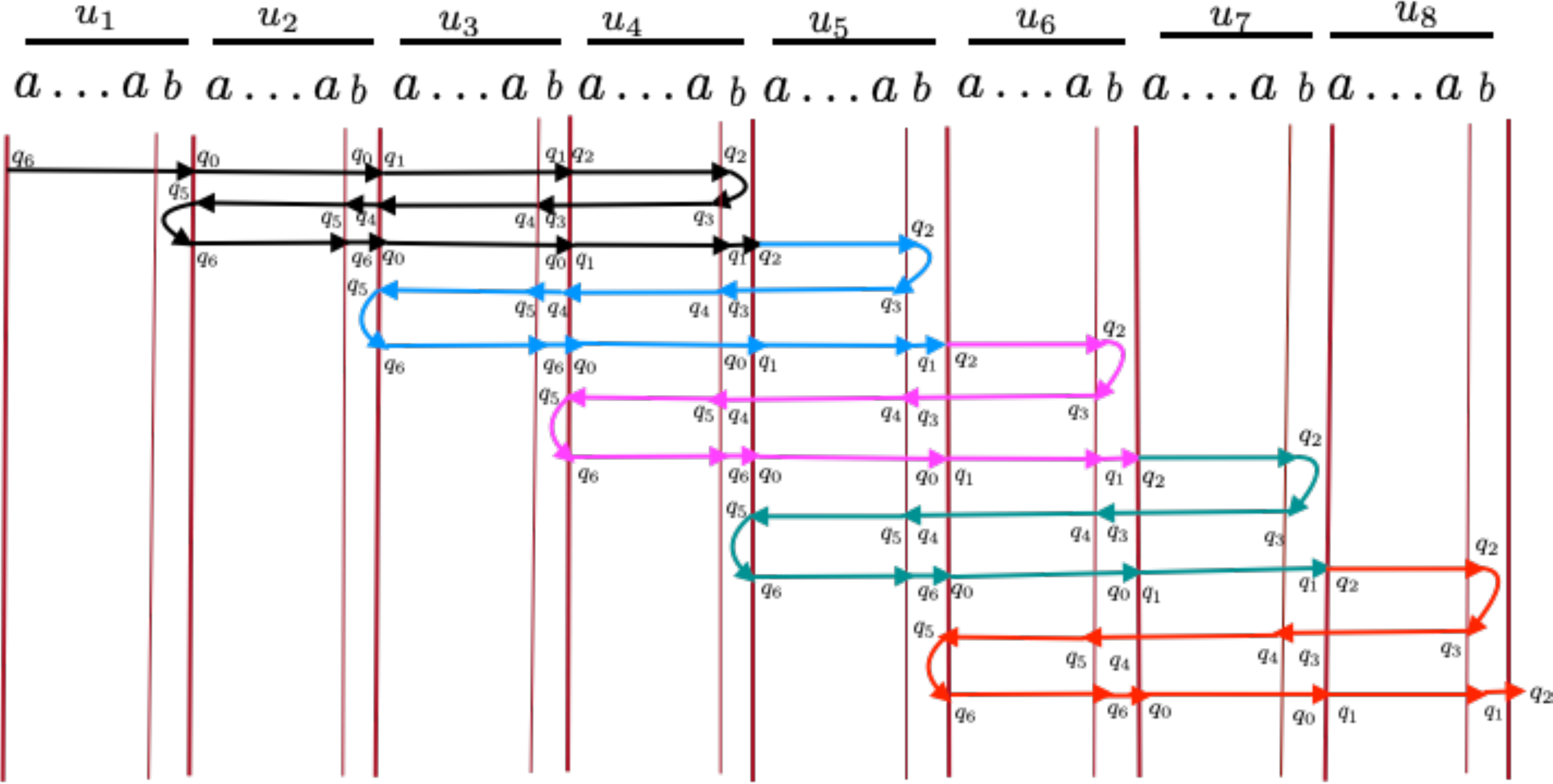}
  \end{center}
  	\caption{Illustration for Example \ref{eg:kstar}}
  	\label{fig:star}
  \end{figure}

  After simplifications, we obtain
  \begin{align*}
    f &= \ZF{(a^+b)^4,Z_2}{a^+b,Y_2}{q_2,\leftleft} 
    \\
    &=
    \Big( (\SimpleFun{(a^+b)^4}{\epsilon}) \cdot 
    (\SimpleFun{a}{a})^{+} \cdot (\SimpleFun{b}{b}) \Big)
    \odot
    \Big( (\SimpleFun{(a^+b)^{2}}{\epsilon}) \cdot
    (\SimpleFun{a}{a})^{+}\cdot(\SimpleFun{b}{b}) \cdot 
    (\SimpleFun{(a^+b)^{2}}{\epsilon}) \Big) \,.
  \end{align*}
  
  For instance, consider $w=u_1u_2\cdots u_8$ with $u_i=a^{i}b$. Then,
  \begin{align*}
    f(u_1\cdots u_{5}) &= a^{5}ba^{3}b
    &
    f(u_2\cdots u_{6}) &= a^{6}ba^{4}b
    \\
    f(u_3\cdots u_{7}) &= a^{7}ba^{4}b
    &
    f(u_4\cdots u_{8}) &= a^{8}ba^{5}b
    \\
    \sem{\C{(a^{+}b)^{4}}{Z_2}{\LR{q_6}{q_2}}}(u_1\cdots u_4) &= aba^{4}ba^{2}b 
    \\
    \sem{\C{(a^{+}b)^{*}}{Z_2}{\LR{q_6}{q_2}}}(u_1\cdots u_8) &= 
    \makebox[10mm][l]{$aba^{4}ba^{2}ba^{5}ba^{3}ba^{6}ba^{4}ba^{7}ba^{4}ba^{8}ba^{5}b$\,.}  
  \end{align*}
\end{example}

We conclude the section by showing how to construct \SDRTEs equivalent to 2DFTs.

\begin{theorem}\label{thm:final}
  Let $\A=(Q,\Sigma,\Gamma,\delta,\gamma,q_0,F)$ be an aperiodic 2DFT. We can construct an
  equivalent \SDRTE $C_{\A}$ over alphabet $\Sigma$ with
  $\dom{\sem{C_\A}}=\dom{\sem{\A}}$ and $\sem{\A}(w)=\sem{C_\A}(w)$ for all
  $w\in\dom{\sem{\A}}$.
\end{theorem}

\begin{proof}
  We first construct below an \SDRTE $C'_{\A}$ with
  $\dom{\sem{C'_\A}}=\leftend\dom{\sem{\A}}\rightend$ and such that
  $\sem{\A}(w)=\sem{C'_\A}(\leftend w \rightend)$ for all $w\in\dom{\sem{\A}}$.
  Then, we obtain $C''_\A$ using Proposition~\ref{prop:LQ-SDRTE} by
  $$
  C''_A=\RQ{\rightend}{(\LQ{\leftend}{C'_\A})} \,.
  $$
Finally, we get rid of lingering endmarkers in $C''_\A$ using Lemma~\ref{lem-projection} to obtain $C_\A$ as the projection of $C''_\A$ on $\Sigma^*$.
  
  Let $\varphi\colon(\Sigma\uplus\{\leftend,\rightend\})^{*}\to\TrMon$ be the
  canonical surjective morphism to the transition monoid of $\A$.  Since $\A$ is
  aperiodic, the monoid $\TrMon$ is also aperiodic.  We can apply
  Theorem~\ref{thm:U-S-SD-expressions} to the restriction of $\varphi$ to
  $\Sigma^{*}$: for each $s\in\TrMon$, we get an unambiguous, stabilising,
  SD-regular expression $E_s$ with $\lang{E_s}=\varphi^{-1}(s)\cap\Sigma^{*}$.
  Let $E=\leftend\cdot(\bigcup_{s\in\TrMon}E_s)$ which is an unambiguous,
  stabilising, SD-regular expression with $\lang{E}=\leftend\Sigma^{*}$.
  Applying Theorem~\ref{thm:main}, for each monoid element $s\in\TrMon$ and each
  step $x\in\{\leftright,\leftleft,\rightright,\rightleft\}\times Q^{2}$, we
  construct the corresponding \SDRTE $\C{E}{s}x$.  We also apply
  Lemma~\ref{lem:InternalRun} and construct for each state $p\in Q$ an \SDRTE
  $\ZF{E,s}{\rightend,t}{p,\leftleft}$ where $t=\varphi(\rightend)$.

\begin{figure}[h]
\begin{center}
\includegraphics[scale=.6]{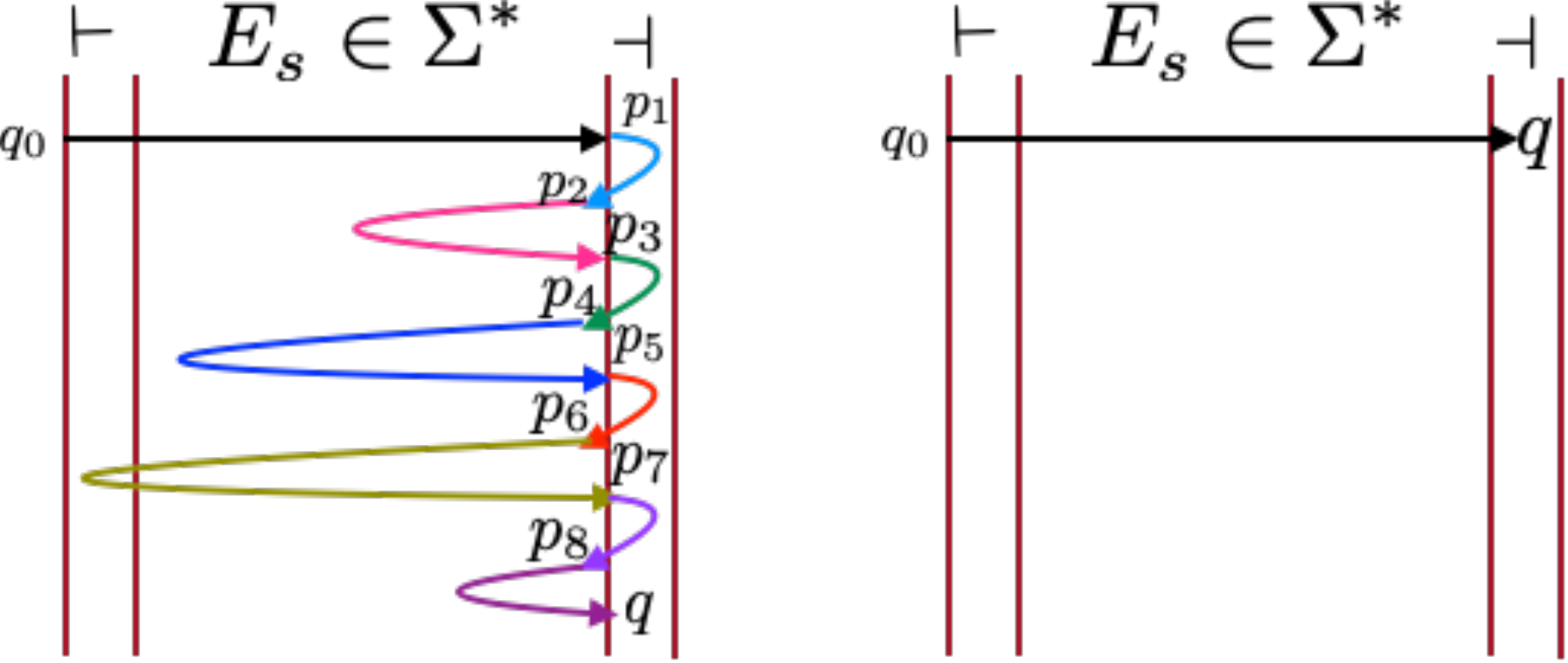}
\end{center}
\caption{Removing end markers. On the left when there is a non trivial zig zag until reaching a final state $q$; on the right when we have an empty zigzag. $q_0$ is the initial state and $q \in F$.}
\label{fig:final}
\end{figure}

  Finally, we define
  $$
  C'_\A=\hspace{-5mm}\sum_{\substack{s,p,q\mid q\in F \\ \LR{q_0}{p}\in s \\ 
  Z_{s,t}(p,\leftleft)=(\rightright,q)} }
  \big( \C{E}{s}{\LR{q_0}{p}}\cdot(\SimpleFun{\rightend}{\varepsilon}) \big)
  \odot \ZF{E,s}{\rightend,t}{p,\leftleft}
  $$

See alo Figure \ref{fig:final} illustrating $C'_{\A}$.   
  We can easily check that $C'_\A$ satisfies the requirements stated above.
\end{proof}

\begin{example}\label{eg:last}	
  We complete the series of examples by giving an \SDRTE equivalent with the transducer 
  $\A$ of Figure~\ref{fig:eg-2dft} on words in $E_1=b(a^{+}b)^{\geq4}\subseteq\dom{\A}$.
  Notice that by Example~\ref{eg:stabilising}, we have $\varphi(E_1)=Z_1=Y_1Z_2$.
  We compute first $\C{E_1}{Z_1}{\LR{s}{q_2}}$. We use the unambiguous product 
  $E_1=b\cdot E_2$ with $E_2=(a^{+}b)^{\geq4}$. From Example~\ref{eg:stabilising}, we have
  $\LR{s}{q_0},\RR{q_5}{q_6}\in Y_1=\varphi(b)$ and $\LL{q_0}{q_5},\LR{q_6}{q_2}\in 
  Z_2=\varphi(E_2)$. We deduce that in the product, the zigzag part consists of the two 
  steps $\LL{q_0}{q_5}$ and $\RR{q_5}{q_6}$. Therefore we obtain:
  \begin{align*}
    \C{E_1}{Z_1}{\LR{s}{q_2}} &= 
    \big( \C{b}{Y_1}{\LR{s}{q_0}}\cdot (\SimpleFun{E_2}{\varepsilon}) \big) 
    \odot \ZF{b,Y_1}{E_2,Z_2}{q_0,\leftleft} \odot 
    \big( (\SimpleFun{b}{\varepsilon}) \cdot \C{E_2}{Z_2}{\LR{q_6}{q_2}} \big) 
    \\
    \ZF{b,Y_1}{E_2,Z_2}{q_0,\leftleft} &= 
    \big( (\SimpleFun{b}{\varepsilon}) \cdot \C{E_2}{Z_2}{\LL{q_0}{q_5}} \big) \odot
    \big( \C{b}{Y_1}{\RR{q_5}{q_6}}\cdot (\SimpleFun{E_2}{\varepsilon}) \big) \,.
  \end{align*}
  Since $\C{b}{Y_1}{\LR{s}{q_0}}=\C{b}{Y_1}{\RR{q_5}{q_6}}=\SimpleFun{b}{\varepsilon}$, 
  we obtain after simplifications
  \begin{align*}
    \C{E_1}{Z_1}{\LR{s}{q_2}} &= 
    \big( (\SimpleFun{b}{\varepsilon}) \cdot \C{E_2}{Z_2}{\LL{q_0}{q_5}} \big) \odot
    \big( (\SimpleFun{b}{\varepsilon}) \cdot \C{E_2}{Z_2}{\LR{q_6}{q_2}} \big) \,.
  \end{align*}
  Let $E=(a^{+}b)^{*}$ as in Example~\ref{eg:kstar}.
  Since $\lang{E_2,Z_2}=\lang{E,Z_2}$ we have $\C{E_2}{Z_2}x=\C{E}{Z_2}x$ for all steps $x$.
  Recall that in Example~\ref{eg:kstar} we have computed $\C{E}{Z_2}{\LR{q_6}{q_2}}$.
  Moreover, $\C{E_2}{Z_2}{\LL{q_0}{q_5}}=\C{(a^{+}b)^{4}}{Z_2}{\LL{q_0}{q_5}}\cdot 
  (\SimpleFun{E}{\varepsilon})$ and a computation similar to Example~\ref{eg:product} 
  gives
  \begin{align*}
    \C{(a^{+}b)^{4}}{Z_2}{\LL{q_0}{q_5}} &= 
    (\SimpleFun{(a^{+}b)^{2}}{\varepsilon}) \cdot 
    (\SimpleFun{a}{a})^{+} \cdot (\SimpleFun{b}{b}) \cdot (\SimpleFun{a^{+}b}{\varepsilon}) \,.
  \end{align*}
  Finally, with the notations of Example~\ref{eg:kstar}, we obtain
  \begin{align*}
    \C{E_1}{Z_1}{\LR{s}{q_2}} ={}\phantom{\odot}&
    \big( (\SimpleFun{b(a^{+}b)^{2}}{\varepsilon}) \cdot 
    (\SimpleFun{a}{a})^{+} \cdot (\SimpleFun{b}{b}) \cdot
    (\SimpleFun{(a^{+}b)^{+}}{\varepsilon}) \big) 
    \\
    {}\odot{}& 
    \big( (\SimpleFun{b}{\varepsilon}) \cdot \C{E}{Z_2}{\LR{q_6}{q_2}} \big) 
    \\
    ={}\phantom{\odot}&
    \big( (\SimpleFun{b(a^{+}b)^{2}}{\varepsilon}) \cdot 
    (\SimpleFun{a}{a})^{+} \cdot (\SimpleFun{b}{b}) \cdot
    (\SimpleFun{(a^{+}b)^{+}}{\varepsilon}) \big) 
    \\
    {}\odot{}& 
    \big( (\SimpleFun{b}{\varepsilon}) \cdot \C{(a^+b)^4}{Z_2}{\LR{q_6}{q_2}} \cdot 
    (\SimpleFun{(a^+b)^*}{\epsilon}) \big)
    \\
    {}\odot{}& 
    \big( (\SimpleFun{b}{\varepsilon}) \cdot \kstar{5}{a^+b}{f} \big)
    \\
    ={}\phantom{\odot}&
    \big( (\SimpleFun{b(a^{+}b)^{2}}{\varepsilon}) \cdot 
    (\SimpleFun{a}{a})^{+} \cdot (\SimpleFun{b}{b}) \cdot
    (\SimpleFun{(a^{+}b)^{+}}{\varepsilon}) \big) 
    \\
    {}\odot{}& 
    \big( (\SimpleFun{b}{\varepsilon}) \cdot 
    (\SimpleFun{a}{a})^{+}\cdot(\SimpleFun{b}{b}) \cdot
    (\SimpleFun{(a^+b)^{2}}{\epsilon}) \cdot
    (\SimpleFun{a}{a})^{+}\cdot(\SimpleFun{b}{b}) \cdot 
    (\SimpleFun{(a^+b)^*}{\epsilon}) \big)
    \\
    {}\odot{}& 
    \big( (\SimpleFun{ba^+b}{\epsilon}) \cdot 
    (\SimpleFun{a}{a})^{+}\cdot(\SimpleFun{b}{b}) \cdot
    (\SimpleFun{(a^+b)^{\geq2}}{\epsilon}) \big)
    \\
    {}\odot{}& 
    \big( (\SimpleFun{b}{\varepsilon}) \cdot \kstar{5}{a^+b}{f} \big)
    \,.
  \end{align*}
   
   Note that for our example, we have a simple case of using Theorem \ref{thm:final},
   since $\leftend$ is visited only once at state $s$, and there are no transitions
   defined on $\rightend$, i.e., $\varphi(\rightend)=\emptyset$. 
   So for $E'=\leftend E_1$, we have $\LR{s}{q_2}\in\varphi(E')$ and
   $\ZF{E',\varphi(E')}{\rightend,\emptyset}{q_2,\leftleft}=\epsilon$ (zigzag part empty
   as seen in the right of Figure \ref{fig:final}).  
   Thus, an expression equivalent to $\A$ on words in $E_1$
   effectively boils down to $\C{E_1}{Z_1}{\LR{s}{q_2}}$. 
   For the sake of simplicity, we stop the example here and do not provide the expression 
   on the full domain $\dom{\A}=b(a^{*}b)^{\geq2}a^{*}$.
   \qed
\end{example}

\section{Adding composition}\label{sec:composition}

\newcommand{\dup}{\mathsf{dup}}
\newcommand{\rev}{\mathsf{rev}}
\newcommand{\id}{\mathsf{id}}

Since \SDRTEs are used to define functions over words, it seems natural to consider the
composition of functions, as it is an easy to understand but powerful operator.  In this
section, we discuss other formalisms using composition as a basic operator, and having the
same expressive power as \SDRTEs.

Theorem~\ref{thm:intro} gives the equivalence between \SDRTEs and aperiodic two-way
transducers, the latter being known to be closed under composition.  Hence, adding
composition to \SDRTEs does not add expressiveness, while allowing for easier modelisation
of transformations.

Moreover, we prove that, should we add composition of functions, then we can replace the
$k$-chained star operator and its reverse by the simpler $1$-star $\kstar{1}{L}{f}$ and
its reverse, which in particular are one-way (left-to-right or right-to-left) operator
when $f$ is also one-way.

Finally, we prove that we can furthermore get rid of the reverse operators as well as the
Hadamard product by adding two basic functions: \emph{reverse} and \emph{duplicate}.  The
reverse function is straightforward as it reverses its input.  The duplicate function
is parameterised by a symbol, say $\#$, duplicates its input inserting $\#$ between the 
two copies:
$\dup_\#(u)=u\#u$.

\begin{theorem}\label{thm-withComp}
  The following families of expressions have the same expressive power:
\begin{enumerate}
\item\label{thm-wC-SDRTE} \SDRTEs,
\item\label{thm-wC-SDRTEc} \SDRTEs with composition of functions,
\item\label{thm-wC-1c} \SDRTEs with composition and chained star restricted to $1$.
\item\label{thm-wc-BasicFun} Expressions with simple functions, unambiguous sum, Cauchy product, $1$-star, duplicate, reverse and composition.
\end{enumerate}
\end{theorem}

\begin{proof}
  It is trivial that $\ref{thm-wC-1c}\subseteq \ref{thm-wC-SDRTEc}$ as $\ref{thm-wC-1c}$
  is obtained as a syntaxical restriction of $\ref{thm-wC-SDRTEc}$.  Although it is not
  needed in our proof, note that $\ref{thm-wC-SDRTE}\subseteq\ref{thm-wC-SDRTEc}$ holds
  for the same reason.
  Now, thanks to Theorem~\ref{thm:intro}, we know that \SDRTEs are equivalent to aperiodic
  two-way transducers that are closed under composition.  Hence, composition does not add
  expressive power and we have $\ref{thm-wC-SDRTEc}\subseteq \ref{thm-wC-SDRTE}$.

To prove that $\ref{thm-wc-BasicFun}\subseteq \ref{thm-wC-1c}$, we simply have to prove
that the duplicate and reverse functions can be expressed with \SDRTEs using only the
$1$-star operator and its reverse.
The duplicate function definition relies on the Hadamard
and is given by the expression: 
$$
\dup_\# = (\id_{\Sigma^{*}}\cdot(\SimpleFun{\varepsilon}{\#}))
\odot \id_{\Sigma^{*}}
$$
where $\id_{\Sigma^{*}}$ is the identity function and can be written as
$\kstar{1}{\Sigma}{\id_\Sigma}$ where $\id_\Sigma=\sum_{a\in\Sigma}\SimpleFun{a}{a}$.
The reverse function is also easy to define using the $1$-star reverse:
$$
\rev = \krstar{1}{\Sigma}{\id_\Sigma}
$$

To prove the last inclusion $\ref{thm-wC-SDRTE}\subseteq \ref{thm-wc-BasicFun}$, we need
to express the Hadamard product and the (reverse) $k$-chained star, using duplicate, reverse
and composition.  

The Hadamard product $f\odot g$ is easy to define using $\dup_\#$ where $\#$ is a 
fresh marker:
$$
f\odot g = (f\cdot(\SimpleFun{\#}{\varepsilon})\cdot g)\circ\dup_\# \,.
$$

We show now how to reduce $k$-star to $1$-star using duplicate and composition. The proof 
is by induction on $k$. When $k=1$ there is nothing to do. Assume that $k>1$. We show how 
to express $\kstar{k}{L}{f}$ using $(k-1)$-star, $1$-star and duplicate.
The main idea is to use composition to mark each factor in $L$ in order to duplicate them, then use a $(k-1)$-star
to have a reach of $k$ factors of $L$ (with some redundant information), and lastly use
composition to prune the input to a form suitable to finally apply $f$.

More formally, let $\#$ and $\$$ be two fresh markers and define
$$
f_1=
\kstar{1}{L}{\dup_\$\cdot(\SimpleFun{\varepsilon}{\#})}
$$ 
with domain $L^{*}$ and, when applied to a word $u=u_1\cdots u_n$ with $u_i\in L$, 
produces $u_1\$u_1\# u_2\$u_2\# u_3\$u_3\# \cdots u_{n-1}\$u_{n-1}\# u_n\$u_n\# 
\in\{\varepsilon\}\cup \Sigma^{*}\$(\Sigma^*\#\Sigma^*\$)^{*}\Sigma^{*}\#$.
Notice that $\Sigma^*\#\Sigma^*\$ $ is a $1$-SD prefix code and that taking $k-1$
consecutive factors from this language allows us to have a reach of $k$ factors of $L$.
Then we define the function $g$
$$
g= \big( \id_{\Sigma^*}\cdot(\SimpleFun{\#\Sigma^{*}\$}{\epsilon}) \big)^{k-2}
\cdot \big( \id_{\Sigma^*}\cdot (\SimpleFun{\#}{\epsilon})\cdot \id_{\Sigma^*}\cdot 
(\SimpleFun{\$}{\epsilon}) \big)
$$
with domain $(\Sigma^*\#\Sigma^*\$)^{k-1}$ and, when applied to a word
$v_1\#v'_1\$v_2\#v'_2\$ \cdots v_{k-1}\#v'_{k-1}\$ $, produces $v_1v_2\cdots 
v_{k-1}v'_{k-1}$. In particular, 
$g(u_{i+1}\# u_{i+2}\$u_{i+2}\# u_{i+3}\$ \cdots u_{i+k-1}\# u_{i+k}\$)=
u_{i+1}\cdots u_{i+k}$.
Finally, we have
$$
\kstar{k}{L}{f} =
\big( (\SimpleFun{\varepsilon}{\varepsilon}) 
+ (\SimpleFun{\Sigma^*\$}{\epsilon})
\cdot \kstar{(k-1)}{\Sigma^*\#\Sigma^*\$}{f\circ g}
\cdot (\SimpleFun{\Sigma^*\#}{\epsilon})\big)
\circ f_1 \,.
$$

The reverse $k$-star $\krstar{k}{L}{f}$ is not expressed in a straightforward fashion
using reverse composed with $k$-star, because while reverse applies on all the input, the
reverse $k$-star swaps the applications of function $f$ while keeping the function $f$
itself untouched.
In order to express it, we reverse a $k$-star operator not on $f$, but on $f$ reversed.
The result is that the applications of $f$ are reversed twice, thus preserving them.
Formally, we have:
$$
\krstar{k}{L}{f}= \rev\circ \kstar{k}{L}{\rev\circ f}
$$
\end{proof}

\section{Conclusion}
We conclude with some interesting avenues for future work, 
arising from the open questions based on this paper. 

We begin with complexity questions, and then move on to other directions for future work.
The complexity of our procedure, especially when going from the declarative language
\SDRTE to the computing machine 2DFT, is open.  This part relies heavily on the
composition of 2DFTs which incurs at least one exponential blowup in the state space. 
A possibility to reduce the complexity incurred during composition, is to obtain
\emph{reversible} 2FT (2RFT) for each of the intermediate functions used in the
composition.  2RFTs are a class of 2DFTs which are both deterministic and co-deterministic,
and were introduced in~\cite{DFJL17}, where they prove that composition of 2RFTs results in
a 2RFT with polynomially many states in the number of states of the input transducers.
Provided that the composition of 2RFTs preserves aperiodicity, if we could produce 2RFTs in
our procedures in section \ref{sec:A2DFTtoSDRET}, then we would construct a 2RFT which is
polynomial in the size of the \SDRTE. Another open question is the efficiency of
evaluation, i.e., given an \SDRTE and an input word, what is the time complexity of
obtaining the corresponding output.  This is crucial for an implementation, along the
lines of DReX \cite{DBLP:conf/popl/AlurDR15}.

Yet another direction is to extend our result to transformations over infinite words.
While Perrin \cite{DBLP:conf/mfcs/Perrin84} generalized the SF=AP result of Schützenberger
to infinite words in the mid 1980s, Diekert and Kufleitner \cite{DiekertK-CSR12,DiekertK-ToCS2015}
generalized Schützenberger's SD=AP result to infinite words.  One could use this SD=AP
over infinite words and check how to adapt our proof to the setting of transformations
over infinite words.  Finally, a long standing open problem in the theory of
transformations is to decide if a function given by a 2DFT is realizable by an aperiodic one. This question has been solved
in the one-way case, or in the case when we have \emph{origin} information
\cite{DBLP:conf/icalp/Bojanczyk14}, but the general case remains open.  We believe that
our characterisation of stabilising runs provided in Section~\ref{sec:run-stable} could
lead to a forbidden pattern criteria to decide this question.  

\bibliography{SFfunctions}

\begin{thebibliography}{10}

\bibitem{AC10}
Rajeev Alur and Pavol {\v{C}}ern{\'y}.
\newblock Expressiveness of streaming string transducers.
\newblock In {\em 30th {I}nternational {C}onference on {F}oundations of
  {S}oftware {T}echnology and {T}heoretical {C}omputer {S}cience, {FSTTCS}
  2010}, volume~8 of {\em LIPIcs. Leibniz Int. Proc. Inform.}, pages 1--12.
  Schloss Dagstuhl. Leibniz-Zent. Inform., Wadern, 2010.

\bibitem{DBLP:conf/popl/AlurDR15}
Rajeev Alur, Loris D'Antoni, and Mukund Raghothaman.
\newblock Drex: {A} declarative language for efficiently evaluating regular
  string transformations.
\newblock In {\em Proceedings of the 42nd Annual {ACM} {SIGPLAN-SIGACT}
  Symposium on Principles of Programming Languages, {POPL} 2015, Mumbai, India,
  January 15-17, 2015}, pages 125--137, 2015.
\newblock \href {https://doi.org/10.1145/2676726.2676981}
  {\path{doi:10.1145/2676726.2676981}}.

\bibitem{AlurFreilichRaghothaman14}
Rajeev Alur, Adam Freilich, and Mukund Raghothaman.
\newblock Regular combinators for string transformations.
\newblock In Thomas~A. Henzinger and Dale Miller, editors, {\em Joint Meeting
  of the 23rd {EACSL} Annual Conference on Computer Science Logic {(CSL)} and
  the 29th Annual {ACM/IEEE} Symposium on Logic in Computer Science (LICS),
  {CSL-LICS} '14, Vienna, Austria, July 14 - 18, 2014}, pages 9:1--9:10. {ACM},
  2014.

\bibitem{BR-DLT18}
Nicolas Baudru and Pierre-Alain Reynier.
\newblock From two-way transducers to regular function expressions.
\newblock In Mizuho Hoshi and Shinnosuke Seki, editors, {\em 22nd International
  Conference on Developments in Language Theory, {DLT} 2018}, volume 11088 of
  {\em Lecture Notes in Computer Science}, pages 96--108. Springer, 2018.

\bibitem{DBLP:conf/icalp/Bojanczyk14}
Mikolaj Bojanczyk.
\newblock Transducers with origin information.
\newblock In {\em Automata, Languages, and Programming - 41st International
  Colloquium, {ICALP} 2014, Copenhagen, Denmark, July 8-11, 2014, Proceedings,
  Part {II}}, pages 26--37, 2014.
\newblock \href {https://doi.org/10.1007/978-3-662-43951-7\_3}
  {\path{doi:10.1007/978-3-662-43951-7\_3}}.

\bibitem{BDK-lics18}
Mikolaj Bojanczyk, Laure Daviaud, and Shankara~Narayanan Krishna.
\newblock Regular and first-order list functions.
\newblock In {\em Proceedings of the 33rd Annual {ACM/IEEE} Symposium on Logic
  in Computer Science, {LICS} 2018, Oxford, UK, July 09-12, 2018}, pages
  125--134, 2018.
\newblock \href {https://doi.org/10.1145/3209108.3209163}
  {\path{doi:10.1145/3209108.3209163}}.

\bibitem{Buchi60}
J.~Richard B{\"u}chi.
\newblock Weak second-order arithmetic and finite automata.
\newblock {\em Zeitschrift f{\"u}r {M}athematische {L}ogik und {G}rundlagen der
  {M}athematik}, 6:66--92, 1960.

\bibitem{CD15}
Olivier Carton and Luc Dartois.
\newblock Aperiodic two-way transducers and fo-transductions.
\newblock In Stephan Kreutzer, editor, {\em 24th {EACSL} Annual Conference on
  Computer Science Logic, {CSL} 2015, September 7-10, 2015, Berlin, Germany},
  volume~41 of {\em LIPIcs}, pages 160--174. Schloss Dagstuhl - Leibniz-Zentrum
  fuer Informatik, 2015.
\newblock \href {https://doi.org/10.4230/LIPIcs.CSL.2015.160}
  {\path{doi:10.4230/LIPIcs.CSL.2015.160}}.

\bibitem{Cou94}
Bruno Courcelle.
\newblock Monadic second-order definable graph transductions: a survey [see
  {MR}1251992 (94f:68009)].
\newblock {\em Theoret. Comput. Sci.}, 126(1):53--75, 1994.
\newblock Seventeenth Colloquium on Trees in Algebra and Programming (CAAP '92)
  and European Symposium on Programming (ESOP) (Rennes, 1992).
\newblock URL: \url{http://dx.doi.org/10.1016/0304-3975(94)90268-2}, \href
  {https://doi.org/10.1016/0304-3975(94)90268-2}
  {\path{doi:10.1016/0304-3975(94)90268-2}}.

\bibitem{DFJL17}
Luc Dartois, Paulin Fournier, Isma{\"{e}}l Jecker, and Nathan Lhote.
\newblock On reversible transducers.
\newblock In Ioannis Chatzigiannakis, Piotr Indyk, Fabian Kuhn, and Anca
  Muscholl, editors, {\em 44th International Colloquium on Automata, Languages,
  and Programming, {ICALP} 2017, July 10-14, 2017, Warsaw, Poland}, volume~80
  of {\em LIPIcs}, pages 113:1--113:12. Schloss Dagstuhl - Leibniz-Zentrum
  f{\"{u}}r Informatik, 2017.
\newblock \href {https://doi.org/10.4230/LIPIcs.ICALP.2017.113}
  {\path{doi:10.4230/LIPIcs.ICALP.2017.113}}.

\bibitem{DGK-lics18}
Vrunda Dave, Paul Gastin, and Shankara~Narayanan Krishna.
\newblock {Regular Transducer Expressions for Regular Transformations}.
\newblock In Martin Hofmann, Anuj Dawar, and Erich Gr{\"a}del, editors, {\em
  {P}roceedings of the 33rd {A}nnual {ACM\slash IEEE} {S}ymposium on {L}ogic
  {I}n {C}omputer {S}cience ({LICS}'18)}, pages 315--324, Oxford, UK, July
  2018. ACM Press.

\bibitem{DiekertK-CSR12}
Volker Diekert and Manfred Kufleitner.
\newblock Bounded synchronization delay in omega-rational expressions.
\newblock In {\em Computer Science - Theory and Applications - 7th
  International Computer Science Symposium in Russia, {CSR} 2012, Nizhny
  Novgorod, Russia, July 3-7, 2012. Proceedings}, pages 89--98, 2012.
\newblock \href {https://doi.org/10.1007/978-3-642-30642-6\_10}
  {\path{doi:10.1007/978-3-642-30642-6\_10}}.

\bibitem{DiekertK-ToCS2015}
Volker Diekert and Manfred Kufleitner.
\newblock Omega-rational expressions with bounded synchronization delay.
\newblock {\em Theory of Computing Systems}, 56(4):686--696, 2015.

\bibitem{Diekert_2016}
Volker Diekert and Manfred Kufleitner.
\newblock A survey on the local divisor technique.
\newblock {\em Theoretical Computer Science}, 610:13--23, Jan 2016.

\bibitem{EH01}
Joost Engelfriet and Hendrik~Jan Hoogeboom.
\newblock M{SO} definable string transductions and two-way finite-state
  transducers.
\newblock {\em ACM Trans. Comput. Log.}, 2(2):216--254, 2001.
\newblock URL: \url{http://dx.doi.org/10.1145/371316.371512}, \href
  {https://doi.org/10.1145/371316.371512} {\path{doi:10.1145/371316.371512}}.

\bibitem{FKT14}
Emmanuel Filiot, Shankara~Narayanan Krishna, and Ashutosh Trivedi.
\newblock First-order definable string transformations.
\newblock In Venkatesh Raman and S.~P. Suresh, editors, {\em 34th International
  Conference on Foundation of Software Technology and Theoretical Computer
  Science, {FSTTCS} 2014, December 15-17, 2014, New Delhi, India}, volume~29 of
  {\em LIPIcs}, pages 147--159. Schloss Dagstuhl - Leibniz-Zentrum fuer
  Informatik, 2014.
\newblock URL: \url{http://dx.doi.org/10.4230/LIPIcs.FSTTCS.2014.147}, \href
  {https://doi.org/10.4230/LIPIcs.FSTTCS.2014.147}
  {\path{doi:10.4230/LIPIcs.FSTTCS.2014.147}}.

\bibitem{paul-inv}
Paul Gastin.
\newblock Modular descriptions of regular functions.
\newblock In {\em Algebraic Informatics - 8th International Conference, {CAI}
  2019, Ni{\v{s}}, Serbia, June 30 - July 4, 2019, Proceedings}, pages 3--9,
  2019.
\newblock \href {https://doi.org/10.1007/978-3-030-21363-3\_1}
  {\path{doi:10.1007/978-3-030-21363-3\_1}}.

\bibitem{McNaughtonPapert}
Robert McNaughton and Seymour Papert.
\newblock {\em Counter-Free Automata}.
\newblock The MIT Press, Cambridge, Mass., 1971.

\bibitem{DBLP:conf/mfcs/Perrin84}
Dominique Perrin.
\newblock Recent results on automata and infinite words.
\newblock In {\em Mathematical Foundations of Computer Science 1984, Praha,
  Czechoslovakia, September 3-7, 1984, Proceedings}, pages 134--148, 1984.
\newblock \href {https://doi.org/10.1007/BFb0030294}
  {\path{doi:10.1007/BFb0030294}}.

\bibitem{Perrin-Pin-Infinite-words}
Dominique Perrin and Jean-Eric Pin.
\newblock {\em Infinite Words: Automata, Semigroups, Logic and Games}, volume
  141.
\newblock Elsevier, 2004.

\bibitem{Schutzenberger_1965}
Marcel~Paul Sch{\"u}tzenberger.
\newblock On finite monoids having only trivial subgroups.
\newblock {\em Information and Control}, 8(2):190--194, 1965.

\bibitem{Schutzenberger1975d}
Marcel-Paul Sch{\"u}tzenberger.
\newblock Sur certaines op\'erations de fermeture dans les langages rationnels.
\newblock In {\em Symposia Mathematica, Vol. XV (Convegno di Informatica
  Teorica, INDAM, Roma, 1973)}, pages 245--253. Academic Press, 1975.

\end{thebibliography}

\end{document}